\documentclass[a4paper,11pt]{article}

\usepackage{chngcntr}  % Include the chngcntr package

\usepackage[T1]{fontenc}

\usepackage[margin=0.9in]{geometry}
\usepackage{amsmath,amsthm,amssymb,accents}
\usepackage{nicefrac, xfrac}
\usepackage{bbm}
\usepackage{centernot} %negating symbols
\usepackage{yhmath} %forwidecheck
\usepackage{tabularx}

\usepackage{pdflscape} %to change format to landscape

\usepackage{comment}
\usepackage{subcaption}
\usepackage{enumitem}
\usepackage{multirow,array}
\usepackage{mathtools}  
\usepackage{xfrac}
\usepackage{soul} %to highlight text
\usepackage{booktabs} %chatgpt wanted this for a table
\usepackage{tablefootnote} %for table footnotes to compile at the bottom of the page

\usepackage{hyperref}
\hypersetup{
    pdftitle={Akkar JMP}, % Title displayed in the thumbnail
    pdfauthor={D. Carlos Akkar}, % Optional: set the author metadata
    pdfsubject={Akkar JMP 2025}, % Optional: subject
}

\usepackage{setspace}
\usepackage{pgf,tikz,pgfplots}
\usetikzlibrary{patterns}
\pgfplotsset{compat=1.18}
\usepackage{mathrsfs}
\usetikzlibrary{arrows}

\usepackage{needspace}

\usepackage[backend=biber,style=authoryear,sorting=nyt]{biblatex}
\AtEveryBibitem{%
  \iffieldundef{doi}{}{%
    \clearfield{url}%
    \clearfield{eprint}%
    \clearfield{urldate}%
  }%
}

\usepackage{graphicx}
\usepackage{capt-of} % For captions outside float environments
\numberwithin{figure}{section}

\usepackage[disable]{todonotes} %disables todonotes

\newcommand{\readernote}[1]{\todo[size=\scriptsize]{#1}}

\DeclareMathOperator{\Exp}{\mathbb{E}}
\DeclareMathOperator{\Prob}{\mathbb{P}}
\DeclareMathOperator{\argmax}{arg\,max}

\DeclareMathOperator{\indic}{\mathbbm{1}}

\usepackage{cleveref}
\usepackage{thmtools}
\usepackage{thm-restate}
\usepackage{cleveref}

\theoremstyle{definition}

\declaretheorem[name=Proposition]{prop}
\declaretheorem[name=Lemma,sibling=prop]{lem}
\declaretheorem[name=Corollary,sibling=prop]{cor}

\theoremstyle{definition}
\declaretheorem[name=Definition]{defn}

\numberwithin{equation}{section}

\crefname{thm}{Theorem}{Theorems}
\Crefname{thm}{Theorem}{Theorems}
\crefname{lem}{Lemma}{Lemmas}
\Crefname{lem}{Lemma}{Lemmas}
\crefname{prop}{Proposition}{Propositions}
\Crefname{prop}{Proposition}{Propositions}
\crefname{cor}{Corollary}{Corollaries}
\Crefname{cor}{Corollary}{Corollaries}
\crefname{claim}{Claim}{Claims}
\Crefname{claim}{Claim}{Claims}
\crefname{defn}{Definition}{Definitions}
\Crefname{defn}{Definition}{Definitions}
\crefname{cond}{Condition}{Conditions}
\Crefname{cond}{Condition}{Conditions}
\crefname{rem}{Remark}{Remarks}
\Crefname{rem}{Remark}{Remarks}

\newtheorem*{claim*}{Claim}

\newcommand{\commentout}[1]{}
\newcommand{\quality}{\theta}
\newcommand{\typesig}{s}

\newcommand{\strat}{\sigma}
\newcommand{\prior}{\rho}

\newcommand{\signalstr}{\mathcal{E}}
\newcommand{\experiment}{\mathcal{E}}
\newcommand{\signalset}{\mathbf{S}}
\newcommand{\interim}{\psi}
\newcommand{\interimfcn}{\Psi}

\newcommand{\payoff}{\pi}

\newcommand{\sumpayoff}{\Pi}

\usepackage{mathrsfs}

\begin{document}

\vspace{15pt}

\title{The (Mis)use of Information in Decentralised Markets
\\
%\footnote{}
    }
\author{\href{https://www.carlosakkar.com}{D. Carlos Akkar}\footnote{Nuffield College and Department of Economics, Oxford. \href{mailto:akkarcarlos@gmail.com}{akkarcarlos@gmail.com}  \\
I thank Ian Jewitt, Margaret Meyer, Daniel Quigley, Ludvig Sinander, Paula Onuchic, and Péter Es\H{o} for long discussions and generous guidance. I also thank Inés Moreno de Barreda, Costas Cavounidis, Stephan Lauermann, Manos Perdikakis, Daniel Rappoport, Clara Schreiner, Alex Teytelboym, Andy Zapechelnyuk, seminar audiences at the Oxford Student Theory Workshop, NASMES 2024, EEA-ESEM 2024, the 2024 Paris Transatlantic Theory Workshop, and speakers at the Nuffield Economic Theory Seminar for feedback.} 
%\\ \small \textit{Nuffield College, Oxford}
}
\date{\today}
\maketitle
\readernote{COMMENTS ON!}

%\vspace{-5pt}

%\vspace{3pt}

% \begin{center}
%     % \large{Job Market Paper}
%     % \\
%     \href{https://carlosakkar.github.io/docs/jmp.pdf}{\large{\textcolor{black}{Click here for latest version.}}}    
% \end{center}

%\vspace{10pt}

\begin{abstract}
    \normalsize{A seller offers an asset in a decentralised market. 
    %Buyers' common value for the asset determines whether trade is efficient; each have private information about it. 
    Buyers have private signals about their common value. 
    I study whether the market becomes allocatively more efficient with (i) more buyers, (ii) better-informed buyers. Both increase the information available about buyers' common value, but also the adverse selection each buyer faces. With more buyers, trade surplus eventually increases and converges to the full-information upper bound if and only if the likelihood ratios of buyers' signals are unbounded from above. Otherwise, it eventually decreases and converges to the no-information lower bound. With better information about trades buyers would have accepted,  trade surplus increases. With better information about trades they would have rejected, trade surplus decreases---unless adverse selection is irrelevant. For binary signals, a sharper characterisation emerges: stronger good news increase total surplus, but stronger bad news eventually decrease it.}
\end{abstract}

\newpage

\section{Introduction}

\looseness=-1
In this paper, I ask whether more information about buyers' common value for an asset improves or harms allocative efficiency in a decentralised market. 
The setting is parsimonious: the seller sequentially visits $n$ buyers until one accepts to trade at the seller's commonly known reservation value.\footnote{\looseness=-1 In Section \ref{section:ultimatumprices}, I show that this simplifying assumption is without loss in a setting where buyers extend take-it-or-leave-it offers to the seller, and the seller takes an offer unless he expects to secure greater surplus in later visits.} Negotiations are private: no buyer knows how many others the seller visited already. 
Trade is efficient when buyers' common value for the asset---its \textit{quality}---is High, but not when it is Low. 
Each buyer holds a private signal about the asset's quality; conditional on quality, these signals are IID. A buyer accepts trade when she expects it to yield positive surplus; otherwise she rejects. 
I \nolinebreak ask:
\begin{enumerate}[left=0pt, itemsep=-2pt]
    \item Does the expected surplus from trade increase with more buyers, each with an additional signal?
    \item Does the expected surplus from trade increase with better-informed buyers, i.e., each with a more informative signal?
\end{enumerate}

Both more and better-informed buyers increase the amount of information available in the market about the asset's quality. However, more information in the market---through either channel---is a double-edged sword for allocative efficiency. On the one hand, it might push buyers to better trades by helping them screen the asset's quality better. On the other hand, it might push them to worse trades by exposing them to greater adverse selection: when there are more buyers in the market, more might have already rejected the seller; when each buyer is better-informed, each rejection might owe to a worse signal. This paper shows that the \textit{kind} of information in the market determines how this trade-off is resolved.

My first main result, Theorem \ref{thm:extensive}, answers how increasing the number of buyers in the market influences allocative efficiency. This hinges on whether the likelihood ratios of buyers' signals are unbounded from above. If they are, the expected surplus from trade eventually increases in the number of buyers and converges to the full-information benchmark: a High quality seller almost surely trades, but a Low quality seller never does. This is the upper bound for equilibrium surplus in the market: all gains from trade are exhausted. If they are not, the expected surplus from trade eventually decreases in the number of buyers and converges to the no-information benchmark: either the seller almost surely trades regardless of his quality, or the expected surplus is zero when he trades. This is the lower bound for equilibrium surplus in the market (Proposition \ref{prop:selectivebetter}):
buyers' ability to screen the asset's value generates no additional gains from trade.

That the outcome (whether trade occurs) in a large market reveals buyers' common value for the asset if and only if the likelihood ratios of their signals are unbounded from above is reminiscent of a large sealed bid common value auction à la \textcite{wilson1977} and \textcite{milgrom1979}. There, too, the outcome (the winning bid) reveals the asset's quality if and only if the bidders' signal structure satisfies the same condition.\footnote{This is when bidders' common value for the item can assume two values. \textcite{milgrom1979} identifies ``distinguishability'' as a condition that generalises ``likelihood ratios unbounded from above'' when buyers' common value for the item can assume any number of finite values.} However, when this condition is violated, trade in a decentralised market either becomes completely uninformative about the asset's quality, or only reveals that the expected gains from trade are zero. In contrast, the winning bid in a large auction may still approximate buyers' common value well.\footnote{See, for instance, Section IV in \textcite{lauermann_wolinsky_auction}.}

\looseness=-1
This also offers an interesting contrast with \textcite{lauermann_wolinsky_searchadverseselection}. In a model like mine\footnote{See the Related Literature section for a more detailed discussion.} but where the gains from trade are always positive, they find that (generically) the outcome in a large market either fully reveals or is completely independent of the asset's quality. Theorem \ref{thm:extensive} establishes that another possibility emerges when, ex-ante, the expected gains from trade are negative: trade might be \textit{partially} informative about the asset's quality. However, trade only reveals the expected gains not to be negative but zero instead\footnote{Section \ref{section:supplements} illustrates this with a numerical example.}---this information has no bearing on total surplus.

\looseness=-1
In Theorems \ref{thm:intensive_binary} and \ref{thm:intensive}, I answer how giving better information to each existing buyer influences allocative efficiency. Theorem \ref{thm:intensive_binary} shows that, when buyers' signal structure is binary, stronger good news (higher likelihood ratio at the top) always increases surplus; but stronger bad news (lower likelihood ratio at the bottom) eventually decreases it. 
The former might prevent a seller from trading, but recovers surplus in doing so. The latter might help a seller trade, but this eventually destroys surplus due to adverse selection. Theorem \ref{thm:intensive} generalises this insight to arbitrary finite signal structures: additional information where a buyer would have accepted trade (a \textit{negative override}) increases surplus; but additional information where she would have rejected trade (a \textit{positive override}) decreases it---unless \textit{adverse selection is irrelevant} in the appropriate sense (Definition \ref{defn:irrelevance}).

\looseness=-1
To understand the main insight, consider buyers with a binary signal structure: each buyer receives either a good, or a bad signal. For simplicity, ignore equilibrium considerations; simply let buyers accept upon a good signal and reject upon a bad signal. Now, consider revealing additional information to each buyer---another binary signal. This additional information could serve two purposes. If it is revealed after an initial good signal, it could lead the buyer to revise her initial decision to a rejection. I call information that serves this purpose a \textit{negative override}.\footnote{I follow the language used in credit markets: a negative (downgrade) override occurs when a human evaluator revises a prospective borrower's algorithmic credit score downwards, in light of overlooked information. A positive (upgrade) override occurs when she revises it upwards. See Section 2.5 in \textcite{vangestel_baesens_credit_scoring}; as well as par. 110 in \textcite{ecb_internal} and pg. 140 in \textcite{autoabs}.} If it is revealed after an initial bad signal, it could lead her to revise her initial decision to an acceptance. I call such information a \textit{positive override}. In this simple binary-on-binary example, we can interpret a negative override as a strengthening of good news: a buyer can rely on two good signals to accept. A positive override strengthens bad news: a buyer can rely on two bad signals to reject.

A negative override increases surplus. It makes it harder for the seller to trade---a seller some buyer would have accepted before the negative override became available might now be rejected by every buyer. But when this happens, it reveals the expected surplus from trade to be negative: each buyer observed a bad signal and concluded that trading would reduce surplus, despite not knowing (but suspecting) that all buyers reached the same conclusion.  

In contrast, a positive override might decrease surplus. It makes it easier for the seller to trade---a seller every buyer would have rejected might now trade with some buyer. However, due to adverse selection, the expected surplus from such a trade might be negative: the buyer who trades with the seller does not observe how many others rejected him previously. If too many did, those buyers' bad signals might reveal expected surplus from trade to be negative despite her good signal. The buyer might find that she traded when she should not have.

I show that adverse selection severely limits positive overrides from raising surplus: unless \textit{adverse selection is irrelevant}, i.e., a buyer need not care about the number of previous refusals the seller received, a positive override reduces total surplus. 

This insight underpins Theorem \ref{thm:intensive_binary}'s sharp characterisation for binary signals. To extend it to arbitrary signal structures in Theorem \ref{thm:intensive}, I formalise a positive (and, negative) override as a \textit{local mean preserving spread}\footnote{See Definitions \ref{defn:local_mps} and \ref{defn:overrides}.}  of a signal upon which buyers reject (and, accept). Studying informativeness at the level of local spreads is essential to the tractability of my exercise but sacrifices no  generality: any Blackwell improvement is a combination of finitely many local spreads.

\looseness=-1
Theorems \ref{thm:intensive_binary} and \ref{thm:intensive} show that too much information can be detrimental for allocative efficiency. So, finally, I study how a regulator can coarsen buyers' information to maximise expected surplus. In Section \ref{section:regulator}, I 
show that through this policy tool, the regulator aims to prevent a buyer from trading unless adverse selection is irrelevant, i.e., unless she should trade even if everyone else rejected the seller. The implication is striking: the regulator wants buyers to base their decisions on the \textit{highest} number of rejections the seller may have received, not the expected number of rejections.

\subsection*{Contribution}
\label{section:contribution}

I view the main contribution of my paper to be twofold. First, I study a question that has been largely overlooked by the literature on information aggregation in markets. Most of this literature\footnote{Prominent and related work in this literature includes \textcite{wilson1977}, \textcite{milgrom1979}, \textcite{riordan1993} and \textcite{lauermann_wolinsky_auction} for auctions (centralised markets) and \textcite{wolinsky1990}, \textcite{zhu2012}, and \textcite{lauermann_wolinsky_searchadverseselection} for decentralised markets.}  asks whether the outcome in a large market reflects all information its participants have. Instead, I ask whether a \textit{finite} decentralised market can convert \textit{more information} among its participants to more efficient outcomes. I study two channels which increase information in the market. The first is through an additional buyer, bringing an additional signal to the market. This paper is the first to explore this channel in a decentralised market.\footnote{\looseness=-1 \textcite{riordan1993} studies how allocative efficiency in a common value auction changes with an additional bidder.} The second is through better-informed buyers. To the best of my knowledge, this channel has not been explored by previous work.\footnote{Notably, \textcite{glode_opp_2020} show that a decentralised OTC market provides buyers with greater information acquisition incentives than a centralised limit-order market, so might be more efficient than the latter. I discuss their work in light of my contribution under Related Literature.}

Second, my findings have important policy implications for markets where trades are negotiated bilaterally and with little to no trade transparency: such as over-the-counter markets,\footnote{\looseness=-1 OTC markets are characterised by sequential contacts and little transparency (\cite{duffie_dark}; \cite{zhu2012}). Liquidity providers typically make ultimatum offers that last ``as long as the breath is warm'' (\cite{besembinder_maxwell}).} credit markets,\footnote{
    In the US and the UK, credit scores mask borrowers' recent applications, and borrowers exercise little bargaining power against lenders (\cite{agarwal2024searching} and \cite{uk_rate_shopping}).} and the housing market.\footnote{
        In the housing market, 
        ``buyers and sellers must search for each other'' (\cite{han_strange_housing_microstructure}).
        Sellers frequently relist, making it difficult to infer how many viewings resulted in no trade:
        \textcite{remax2024} advises ``if a property has been sitting on the market and going stale, there is no harm in relisting it so that it appears fresh and new''.
    } 
Recent technological advances have allowed participants in these markets to enjoy increasingly greater access to information.\footnote{
    Hedge funds and broker-dealers use increasingly sophisticated data and algorithms to assess trades' profitability (\cite{fsb2017}); lenders use cutting-edge ML technology in credit scoring (\cite{fsb2017}); algorithmic traders in housing markets analyse and execute trades faster than traditional investors (\cite{Raymond_market_algorithms}).}
It is commonly presumed that the ``more efficient processing of information, for example in credit markets, financial markets, [...] contribute to a more efficient financial system'' (\cite{fsb2017}). My Theorems \ref{thm:intensive_binary} and \ref{thm:intensive} show that this presumption---which ignores the adverse selection problem in these markets---is misleading: adverse selection may claw back on market participants' ability to screen for efficient trades, and lead to \textit{lower} surplus when each can access better information.
My Theorem \ref{thm:extensive} shows that adverse selection might cause increased competition to hurt efficiency, too. This validates a concern empirically recognised by regulators and industry leaders.\footnote{Regulators (partially) blamed adverse selection for the collapse of a British bank, HBOS: ``the borrowers who came through its doors inevitably included many whom better established banks had turned away'' (\textcite{kay_ft_advsel}).} 

\looseness=-1
Existing regulation in credit markets already limits the information lenders can use to assess borrowers.\footnote{For instance, following the 2008 crash, the Basel III Accord severely limited the use of ``advanced internal ratings systems'' to determine credit risk exposure. This overturned the conventions set in Basel II. See \textcite{basel3_summary}.} Directly resonating with Theorems \ref{thm:intensive_binary} and \ref{thm:intensive}, ECB guidelines (\citeyear{ecb_internal}) state that ``institutions should be more restrictive with positive overrides than with negative ones''.
I offer a novel justification for such policies, rooted in adverse selection.\footnote{Currently, these policies are mostly justified by a distrust in the ``robustness and prudence'' of lenders' abilities to screen borrowers (\cite{basel3_summary}).} In Section \ref{section:regulator}, I study how a regulator can design restrictions on market participants' information to increase total surplus in the market.

\subsection*{Literature Review}
\label{section:lit}

\looseness=-1
The first question I ask is whether allocative efficiency in a decentralised market increases with more buyers. 
\textcite{riordan1993} asks this question in a first price auction with common values.\footnote{Relatedly, \textcite{ottaviani_sorensen_ditillio_2021} study whether the winning bid in a first price common value auction becomes more informative about the asset's value when there are more bidders. Efficiency is of no direct concern in their setting: trade is always efficient and always materialises.} There, the adverse selection problem is simply the winner's curse---the winner understands that she had the highest signal among all bidders. In contrast, here, a buyer who trades understands that she had the highest signal among those buyers the seller previously visited.
Consequently, the sufficient condition \textcite{riordan1993} identifies\footnote{Where $F_{\quality}(s)$ is the CDF of the signal distribution for quality $\quality \in \{L,H\}$, he finds that the sign of the expression $\frac{{f_H(.)}/{F_H(.)}}{{f_L(.)}/{F_L(.)}} - \frac{F_H(.)}{F_L(.)}$ over the support is a sufficient condition to determine this.} for surplus to be increasing or decreasing in the number of bidders differs from the necessary and sufficient condition Theorem \ref{thm:extensive} supplies for a decentralised market. 

My second question---whether efficiency increases when buyers are \textit{better-informed}---is novel in this literature. The closest papers, to the best of my knowledge, are \textcite{levin_2001} and \textcite{glode_opp_2020}.
\textcite{levin_2001} asks whether in a lemons market à la \textcite{akerlof70}---where the seller's private information is the root cause of inefficiency---a better-informed seller necessarily hurts efficiency. He finds that the right kind of information can increase efficiency. His environment, and therefore his characterisation, differs from mine. \textcite{glode_opp_2020} show that buyers might have greater incentives to acquire information in an OTC market than in a limit-order market; so, the former can be more efficient than the latter. They investigate a particular information technology: buyers invest in their probability of getting a fully revealing signal. My Theorems \ref{thm:intensive_binary} and  \ref{thm:intensive} show that in general, a market with better-informed buyers might be less efficient due to adverse selection.

My model 
is closest to \textcite{zhu2012} and \textcite{lauermann_wolinsky_searchadverseselection}. \textcite{zhu2012} assumes that trade is efficient regardless of the asset's quality;\footnote{Both buyers' common value and the seller's reservation value for a Low quality asset is 0.} otherwise his model is identical to mine. He shows that unless the likelihood ratios of buyers' signals are unbounded from above, a large market might fail to be efficient---a High quality seller might fail to trade. Where trading with a Low quality seller is inefficient, Theorem \ref{thm:extensive} offers a stronger conclusion: unless the same condition holds, surplus in a large market converges to the no-information lower bound. 

\looseness=-1
\textcite{lauermann_wolinsky_searchadverseselection}, too, study a decentralised market for a common value asset; but they focus on a large market where the seller (there, the buyer) (i) pays a small cost for each buyer (there, seller) he visits and (ii) has bargaining power. Importantly, there is no efficiency concern: trade is always efficient and executed. They find that generically, the transaction price carries either full or no information about the asset's common value. Furthermore, the seller's costly search further impedes the revelation of the asset's value: the condition necessary for the transaction price to be fully revealing is stronger than the unboundedness of likelihood ratios. 

\textcite{Chen_lauermann_veto} study a similar model to mine. In their model, an organiser sequentially contacts voters, searching for a veto. Each contact is costly. Voters wish the reform to be vetoed in state $\alpha$, but not in state $\beta$. The organiser knows the state, but the voters do not. \textcite{Chen_lauermann_veto} show that as the number of voters rises, voters' payoffs converge to the full-information benchmark: the reform is nearly always vetoed in state $\alpha$, but passed in state $\beta$. This is achieved through the organiser's incentives to continue sampling in state $\alpha$, but to stop sampling in state $\beta$. 

Besides having an organiser who samples at cost rather than a seller who samples freely, \textcite{Chen_lauermann_veto} differs crucially from this paper in their voters' behaviour. Their voters have common interests: they care about whether \textit{some} voter vetoes the policy. So, their behaviour maximises their aggregate payoffs: each voter conditions on being pivotal---the last voter who may refuse a veto. In my model, a buyer cares only about whether \textit{she} buys the asset. Her behaviour maximises her individual payoff: she conditions on the expected number of buyers who refused the seller before her. This acts as an additional friction which prevents the market from achieving higher efficiency. 

\textcite{Cavounidis_Chunky} also studies a model very similar to mine; but there, ``appraisers'' (replacing the buyers here) must acquire costly information about the state. Rather than focusing on the effect of increasing the amount of information available in the market, \textcite{Cavounidis_Chunky} asks how many appraisers the market can profitably accommodate, given those appraisers' information acquisition costs. He finds that when the cost of information is ``divisible'', the market can accommodate arbitrarily many appraisers; but when it is ``chunky'', the size of the market is constrained by the adverse selection problem this paper also highlights. 

\looseness=-1
My model also admits a social learning interpretation, in the tradition of \textcite{bhw_92}. It can be considered as a variant of the classic model: later decision makers (buyers) are called to decide only if those before them reject, and no one observes her position in the queue (as in \textcite{herrera_horner_2013}). Most work in this literature focuses on whether full learning attains with a large number of decision makers. Instead, my results speak to how more information, through more or better-informed decision makers, influences the welfare of finitely many decision makers. 

\looseness=-4 The remainder is organised as follows. Section \ref{section:model} presents the model. Section \ref{section:equilibria} presents preliminary analyses about equilibria and total surplus in the market. Section \ref{section:morebuyers} presents Theorem \ref{thm:extensive}. Section \ref{section:betterinfo} presents Theorems \ref{thm:intensive_binary} and \ref{thm:intensive} (in Subsections \ref{section:binary} and \ref{section:finite}). Section \ref{section:regulator} discusses how buyers' information can be coarsened to maximise total surplus. Section \ref{section:ultimatumprices} presents an extension where buyers offer take-it-or-leave it prices to the seller. Section \ref{section:supplements} presents a numerical example that supplements the discussion in Section \ref{section:morebuyers}. Section \ref{section:proofs} presents the proofs and results the main text omits.

\section{Model}
\label{section:model}

\looseness=-1
The seller (he) of an indivisible asset sequentially visits $n \in \mathbb{N}$ prospective buyers (she) in a uniformly random order. He sells to the first one who accepts to pay his reservation value $c \in [0,1]$. 
The asset's (seller's) quality $\quality$ is either High or Low, $\quality \in \{H,L\}$. When the buyer he visits rejects trade, no transaction takes place and the seller proceeds to visit the next buyer. If the buyer accepts, she pays the seller his reservation value and receives the asset. She enjoys a return of 1 if the asset's quality is High, but 0 if it is Low. The game ends when a buyer accepts the seller, or they all reject him.

At the outset of the game, the asset's quality is unknown;\footnote{Until Section \ref{section:ultimatumprices}, the seller's knowledge about the asset's quality is immaterial.} all players share the common prior that it is High with probability $\prior$ and Low otherwise. Each buyer obtains additional {private information} about the asset's quality through the outcome of a Blackwell experiment $\experiment = (\signalset, p_L, p_H)$. The outcome $\typesig$ of the experiment---the buyer's \textit{signal}---is drawn from the finite set $\signalset$ with a distribution $p_{\quality}$. Conditional on the asset's quality, buyers' signals are IID. The joint distribution of buyers' signals conditional on the asset's quality is common knowledge. 

The buyer visited by the seller receives no information about how many others the seller previously visited. Nonetheless, she deduces that all those buyers rejected the seller. Through this, she extracts additional information about the asset's quality.

\looseness=-1
The buyer forms her posterior belief about the asset's quality using the information conveyed by the seller's visit and her private signal. 
First, she uses the information conveyed by the seller's visit to revise her prior belief $\prior$ to an interim belief $\interim$. Then, she uses her private signal to revise her interim belief to a posterior belief $\Prob_{\interim} \left( \quality = H \mid \typesig \right)$.

A buyer's strategy $\strat: \signalset \to [0,1]$ maps every signal $\typesig \in \signalset$ she might observe to a probability $\strat(\typesig)$ with which she accepts to trade. Her strategy $\strat$ is \textit{optimal against the interim belief} $\interim$ if, given this interim belief and her signal, the buyer accepts (rejects) to trade whenever her expected payoff from trading with the seller is positive (negative):
\begin{equation*}
    \strat(\typesig) =
    \begin{cases}
        0  & \Prob_{\interim} \left( \quality = H \mid \typesig \right) < c \\
        %\in [0,1] & \Prob_{\interim} \left( \quality = H \mid \typesig \right) = c \\
        1 & \Prob_{\interim} \left( \quality = H \mid \typesig \right) > c
    \end{cases}
\end{equation*}
She may accept to trade with any probability when she expects zero surplus from trading.

I focus on \textit{symmetric Bayesian Nash Equilibria} of this game. Hereafter, I reserve the term \textit{equilibrium} for such equilibria unless I state otherwise. An \textit{equilibrium} is a strategy and interim belief pair $(\strat^*, \interim^*)$ such that:
\begin{enumerate}[itemsep=-2pt]
    \item The interim belief $\interim^*$ is \textit{consistent} with the strategy $\strat^*$; i.e., it is the interim belief of a buyer who believes all other buyers use the strategy $\strat^*$.
    \item The strategy $\strat^*$ is optimal given the interim belief $\interim^*$.
\end{enumerate}
I call any strategy $\strat^*$ that constitutes part of an equilibrium an \textit{equilibrium strategy}.

\section{Buyers' Beliefs, Equilibria, and Total Surplus}
\label{section:equilibria}

This section lays the necessary groundwork to discuss my main results. First, I discuss how buyers form their interim beliefs, and the fundamental properties of the set of equilibria. Then, I discuss how the total surplus from trade varies across different equilibria.

\subsection{Buyers' Beliefs and Equilibria}

No buyer learns how many others the seller visited before her. But, she deduces that all those past visits resulted in rejections. How does she interpret this information?

When each buyer uses a strategy $\strat$, a seller of quality $\quality$ faces a probability $r_{\quality} (\strat; \experiment)$ of getting rejected in any of his visits:
\begin{equation*}
    r_{\quality} (\strat; \experiment) := 1 - \sum\limits_{j=1}^{m} p_{\quality} \left( s_j  \right) \times \strat \left( s_j \right)
\end{equation*}
\noindent 
\looseness=-1
Every buyer understands that the seller is equally likely to decide to visit any number $k \in \left\{ 0,1,2,..., n-1 \right\}$ of other buyers before her. She will receive the seller's visit if and only if he is rejected by all those $k$ buyers.
Therefore, she assigns a probability $\nu_{\quality} \left( \strat; \experiment \right)$ to being visited by the seller:
\begin{equation*}
    \nu_{\quality} \left( \strat; \experiment \right) := \frac{1}{n} \times \sum\limits_{k=0}^{n-1} r_{\quality} \left( \strat; \experiment \right)^{k}
\end{equation*}
When the seller \textit{does} visit her, the buyer uses this information to update her prior belief about the seller's quality to an interim belief $\interim$:
 \begin{align*}
    \interim = \Prob \left( \quality = H \mid \textrm{visit received}
    \right) 
    &= 
    \frac{
        \Prob \left( \textrm{visit received} \mid \quality = H \right) \times \Prob ( \quality = H )
    }{
        \Prob \left( \textrm{visit received} \right)
    }
    \\[3pt]
    &=
    \frac{
        \prior \times \nu_{H} \left( \strat; \experiment \right)
        }{
        \prior \times \nu_{H} \left( \strat; \experiment \right)
        +
        (1 - \prior) \times \nu_{L} \left( \strat; \experiment \right)
        }
\end{align*}
This is the unique interim belief that is \textit{consistent} with every buyer using the strategy $\strat$. 
The buyer then uses her private signal $\typesig \in \signalset$ about the seller's quality to update her interim belief to a posterior belief:
\begin{equation*}
    \Prob_{\interim} \left( \quality = H \mid \typesig \right) = 
    \frac{ \interim \times p_H (s) }{ \interim \times p_H (s) + ( 1 - \interim ) \times p_L(s) }
\end{equation*}
Note that the informational content of the buyer's signal $\typesig \in \signalset$ is distilled by the ratio $\frac{p_H(s)}{p_H(s) + p_L(s)}$. For notational convenience, I will use the signal's label, $\typesig$, to refer to this ratio:
\begin{equation*}
    \typesig:= \frac{p_H(\typesig)}{p_H(\typesig) + p_L(\typesig)} \in [0,1] \qquad \textrm{for all } \typesig \in \signalset
\end{equation*}
Under this notation, the ratio $\frac{s}{1-s}$ simply equals the signal's likelihood ratio, $\frac{p_H(s)}{p_L(s)}$. For further convenience, I also enumerate the signals $\signalset$ in order of increasing likelihood ratios; $\signalset:= \left\{ s_1, s_2, ..., s_m \right\}$ where $s_1 \leq s_2 \leq ... \leq s_m$. Note that, for the same interim belief, a buyer's posterior belief is increasing in her signal's index.  

\looseness=-1
Whenever a strategy $\strat^*$ is optimal against the unique interim belief $\interim^*$ consistent with it, the pair $\left( \strat^*, \interim^* \right)$ forms an equilibrium. In principle, there might be many such pairs, or none at all. Proposition \ref{prop:eqmexist} sets the ground by ruling the latter possibility out and characterising the set of equilibria.

\begin{restatable}{prop}{propeqmexist}
    \label{prop:eqmexist}
    Let $\Sigma$ be the set of equilibrium strategies. Then:
    \begin{enumerate}
        \item $\Sigma$ is non-empty and compact.
        \item Any equilibrium strategy $\strat^*$ is monotone: for any $\strat^* \in \Sigma$, $\strat^*(s) > 0$ for some $s \in \signalset$ implies that $\strat^* (s') = 1$ for every $ s' \in \signalset'$ such that $s' > s$. 
        \item All equilibria exhibit adverse selection:  $\interim^* \leq \prior$ for any interim belief $\interim^*$ consistent with an equilibrium strategy $\strat^* \in \Sigma$.
    \end{enumerate}
\end{restatable}

\begin{proof}[Proof outline:]
    To establish the existence of an equilibrium, I construct a best response correspondence: $\Phi$ for buyers. $\Phi$ maps any strategy $\strat$ to the set of strategies that are optimal against the unique interim belief consistent with $\strat$; i.e. those that maximise a buyer's expected payoff when her all her peers use the strategy $\strat$. Note that $\strat^*$ is an equilibrium strategy if and only if it is a fixed point of this best response correspondence; i.e., $\strat^* \in \Phi (\strat^*)$. Through a routine application of Kakutani's Fixed Point Theorem, I show that $\Phi$ indeed has a fixed point. In the process, I prove that $\Phi$ is upper semicontinuous; this also establishes that the set of equilibrium strategies is compact.

    Monotonicity is a straightforward necessity for a strategy to be optimal: higher signals induce higher posterior beliefs, so buyers (weakly) prefer to trade upon higher signals. A crucial consequence of monotonicity is that a Low quality seller is likelier to be rejected in any of his visits, as buyers are likelier to observe lower signals for him. Thus, the seller is adversely selected through his past rejections, and buyers' interim beliefs always lie below their prior beliefs. 
\end{proof}

An equilibrium is guaranteed to exist, but it need not be unique. The following example, which I will modify and revisit on occasion, illustrates this. Let there be two buyers who share the prior belief $\prior = 0.5$, and the seller's reservation value be $c = 0.2$. Furthermore, let buyers' experiment $\experiment$ be binary, $\signalset = \left\{ 0.2, 0.8  \right\}$, and its outcome have the conditional distribution:
\begin{align*}
    p_L(\typesig) &=
    \begin{cases}
       0.8 & \typesig = 0.2 \\
       0.2 & \typesig = 0.8
    \end{cases}
    &
    p_H(\typesig) &=
    \begin{cases}
        0.2 & \typesig = 0.2 \\
        0.8 & \typesig = 0.8
    \end{cases}
\end{align*}

There are two equilibrium strategies in this example,\footnote{There are no other equilibria. Under any monotone strategy that assigns a positive probability to trade, buyers' interim belief lies between $0.5$ and $0.4$---so buyers must always trade upon the high signal $0.8$. Lemma \ref{lem:binary_nolowmixing} in Section \ref{section:omitted_proofs} shows that in equilibrium, buyers either always or never trade upon the low signal.} which I denote as $\hat{\strat}$ and $\check{\strat}$. Table \ref{tab:eg_comparing_eqa} summarises these strategies, and how buyers form their interim and posterior beliefs under them.

\begin{table}[h]
    \centering
    \setlength{\tabcolsep}{12pt} % Adjusts horizontal padding
    \renewcommand{\arraystretch}{1.5} % Adjusts vertical spacing
    \begin{tabular}{|l|c|c|}
        \hline
        & {$\mathbf{\check{\strat}}$} & {$\mathbf{\hat{\strat}}$} \\ \hline
         \textbf{Prob. the buyer accepts upon} {$\mathbf{s = 0.8}$}
         & $\check{\strat}(0.8) = 1$ & $\hat{\strat}(0.8) = 1$  \\ \hline
         \textbf{Prob. the buyer accepts upon} $\mathbf{s = 0.2}$
         & $\check{\strat}(0.2) = 1$ & $\hat{\strat}(0.2) = 0$ \\ \hline
         \textbf{{Buyers' interim belief}} & $\interim = 0.5$ & $\interim = 0.4$ \tablefootnote{The interim belief in this case is easily calculated as: $\interim = \frac{1 + r_H (\strat; \experiment)}{(1 + r_H (\strat; \experiment)) + (1 + r_L (\strat; \experiment))} = \frac{1.2}{1.2 + 1.8} = 0.4$.} \\ \hline
         {\textbf{Buyers' posterior belief upon} $\mathbf{s = 0.8}$}
         & $0.8$ & $0.7$ \\ \hline
         {\textbf{Buyers' posterior belief upon} $\mathbf{s = 0.2}$}
         & 0.2 & $\approx 0.14$ \\ \hline
    \end{tabular}
    \caption{Running Example: Comparing Equilibrium Strategies $\check{\strat}$ and $\hat{\strat}$}
    \label{tab:eg_comparing_eqa}
\end{table}

\looseness=-1
Under the strategy $\check{\strat}$, a buyer accepts trade regardless of her signal. This eliminates adverse selection: since the first buyer the seller visits accepts trade, whoever he visits is certain that she is the first one he visited; and so, buyers' interim belief $\interim$ equals their prior $\prior$. In this equilibrium, a buyer finds trade optimal even if she receives the low signal $0.2$; given the posterior belief this signal induces, she expects zero net surplus from trade.

\looseness=-1
Under the strategy $\hat{\strat}$, a buyer accepts trade if only if she receives the high signal $0.8$. Buyers' selectivity triggers adverse selection: each buyer understands that she need not be the first one the seller visited. So, buyers' interim belief $\interim$ falls below their prior belief.
A buyer no longer finds trade optimal when she receives the low signal $0.2$: this signal induces a posterior belief $\approx 0.14$, so she expects a loss. She does, however, find it optimal when she receives high signal $s = 0.8$: this signal induces a posterior belief of $0.7$, so she expects positive net surplus from trade.

In the equilibrium where buyers use the strategy $\hat{\strat}$, the seller is likelier to be rejected in any of his visits. The buyers are \textit{more selective}---they are (weakly) likelier to reject trade at any signal they might observe. 

\begin{defn}
    Where $\strat'$ and $\strat$ are two strategies, $\strat'$ is \textit{more selective than} $\strat$ (or, $\strat$ is \textit{less selective than} $\strat'$) if $\strat' ( s ) \leq \strat (s)$ for all $s \in \signalset$.
\end{defn}

Selectivity offers a natural way to order buyers' equilibrium strategies. Proposition \ref{prop:selectivityorder} shows that it is also a complete order over them.

\begin{restatable}{prop}{propselectivityorder}
    \label{prop:selectivityorder}
    Selectivity is a complete order over the set of equilibrium strategies $\Sigma$. 
    Moreover, $\Sigma$ contains a \textit{most} and \textit{least }selective strategy, $\hat{\strat} \in \Sigma$ and  $\check{\strat} \in \Sigma$ respectively:
    \vspace{-0.4cm}
    \begin{equation*}
        \hat{\strat} (s) \leq \strat^* (s) \leq
        \check{\strat} (s)
        \qquad
        \textrm{for all } s \in \signalset \textrm{ and } \strat^* \in \Sigma
    \end{equation*}
\end{restatable}

\begin{proof}
    By Proposition \ref{prop:eqmexist}, the set of equilibrium strategies $\Sigma$ is a subset of the set of monotone strategies. The latter is a chain under the selectivity order: for any signal $\typesig \in \signalset$ and two monotone strategies $\strat$ and $\strat'$, we have:
    \begin{equation*}
        \strat'(s) > \strat(s) \implies 
        \begin{aligned}
        &1 = \strat'(s^{\cdot}) \geq \strat(s^{\cdot})
        \qquad
        \textrm{for any } s^{\cdot} > s \in \signalset
        \\
        &\strat'(s_{\cdot}) \geq \strat(s_{\cdot}) = 0
        \qquad
        \textrm{for any } s_{\cdot} < s \in \signalset
        \end{aligned}
    \end{equation*}
    Since any subset of a chain is also a chain, $\Sigma$ is a chain too.

    By Proposition \ref{prop:eqmexist}, $\Sigma$ is a compact set. Since it is also a chain, by applying a suitably general Extreme Value Theorem (see Theorem 27.4 in \textcite{munkres_topology}) 
    to the identity mapping on $\Sigma$, one verifies that $\Sigma$ has a minimum and maximum element with respect to this order; i.e. there are two strategies $\hat{\strat}, \check{\strat} \in \Sigma$ such that for any other strategy $\strat^* \in \Sigma$ we have $\hat{\strat} (s) \leq \strat^* (s) \leq
        \check{\strat} (s)$ for all $s \in \signalset$.
\end{proof}

\noindent I call an equilibrium ``more selective than another'' whenever buyers use a more selective strategy in the former. 

\subsection{Total Surplus}

Trading generates a surplus of $1-c$ when the asset's quality is High, but destroys a surplus of $c$ when the asset's quality is Low. 
So, the expected surplus from trade in the market---total surplus, for short---depends on how well buyers can screen the asset's quality before they decide whether to trade. Given buyers' experiment $\experiment$ and strategy $\strat$, total surplus equals:
\begin{align*}
    \sumpayoff (\strat; \experiment) &:= (1-c) \times \Prob_{\strat; \experiment} \left( \quality = H \cap \textrm{some buyer trades} \right) - c \times \Prob_{\strat; \experiment} \left( \quality = L \cap \textrm{some buyer trades} \right)
    \\
    &= (1-c) \times \prior \times \left[ 1 - r_H (\strat; \experiment)^n \right] - c \times (1 - \prior) \times \left[ 1 - r_L (\strat; \experiment)^n \right]
\end{align*}
Buyers fully appropriate this surplus---the seller is only paid his reservation value when he trades. 

Two benchmarks are natural to consider. If buyers had full information about the asset's quality, they would trade whenever the asset's quality is High, but never when it is Low; all gains from trade would be realised. Total surplus in this \textit{full-information benchmark} would equal $\sumpayoff^f:= \prior \times (1-c)$. 

If instead, buyers had no information about the asset's quality, their decisions would be guided solely by their prior belief. There would be no adverse selection: previous rejections would convey no private information since no buyer has any. A buyer would trade if, given her prior belief, she expected positive surplus from doing so, and reject otherwise. Total surplus in this \textit{no-information benchmark} would equal $\sumpayoff^n:= \max \left\{ 0, \prior - c \right\}$.

Proposition \ref{prop:selectivebetter} establishes that, in equilibrium, total surplus is bounded by these benchmarks; moreover, total surplus is always higher in more selective equilibria.

\begin{restatable}{prop}{propselectivebetter}
    \label{prop:selectivebetter}
    \looseness=-1
    Equilibrium total surplus is bounded above by the full-information benchmark $\sumpayoff^f$ and below by the no-information benchmark $\sumpayoff^{\emptyset}$. 
    Furthermore, it is higher under more selective equilibrium strategies:
    \vspace{-0.3cm}
    \begin{equation*}
        \max \left\{ 0, \prior - c \right\} = \sumpayoff^{\emptyset} \leq \sumpayoff ( \strat^{**}; \experiment ) \leq \sumpayoff ( \strat^*; \experiment ) \leq \sumpayoff^f = \prior \times (1-c)
        \vspace{-0.3cm}
    \end{equation*}
    where $\strat^*$ and $\strat^{**}$ are two equilibrium strategies such that $\strat^{**}$ is more selective than $\strat^*$.
\end{restatable}

\begin{proof}
    The second part of the Proposition follows as a corollary to Lemma \ref{lem:selectivebetter} in Section \ref{section:omitted_results}, where I prove that total surplus decreases whenever buyers jointly deviate from an equilibrium strategy to a less selective monotone strategy.
    
    For the first part of Proposition \ref{prop:selectivebetter}, note that all gains from trade is realised when buyers have full information; hence $\sumpayoff^f$ bounds total surplus from above. 
    Since total surplus equals buyers' surplus from trade, it is bounded below by $0$ in any equilibrium---a buyer can always reject. Thus, when $\prior \leq c$, $\sumpayoff^{\emptyset} = 0$ bounds total surplus from below. Now let $\prior > c$, and assume for contradiction that there is an equilibrium strategy $\strat^*$ such that $\sumpayoff ( \strat^*; \experiment ) < \sumpayoff^{\emptyset}$. Then:
    \begin{align*}
        \Prob_{\strat^*; \experiment} \left( \textrm{some buyer trades} \right) \times &\left[ \Prob_{\strat^*; \experiment} \left( \quality = H \mid \textrm{some buyer trades} \right) - c \right] < \\
        \Prob_{\strat^*; \experiment} \left( \textrm{some buyer trades} \right) \times &\left[ \Prob_{\strat^*; \experiment} \left( \quality = H \mid \textrm{some buyer trades} \right) - c \right] +
        \\
        &\Prob_{\strat^*; \experiment} \left( \textrm{no buyer trades} \right) \times \left[ \Prob_{\strat^*; \experiment} \left( \quality = H \mid \textrm{no buyer trades} \right) - c \right] = \sumpayoff^{\emptyset}
    \end{align*}
    So, $\Prob_{\strat^*; \experiment} \left( \quality = H \mid \textrm{no buyer trades} \right) > c$. However, $\strat^*$ then cannot be an equilibrium strategy; each buyer has a profitable deviation to trade with the seller whenever he visits. 
\end{proof}

\looseness=-1
Buyers' private information helps them raise surplus by avoiding trade when the asset's quality is Low, and executing it when it is High. Equilibrium multiplicity presents a trade-off: in a more selective equilibrium, trade is less likely---this conserves surplus when the asset's quality is Low. In a less selective equilibrium, trade is more likely---this raises surplus when the asset's quality is High. Proposition \ref{prop:selectivebetter} establishes that this trade-off is always resolved in favour of more selective equilibria.

The intuition is illustrated by our running example. There, we identified two equilibrium strategies, $\hat{\strat}$ and $\check{\strat}$. These are the only equilibrium strategies; so, $\hat{\strat}$ is the most selective equilibrium strategy in that example, and $\check{\strat}$ is the least selective. Table \ref{tab:eg_surplus} summarises how the probability that the seller trades and the total surplus vary across these equilibria, and compares them to the full- and no-information benchmarks.

\begin{table}[h]
    \centering
    \setlength{\tabcolsep}{8pt} % Adjusts horizontal padding
    \renewcommand{\arraystretch}{1.5} % Adjusts vertical spacing
    \begin{tabular}{|l|c|c|c|c|}
        \hline
         & \textbf{no-info.}& $\mathbf{\check{\strat}}$ & $\mathbf{\hat{\strat}}$ & \textbf{full-info.}\\
         \hline
        \textbf{Prob. seller trades when} $\mathbf{\quality = H}$ & 1 & 1 & 0.96 \tablefootnote{This is the probability that at least one buyer will receive the high signal $0.8$ when the seller has High quality: $1 - p_H^2(0.2) = 0.96$. For a seller of Low quality, this probability is: $1 - p_L^2(0.2) = 0.36$.} & 1 \\
        \hline
        \textbf{Prob. seller trades when} $\mathbf{\quality = L}$ & 1& 1 & 0.36 & 0\\
        \hline
        \textbf{Total surplus} & 0.3 & 0.3 & 0.348 & 0.4 \\
        \hline
    \end{tabular}
    \caption{Running Example: Comparing Surplus Across Equilibria}
    \label{tab:eg_surplus}
\end{table}

Unless both buyers observe the low signal $0.2$, the seller trades under both equilibria. If both observe the low signal, both buyers reject the seller in the most selective equilibrium: each fears that the other buyer might have already rejected him, which would indicate that that buyer also observed a low signal. The fact that neither buyer wants to trade in the event that both observe a low signal reveals that the expected surplus from trading with such a seller is negative. Nonetheless, this seller trades in the least selective equilibrium: now, each buyer understands that the seller visits her only if he has not visited the other yet, so she becomes more optimistic about the other buyer's signal. This results in lower total surplus.

\section{Efficiency with More Buyers}
\label{section:morebuyers}

In this section, I discuss how total surplus changes as the number of buyers in the market increases. %

\begin{restatable}{thm}{thmextensive}
    \label{thm:extensive}
    \looseness=-1
    Let $\sumpayoff^n \left( \hat{\strat}; \experiment \right)$ be total surplus under the most selective equilibrium in a market with $n$ buyers. If $\experiment$ has an outcome that fully reveals High quality ($s_m = 1$), the sequence $\left\{ \sumpayoff^n \left( \hat{\strat}; \experiment \right)  \right\}_{n=1}^{\infty}$ is eventually increasing and converges to surplus in the full-information benchmark. Otherwise, it is eventually decreasing and converges to surplus in the no-information benchmark. 
\end{restatable}

Each additional buyer brings an additional signal about the asset's quality to the market. 
As the market becomes arbitrarily large, buyers' collective information becomes sufficient to fully reveal the asset's quality---unless buyers' signals carry no information.
Whether the market outcome incorporates this information and reaches full allocative efficiency depends on the \textit{kind} of information each buyer has. 

If buyers have a signal which fully reveals High quality, total surplus reaches its full-information upper bound as the number of buyers grows arbitrarily large. Beyond a threshold number of buyers, adverse selection forces each buyer to reject the seller unless she observes that fully-revealing high signal. Therefore, a Low quality seller never trades. A High quality seller, however, might. Moreover, a High quality seller is likelier to trade in a larger market, as it is likelier that at least one buyer will observe the fully-revealing high signal. 

If, however, buyers do not have such a signal, the market experiences a surplus breakdown as the number of buyers grows large---total surplus reaches its no-information lower bound. How market outcomes evolve as the number of buyers grows depends on whether the expected gains from trade are positive, $\prior > c$, or weakly negative, $\prior \leq c$:
\begin{itemize}
    \item When the expected gains from trade are positive, no matter how large the market, a buyer who observes the highest possible signal trades when the seller visits. The larger the market, the likelier that some buyer will observe this signal---regardless of the seller's quality. Beyond a threshold number of buyers, this hurts total surplus through the increased probability that a Low quality seller trades. As the number of buyers grows arbitrarily large, the seller almost surely trades. This yields the level of surplus in the no-information benchmark, $\prior - c$. 
    \item When the expected gains from trade are negative, a buyer accepts the seller with an arbitrarily small probability in a large market, if at all. Adverse selection, however, ensures that she expects zero surplus from doing so. The seller may trade with positive probability, but total surplus is equal to that in the no-information benchmark: $0$. 
\end{itemize}

Notably, even when buyers have a signal that fully reveals High quality, total surplus need not converge to the full-information benchmark in \textit{every} equilibrium. To see this, consider a slightly modified version of our running example. As before, the common prior is $\prior = 0.5$, and the seller's reservation value is $c = 0.2$. But we modify buyers' experiment; they now observe the outcome of $\experiment^g = \left( \signalset^g, p_L^g, p_H^g \right)$:
\begin{align*}
    p_L^{g} (s) &= 
    \begin{cases}
        1 & s = 0.2 \\
        0 & s = 1
    \end{cases}
    &
    p_H^{g} (s) &= 
    \begin{cases}
        0.25 &  s = 0.2 \\
        0.75 &  s = 1
    \end{cases}
\end{align*}
The signal $s = 1$ fully reveals High quality. However, buyers trade regardless of the signal they receive in the least selective equilibrium---irrespective of the number of buyers in the market. Thus, the seller always trades. Total surplus always equals that in the no-information benchmark, $\sumpayoff^{\emptyset} = 0.3$. On the other hand, buyers only trade upon the high signal $s = 1$ in the most selective equilibrium. Hence, a Low quality seller never trades. A High quality seller trades with probability $1 - (0.25)^n$ in a market with $n$ buyers. In an arbitrarily large market, he trades almost surely; total surplus converges to the full-information benchmark, $\sumpayoff^f = 0.5 \times [1 - 0.2]$.

Theorem \ref{thm:extensive} helps compare a decentralised market with a first price auction where the seller simultaneously solicits buyers' bids and sells to the highest bidder. There too, total surplus may decrease with an additional bidder (\textcite{riordan1993}), due to an increased probability of trade when it is inefficient. Similar to a decentralised market, adverse selection is the culprit: the winner's curse is more severe in an auction with more bidders. However, adverse selection in a decentralised market differs from the winner's curse in an auction: a buyer who trades in a decentralised market understands that her signal was the highest among the \textit{previous} buyers to be visited; a bidder who wins in an auction understands that her signal was the highest among \textit{all} bidders. Thus, \textcite{riordan1993}'s sufficient condition for an additional participant to decrease total surplus differs from the necessary and sufficient condition Theorem \ref{thm:extensive} recovers for total surplus to eventually increase with the number of buyers in a decentralised market. 

This condition is also necessary and sufficient for buyers' common value to be revealed through trade in a decentralised market. Strikingly, the same condition is also necessary and sufficient for a first price common value auction to reveal bidders' common value through the winning bid.\footnote{See \textcite{wilson1977} and \textcite{milgrom1979} for this classic result: in a setting where the item's common value may take countably many values, \textcite{milgrom1979} recovers ``distinguishability'' as a necessary and sufficient condition. When the item's value is binary, ``distinguishability'' corresponds to an unbounded likelihood ratio from above.} 
However, when this condition is violated, a common value auction with an arbitrarily large number of bidders may nonetheless aggregate information ``well''.\footnote{See, for instance, Section IV in \textcite{lauermann_wolinsky_auction}.} This is not the case in a decentralised market with an arbitrarily large number of buyers. Instead, three possibilities emerge: 

\needspace{2\baselineskip} % Checks for enough space for about 2 lines
\begin{enumerate}
    \item The expected gains from trade are positive, $\prior > c$, and the seller almost surely trades. The fact that she trades is completely uninformative
    about the asset's quality.
    \item The expected gains from trade are negative, $\prior < c$, and the seller never trades\footnote{For instance, if buyers' experiment is uninformative about the asset's value.}. The fact that she does not trade is completely uninformative about the asset's quality. 
    \item The expected gains from trade are negative, $\prior < c$, and if a buyer trades, she expects zero surplus from doing so. Trade---when it happens---reveals that the asset has High quality with probability $c$, i.e., that the expected surplus from trade is zero\footnote{Section \ref{section:supplements} illustrates this scenario with a numerical example.}.
\end{enumerate}

So, whenever the likelihood ratio of buyers' signals is \textit{not} unbounded at the top, trade is at most partially informative about the asset's quality, unlike in an auction. The information it reveals has no bearing on market participants' surplus: at most, trade is revealed to be to be no worse than no trade in expectation. 

This offers an interesting contrast with \cite{lauermann_wolinsky_searchadverseselection}, too\footnote{\textcite{lauermann_wolinsky_auction} find a similar result in a first price auction where the number of bidders varies with the common value of the auctioned item.}. They assume that trade is always efficient, and find that (generically) the outcome in a large market either fully reveals or is completely uninformative about buyers' common value for the asset. Theorem \ref{thm:extensive} echoes their finding where the expected gains from trade are positive, $\prior \geq c$. However, it also shows that another possibility arises when the expected gains from trade are negative, $\prior < c$: trade in a large market may be partially informative about the asset's quality.

\section{Efficiency with Better-Informed Buyers}
\label{section:betterinfo}

\looseness=-1
In this section, I discuss how giving each buyer better information---a Blackwell more informative experiment---affects total surplus. Throughout, I take the number of buyers $n$ to be a primitive rather than a parameter.
Equilibrium surplus attains its extremes in the most and least selective equilibria. My main results describe the comparative statics of total surplus under both of these equilibria. For brevity, I write \textit{equilibrium$^\star$} wherever the reader may read \textit{the most} or \textit{least selective equilibrium}. I denote the equilibrium$^{\star}$ strategy under an experiment $\experiment$ as $\strat^{\star}_{\experiment}$.

If there were a single buyer in the market---exposed to no adverse selection---a Blackwell improvement of her experiment would be necessary and sufficient for total surplus to rise regardless of the seller's reservation value.\footnote{In general, a Blackwell improvement is sufficient, but not necessary for a decision maker to extract higher value from a decision problem (see \textcite{blackwell53}). However, it is necessary for the decision maker to extract higher value from \textit{any} decision problem where the unknown state is binary---such our buyers' screening problem. I present a self contained proof of this fact in Section \ref{section:omitted_results}, Lemma \ref{lemma:blackwell_necessity} for completeness. } The reason is simple: better information improves her ability to screen the asset's quality and target efficient trades.

In a market with multiple buyers, buyers' ability to screen the seller is shaped both by the quality of their private signals and the extent of adverse selection each face in the market. So, better information becomes a double-edged sword. On the one hand, it allows each buyer to screen the seller more effectively. On the other hand, it might exacerbate adverse selection: previous rejections might become likelier, and each might carry worse news about the quality of the asset. This latter channel pushes the buyer to worse trades. If it overwhelms buyers' increased ability to screen the asset's quality, it might lead to lower total surplus. How this trade-off is resolved is determined by the \textit{kind} of improvement in buyers' information. 

\subsection{Binary Experiments}
\label{section:binary}

I start by restricting buyers to binary signals. In this setting, I obtain a sharper characterisation and uncover the main insights which drive my result for general signal structures, in Theorem \ref{thm:intensive}.

A binary experiment $\experiment$ has two possible outcomes $s_1, s_2 \in \signalset$, which I relabel as $s_L, s_H \in \signalset$ for convenience.\footnote{Recall from Section \ref{section:equilibria} that, to ease notation, I use labelling convention $s := \frac{p_H(s)}{p_H(s) + p_L(s)}$.}
The low outcome $s_L \in [0, 0.5]$ decreases a buyer's interim belief about the quality of the asset, while a high outcome $s_H \in [0.5, 1]$ increases it. 

Ranking two binary experiments $\experiment'$ and $\experiment$ in their (Blackwell) informativeness is a simple exercise. Where the former has the possible outcomes $s_L', s_H' \in \signalset'$, the experiment $\experiment'$ is Blackwell more informative than $\experiment$ if and only if:\footnote{See Section 12.5 in \textcite{blackwell_girshick} for a textbook exposition of this classic result.}
\begin{itemize}[itemsep=-2pt]
    \item it has a lower likelihood ratio at the bottom, $s_L' \leq s_L$; i.e., delivers \textit{stronger bad news}, and
    \item it has a greater likelihood ratio at the top, $s_H' \geq s_H$; i.e. delivers \textit{stronger good news}.
\end{itemize}
I illustrate these two improvements in information in Figures \ref{fig:strongergood} and \ref{fig:strongerbad}.

\looseness=-1
Theorem \ref{thm:intensive_binary} (illustrated in Figure \ref{fig:thmbinary}) answers how total surplus evolves when buyers' binary experiment becomes Blackwell more informative.

\begin{restatable}{thm}{thmintensivebinary}
    \label{thm:intensive_binary}
    Let buyers' experiment $\experiment$ be binary. Then, equilibrium$^{\star}$ total surplus is increasing in the strength of good news ($s_H$) but is quasiconcave and eventually decreasing in the strength of bad news ($s_L$).
\end{restatable}

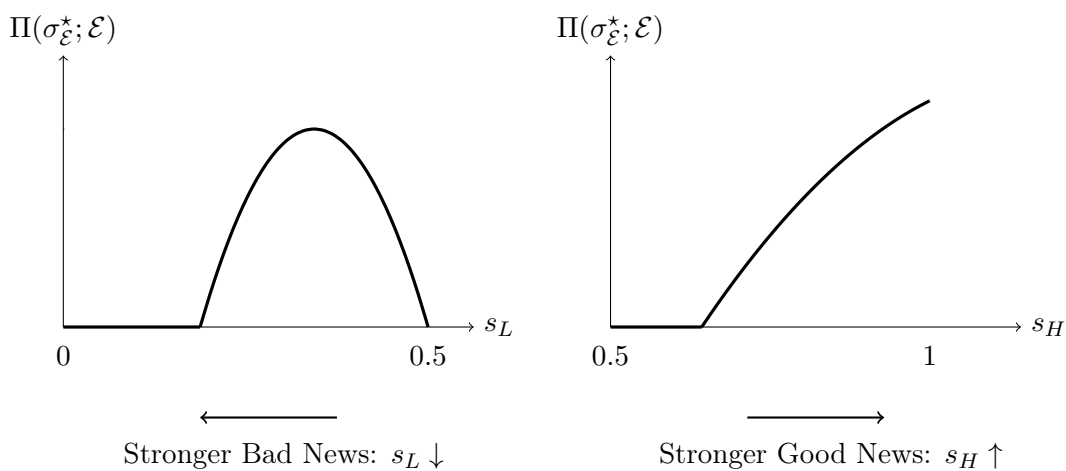
\begin{figure}[h]
        \centering
        \begin{tikzpicture}[scale = 0.6]
            \draw[->] (0,0) -- (9,0) node[right] {$s_L$};
            \draw[->] (0,0) -- (0,6) node[above] {$ \sumpayoff ( \strat^{\star}_{\experiment}; \experiment ) $};
            \draw[scale=1, domain=3:8, smooth, variable=\x, very thick] plot ({\x}, {(0.7)*(3 - \x)*(\x - 8) });
            \coordinate (crd_end) at (8,0);
            \coordinate (crd_top) at (5.5,0);
            % \node at (crd_end) [below = 1mm of crd_end] {$1$};
            % \node[red] at (crd_top) [below = 1mm of crd_top] {$F^*(S_X)$};
            \draw[scale=1, domain=0:3, smooth, variable=\x, very thick] plot ({\x}, {0});
            % \draw[red, dashed] (5.5,0) -- (5.5,4.375);
            %
            % Arrow for left plot, from (6, -2) to (3, -2)
            \draw[->, thick, black] (6,-2) -- (3,-2) node[below, pos=0.5, xshift=0.2cm, yshift = -0.2cm] {Stronger Bad News: $s_L \downarrow$};
            \draw[->] (12,0) -- (21,0) node[right] {$ s_H $};
            \draw[->] (12,0) -- (12,6) node[above] {$\sumpayoff ( \strat^{\star}_{\experiment}; \experiment )$};;
            \node[below] at (12,-0.2) {$0.5$};
            \node[below] at (19,-0.2) {$1$};
            \coordinate (xs_start) at (12,0) ;
            % \node at (xs_start) [below = 1mm of xs_start] {$1$};
            \draw[scale=1, domain=14:19, smooth, variable=\x, very thick] plot ({\x}, { (0.1)*(\x-14)*(29 - \x) });
            \draw[scale=1, domain=12:14, smooth, variable=\x, very thick] plot ({\x}, { 0 });
            \coordinate (brace1) at (1.4,-2.2);
            \coordinate (brace11) at (1.4,-3);
            \coordinate (brace2) at (5.5,-2.2);
            \coordinate (brace21) at (5.5,-3);
            %                 
            % Label the x-coordinate
            \node[below] at (0,-0.2) {$0$};
            \node[below] at (8,-0.2) {$0.5$};
            %
            % Draw the dashed line to the y-axis
            \draw[dashed] (0,4.375) -- (0,4.375);
            %
            % Label the y-coordinate
            \draw[->, thick, black] (15,-2) -- (18,-2) node[below, pos=0.5, xshift=0.2cm, yshift = -0.2cm] {Stronger Good News: $s_H \uparrow$};
        \end{tikzpicture}  
        \caption{Theorem \ref{thm:intensive_binary} illustrated}
        \label{fig:thmbinary}
    \end{figure}

To understand the intuition behind Theorem \ref{thm:intensive_binary}, let us start from the case of stronger good news. 
Instead of the experiment $\experiment$, buyers now observe the outcome of $\experiment' = (\signalset', p_L', p_H')$, which delivers stronger good news than $\experiment$, $s_H' > s_H$, but the same strength of bad news as $\experiment$, $s_L = s_L'$. 
Brushing equilibrium considerations aside, simply assume that, under both experiments, a buyer accepts upon the high signal, $s_H$ or $s_H'$, and rejects upon the low signal, $s_L$ or $s_L'$. How does, then, this improvement in buyers' information affect the seller's chances of trading?

The answer is clearest when we reinterpret the additional information a buyer can extract from $\experiment'$ as an additional signal she might observe. 
Instead of replacing the original experiment $\experiment$ with $\experiment'$, imagine that a buyer who observes an initial high signal $s_H$ from the experiment $\experiment$ then observes the outcome of an \textit{additional} binary experiment $\experiment^a = \left( \{ s_L^a, s_H^a \}, p_L^a, p_H^a \right)$. 
Conditional on the asset's quality, the outcomes of these additional signals are IID across buyers, and independent of the first signal they observe. Their conditional distributions are such that when appended to $\experiment$ as such, the experiment $\experiment^a$ mimics the improvement in information that $\experiment'$ offers:
\begin{align*}
    \frac{s_H}{1 - s_H} \times \frac{p_H^a \left( s_H^a \right)}{p_L^a \left( s_H^a \right)} &= \frac{s_H'}{1 - s_H'}
    &
    \frac{s_H}{1 - s_H} \times \frac{p_H^a \left( s_L^a \right)}{p_L^a \left( s_L^a \right)} &= \frac{s_L'}{1 - s_L'}
\end{align*}
A buyer who observes the initial low signal $s_L$ receives no further information. I illustrate this construction in 
Figure \ref{fig:strongergood_additional}.

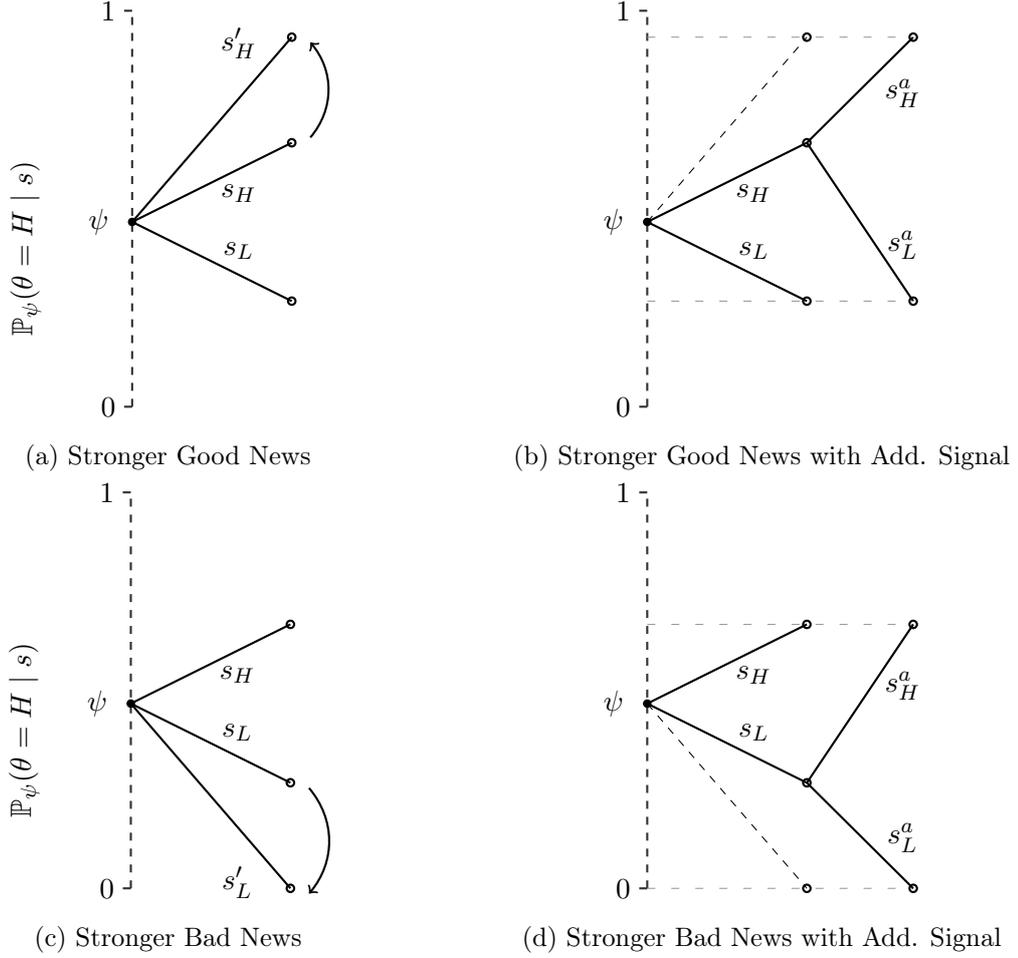
\begin{figure}[h]
    \centering
    \begin{tabular}{cc}
        % First Square
        \begin{subfigure}[t]{0.45\textwidth}
            %\label{fig:strongergood}
            %
            \centering
            \begin{tikzpicture}[scale=0.7]

                % Draw the dashed vertical line from (0,0) to (0,10)
                \node[rotate=90] at (-2,4.5) {$\Prob_{\interim} ( \quality = H \mid \typesig )$};
                \draw[dashed, thick, black!80] (0,1.5) -- (0,9);
                \draw[dashed, thick, black!80] (0,1.5) -- (-0.25,1.5);
                \draw[dashed, thick, black!80] (0,9) -- (-0.25,9);
                \node at (-0.45,1.5) {$0$};
                \node at (-0.45,9) {$1$};
                
                % Draw the black filled circle at (0,5)
                \filldraw[black] (0,5) circle (2pt);
                
                % Draw the lines emanating from (0,5) to (4,2.5) and (4,7.5)
                \draw[thick, black] (0,5) -- (3,3.5);
                \draw[thick, black] (0,5) -- (3,6.5);
                \draw[thick, black] (0,5) -- (3,8.5);
                
                % Draw the empty circles at (4,2.5) and (4,7.5)
                \draw[thick] (3,3.5) circle (2pt);
                \draw[thick] (3,6.5) circle (2pt);
                \draw[thick] (3,8.5) circle (2pt);
                
                \node[left] at (-0.25,5) {$\psi$};
                
                % Add an arrow from above the (0,5) -- (3,3.5) line to the (0,5) -- (3,6.5) line
                \draw[->, thick] (3.35,6.6) to[bend right=40] (3.35,8.4);
                
                % Labels for signal types
                \node[below] at (2,5.9) {$\typesig_H$};
                \node[above] at (2,4.1) {$\typesig_L$};
                \node[above] at (2,7.9) {$\typesig_H'$};

            \end{tikzpicture}
            \caption{Stronger Good News}
            \label{fig:strongergood}
        \end{subfigure} &
        
        % Second Square (duplicate of Figure 1)
        
        \begin{subfigure}[t]{0.45\textwidth}
            \centering
            \begin{tikzpicture}[scale=0.7]

                % Draw the dashed vertical line from (0,0) to (0,10)
                %\node[rotate=90] at (-2,4.5) {$\Prob_{\interim} ( \quality = H \mid \typesig )$};
                \draw[dashed, thick, black!80] (0,1.5) -- (0,9);
                \draw[dashed, thick, black!80] (0,1.5) -- (-0.25,1.5);
                \draw[dashed, thick, black!80] (0,9) -- (-0.25,9);
                \node at (-0.45,1.5) {$0$};
                \node at (-0.45,9) {$1$};
                
                % Draw the black filled circle at (0,5)
                \filldraw[black] (0,5) circle (2pt);
                
                % Draw the lines emanating from (0,5) to (4,2.5) and (4,7.5)
                \draw[thick, black] (0,5) -- (3,3.5);
                \draw[thick, black] (0,5) -- (3,6.5);
                \draw[dashed] (0,5) -- (3,8.5);
                
                % Draw the empty circles at (4,2.5) and (4,7.5)
                \draw[thick] (3,3.5) circle (2pt);
                \draw[thick] (3,6.5) circle (2pt);
                \draw[thick] (3,8.5) circle (2pt);

                \draw[thick, black] (3,6.5) -- (5,8.5);
                \draw[thick] (5,8.5) circle (2pt);
                \draw[loosely dashed,black!40] (0,8.5) -- (5,8.5);
                \draw[thick, black] (3,6.5) -- (5,3.5);
                \draw[thick] (5,3.5) circle (2pt);
                \draw[loosely dashed,black!40] (0,3.5) -- (5,3.5);
                
                \node[left] at (-0.25,5) {$\psi$};
                
                % Labels for signal types
                \node[below] at (2,5.9) {$\typesig_H$};
                \node[above] at (2,4.1) {$\typesig_L$};
                %\node[above] at (2,7.9) {$\typesig_H'$};
                \node[below] at (4.8,7.9) {$\typesig_H^a$};
                \node[above] at (4.8,4.1) {$\typesig_L^a$};

            \end{tikzpicture}
            \caption{Stronger Good News with Add. Signal}
            \label{fig:strongergood_additional}
        \end{subfigure} \\

        % Third Square
        \begin{subfigure}[t]{0.45\textwidth}
            \centering
            \begin{tikzpicture}[scale=0.7]

                % Draw the dashed vertical line from (0,0) to (0,10)
                \node[rotate=90] at (-2,4.5) {$\Prob_{\interim} ( \quality = H \mid \typesig )$};
                \draw[dashed, thick, black!80] (0,1.5) -- (0,9);
                \draw[dashed, thick, black!80] (0,1.5) -- (-0.25,1.5);
                \draw[dashed, thick, black!80] (0,9) -- (-0.25,9);
                \node at (-0.45,1.5) {$0$};
                \node at (-0.45,9) {$1$};
                
                % Draw the black filled circle at (0,5)
                \filldraw[black] (0,5) circle (2pt);
                
                % Draw the lines emanating from (0,5) to (4,2.5) and (4,7.5)
                \draw[thick, black] (0,5) -- (3,3.5);
                \draw[thick, black] (0,5) -- (3,1.5);
                \draw[thick, black] (0,5) -- (3,6.5);
                
                % Draw the empty circles at (4,2.5) and (4,7.5)
                \draw[thick] (3,3.5) circle (2pt);
                \draw[thick] (3,1.5) circle (2pt);
                \draw[thick] (3,6.5) circle (2pt);
                
                \node[left] at (-0.25,5) {$\psi$};
                
                % Add curved arrow for direction of improvement
                \draw[->, thick] (3.35,3.4) to[bend left=40] (3.35,1.4);
                
                % Labels for signal types
                \node[below] at (2,5.9) {$\typesig_H$};
                \node[above] at (2,4.1) {$\typesig_L$};
                \node[below] at (2,2.1) {$\typesig_L'$};

            \end{tikzpicture}
            \caption{Stronger Bad News}
            \label{fig:strongerbad}
        \end{subfigure} &
        
        % Fourth Square (duplicate of Figure 3)
        \begin{subfigure}[t]{0.45\textwidth}
            \centering
            \begin{tikzpicture}[scale=0.7]

                % Draw the dashed vertical line from (0,0) to (0,10)
                \draw[dashed, thick, black!80] (0,1.5) -- (0,9);
                \draw[dashed, thick, black!80] (0,1.5) -- (-0.25,1.5);
                \draw[dashed, thick, black!80] (0,9) -- (-0.25,9);
                \node at (-0.45,1.5) {$0$};
                \node at (-0.45,9) {$1$};
                
                % Draw the black filled circle at (0,5)
                \filldraw[black] (0,5) circle (2pt);
                
                % Draw the lines emanating from (0,5) to (4,2.5) and (4,7.5)
                \draw[thick, black] (0,5) -- (3,3.5);
                \draw[dashed] (0,5) -- (3,1.5);
                \draw[thick, black] (0,5) -- (3,6.5);
                
                % Draw the empty circles at (4,2.5) and (4,7.5)
                \draw[thick] (3,3.5) circle (2pt);
                \draw[thick] (3,1.5) circle (2pt);
                \draw[thick] (3,6.5) circle (2pt);
                
                \node[left] at (-0.25,5) {$\psi$};
                
                % Labels for signal types
                \node[below] at (2,5.9) {$\typesig_H$};
                \node[above] at (2,4.1) {$\typesig_L$};

                \draw[loosely dashed,black!40] (0,6.5) -- (5,6.5);
                \draw[loosely dashed,black!40] (0,1.5) -- (5,1.5);

                \draw[thick, black] (3,3.5) -- (5,6.5);
                \draw[thick] (5,6.5) circle (2pt);
                \draw[thick, black] (3,3.5) -- (5,1.5);
                \draw[thick] (5,1.5) circle (2pt);
                \node[below] at (4.8,5.8) {$\typesig_H^a$};
                \node[above] at (4.8,2) {$\typesig_L^a$};

            \end{tikzpicture}
            \caption{Stronger Bad News with Add. Signal}
            \label{fig:strongerbad_additional}
        \end{subfigure}
    \end{tabular}
    \caption{Blackwell Improvements of a Binary Signal}
    \label{figure:binaryimprovement}
\end{figure}

\looseness=-1
So, observing the sequence $\left( s_H, s_H^a \right)$ conveys the same information as the signal $s_H'$ from $\experiment'$ does. 
This information leads to an acceptance. Either the sequence $\left( s_H, s_L^a \right)$ or the signal $s_L$ convey the same information as the signal $s_L'$ from $\experiment'$ does. This information leads to a rejection.

Our reinterpretation reveals how this additional information affects trade. 
A buyer who observes an initial low signal $s_L$ receives no additional information. She rejects the seller, as before. However, a buyer who receives an initial high signal $s_H$ receives additional information through $\experiment^a$. Absent this additional information, she would have traded. But a low signal $s_L^a$ from $\experiment^a$ ``negatively overrides'' that initial verdict: now, she rejects the seller.

Thus, stronger good news jeopardises trade: a seller who previously would have traded with some buyer may now be rejected by every buyer. This raises total surplus: the seller is rejected by every buyer because each observe a low signal, either $s_L$ or $s_L^a$. Trading with such a seller would yield negative surplus in expectation: if a buyer with a low signal expects a negative surplus from trade when she merely \textit{suspects} others to have received low signals, she would certainly expect a negative surplus if she \textit{knew} all others had received low signals too.

Now, let us turn to the case of stronger bad news. 
To shed light on the threshold beyond which stronger bad news hurts total surplus, let us focus on a an experiment $\experiment'$ that delivers ``marginally'' stronger bad news than $\experiment$, i.e., $s_L' = s_L - \delta$ for a vanishingly small $\delta > 0$, but the same strength of good news as $\experiment$, $s_H' = s_H$.

As before, let us reinterpret the additional information a buyer can extract from $\experiment'$ as an additional signal she can observe. This time, she observes this additional signal only after an initial low signal $\typesig_L$. As before, the conditional distribution of this additional signal is such that when appended to $\experiment$, it mimics the improvement $\experiment'$ offers over $\experiment$:
\begin{align*}
    \frac{\typesig_L}{1 - \typesig_L} \times \frac{p_H^a \left(\typesig_H^a \right)}{p_L^a \left(\typesig_H^a \right)} &= \frac{\typesig_H}{1 - \typesig_H}
    &
    \frac{\typesig_L}{1 - \typesig_L} \times \frac{p_H^a \left(\typesig_L^a \right)}{p_L^a \left(\typesig_L^a \right)} &= \frac{\typesig_L'}{1 - \typesig_L'} = \frac{\typesig_L - \delta}{1 - (\typesig_L - \delta)}
\end{align*}
I illustrate this construction in Figure \ref{fig:strongerbad_additional}.

A buyer who receives an initial high signal $\typesig_H$ observes no additional information. 
She trades with the seller. However, a buyer who receives an initial low signal $\typesig_L$ receives additional information through $\experiment^a$. Absent this additional information, she would have rejected the seller. But a high signal $\typesig_H^a$ from $\experiment^a$ ``positively overrides'' her initial verdict: now, she trades. 

So, stronger bad news encourages trade: a seller who previously would have been rejected by every buyer might now trade with one. The effect this has on total surplus is less clear. What can we infer about the quality of this seller?

The seller who would have been rejected by every buyer absent the additional information $\experiment^a$ supplies; so, initially all buyers observed low signals, $\typesig_L$. Upon the additional information conveyed by $\experiment^a$, some buyer revises her decision to an acceptance; so, at least one buyer must have observed a high signal $\typesig_H^a$ from this additional experiment, but the rest must have observed the low signal $s_L^a$. However, the rest observed the low signal $s_L^a$ from this additional experiment. Whether the expected surplus from trade is positive depends on \textit{how many} did so:
\begin{equation*}
    \Prob \left( \quality = H \mid \geq 1 \textrm{ buyer observed } s_H^a  \right) 
    = \sum\limits_{k=1}^{n} 
    \left[
        \begin{aligned}
            &\Prob \left( \quality = H \mid k \textrm{ buyers observed } s_H^a \right) \\
            \times &\Prob \left( k \textrm{ buyers observed } s_H^a \mid  \geq 1 \textrm{ buyer obs. } s_H^a \right)
        \end{aligned}
    \right]
\end{equation*}

However, we may deduce that almost surely only one buyer observed the high signal $s_H^a$. To see this, observe the likelihood ratios of the signals a buyer may observe from the experiment $\experiment^a$:
\begin{align*}
    \frac{s_L^a}{1 - s_L^a} &= \frac{
        \frac{s_L - \delta}{1 - (s_L - \delta)}
    }{
        \frac{s_L}{1 - s_L}
    } %\overset{\delta}{\to} 0
    &
    \frac{s_H^a}{1 - s_H^a} = \frac{
        \frac{s_H}{1 - s_H}
    }{
        \frac{s_L}{1 - s_L}
    }
\end{align*}
While the likelihood ratio for the high signal $s_H^a$ is constant, that for the low signal $s_L^a$ converges to $1$ as $\delta \downarrow 0$ and the improvement in buyers' information becomes ``smaller''. Due to the martingale property of likelihood ratios, these likelihood ratios must average to $1$; so, the probability that buyers will observe the additional high signal $s_H^a$ must vanish as $\delta \downarrow 0$. 

So, the expected surplus from this trade is non-negative if and only if:
\begin{align*}
    \lim\limits_{\delta \downarrow 0} \Prob \left( \quality = H \mid \geq 1 \textrm{ buyer observed } s_H^a  \right) - c &= \Prob \left( \quality = H \mid 1 \textrm{ buyer observed } s_H^a  \right) - c \geq 0 \\
    &\iff \frac{\prior}{1 - \prior} \times \left[ \frac{s_L}{1 - s_L} \right]^{n-1} \times \frac{s_H}{1 - s_H} \geq  \frac{c}{1-c}
\end{align*}

Marginally stronger bad news allows the most adversely selected seller to trade. Had buyers' not received additional information, he would not have traded. Because they do, one buyer accepts him. The RHS of the expression above reflects this: the expected surplus from trading with this seller is positive if and only if the high signal $s_H^a$ observed by that one buyer overpowers the low signals $s_L^a$ observed by the remaining $n-1$ buyers.\footnote{Note that $s_L^a \to s_L$ as $\delta \downarrow 0$.}

This is a stark condition; when it holds, we may say that \textit{adverse selection is irrelevant}: a buyer who observes the high signal need not be concerned about any previous rejections the seller may have received. It is also closely linked to the threshold Theorem \ref{thm:intensive_binary} identifies, which Proposition \ref{prop:thresholdbinary} characterises. 

\begin{restatable}{prop}{propthresholdbinary}
    \label{prop:thresholdbinary}
    Let buyers' experiment be binary. Then, equilibrium$^{\star}$ total surplus weakly decreases with stronger bad news (lower $s_L$) when:
    \begin{equation*}
        \frac{\prior}{1 - \prior} \times \max \left\{ \frac{s_L}{1 - s_L}, \left[ \frac{s_L}{1 - s_L} \right]^{n-1} \times  \frac{s_H}{1 - s_H} \right\} \leq \frac{c}{1-c}
    \end{equation*}
    This condition is also necessary in the least selective equilibrium. 
\end{restatable}

\begin{cor}
    \label{cor:thresholdbinary}
    Where buyers' experiment is binary and $\prior \leq c$, equilibrium$^{\star}$ total surplus weakly decreases with stronger bad news (lower $s_L$) when $\left( \frac{s_L}{1 - s_L} \right)^{n-1} \times \frac{s_H}{1 - s_H} \leq \frac{c}{1-c}$.
\end{cor}

\noindent \looseness=-1 Proposition \ref{prop:thresholdbinary} reveals that total surplus falls with stronger bad news once bad news is strong enough to:
\vspace{-2pt}
\begin{itemize}[itemsep=-2pt]
    \item violate the condition we informally identified as the ``irrelevance of adverse selection'', and
    \item for buyers to reject the seller with positive probability in their equilibrium$^{\star}$ strategies.
\end{itemize}
Under some parameter constellations, even when adverse selection is no longer irrelevant, bad news might need to get stronger before, in equilibrium, buyers start rejecting the seller upon a low signal. Before the strength of bad news hits this critical mark, total surplus equals the no-information benchmark. Once it does, total surplus experiences a one time upward jump. Thereafter, stronger bad news decreases total surplus. Corollary \ref{cor:thresholdbinary} reveals, however, that this is not a concern when $\prior < c$: then, buyers always reject upon a low signal, so stronger bad news decreases total surplus whenever the condition we dubbed ``the irrelevance of adverse selection'' is satisfied. 

\subsection{Finite Experiments}
\label{section:finite}

Here, I show how the ideas we through Theorem \ref{thm:intensive_binary} generalise to Blackwell improvements of experiments with an arbitrary finite number of outcomes. 
We cannot deploy those ideas immediately: first, for non-binary experiments, the ideas of stronger good and bad news lose meaning; second, Blackwell improvements of such experiments are complex---they cannot be described by simple changes in signals' likelihood ratios. Nonetheless, the core idea behind Theorem \ref{thm:intensive_binary} supplies the answer: whether total surplus improves depends on whether an improvement is a \textit{positive override}---information about a seller who would be rejected---or a \textit{negative override}---information about a seller who would be approved. 

Before I state Theorem \ref{thm:intensive}, I establish two definitions that will provide useful langugage: Definition \ref{defn:local_mps} formalises the notion of a positive and negative override, and Definition \ref{defn:irrelevance} formalises the notion of the ``irrelevance of adverse selection''.

\begin{defn}
    \label{defn:local_mps}
    Enumerate the joint outcome set of the experiments $\experiment = (\signalset, p_L, p_H)$ and $\experiment' = (\signalset', p_L', p_H')$ as $ \signalset \cup \signalset' = \left\{ s_1, s_2, ..., s_M \right\}$. The experiment $\experiment'$ differs from $\experiment$ by a \textit{local mean preserving spread} (or, \textit{local spread}) \textit{at} $s_j \in \signalset$ if:
    \begin{align*}
        p_{\quality}' (s_j) &= 0 & 
        p_{\quality} (s_{j+1}) = p_{\quality} (s_{j-1}) &= 0  &
        p_{\quality}'(s_{j+1}) + p_{\quality}' (s_{j-1}) &= p_{\quality} (s_j)
    \end{align*}
    and $p_{\quality}'(s) = p_{\quality} (s)$ for any $s \in \signalset' \cup \signalset \setminus \{s_{j-1}, s_j, s_{j+1} \}$.
\end{defn}

A local spread moves all the probability mass experiment $\experiment$ places on some signal $s_j$ to two new signals---one better news about the asset's quality, $s_{j+1}$, one worse, $s_{j-1}$. In this sense, we can think of it as providing additional information to a buyer who observes the original signal $s_j \in \signalset$. 

Every local spread is an ordinary mean preserving spread.\footnote{Specifically, a ``3-part MPS'', in the language of \textcite{rasmusen_petrakis_3ptmps_92}. Mean preserving spreads are originally due to \textcite{muirhead_1900} and were popularised in economics by \textcite{rothschild_stiglitz_1970}.}. The converse is not true; local spreads must move \textit{all} the probability mass on a signal, and they must move it to its neighbouring signals: the mass on $s_j$ is spread to $s_{j+1}$ and $s_{j-1}$. I visualise the construction of a local spread in Figure \ref{fig:localspread}.\footnote{For clarity, I fix $\interim = 0.5$ in this illustration.}

\looseness=-1
Restricting to local spreads is without loss for finite experiments\footnote{However, local mean preserving spreads are only defined for finite experiments; see \textcite{muller_scarsini} and \textcite{muller_stoyan_comparison}.}---every Blackwell improvement, and \textit{a fortiori}, ordinary mean preserving spread can be constructed through a finite number of local spreads.

\begin{restatable}{rem}{remlocalmps}[\textcite{muller_stoyan_comparison}, Theorem 1.5.29]
    \label{rem:localmps}
    An experiment $\experiment'$ is Blackwell more informative than another, $\experiment$, if and only if there is a finite sequence of experiments $\signalstr_1, \signalstr_2, ..., \signalstr_k$ such that $\signalstr_1 = \signalstr$, $\signalstr_k = \signalstr'$, and $\signalstr_{i+1}$ differs from $\signalstr_i$ by a local spread.
\end{restatable}

\begin{center}
           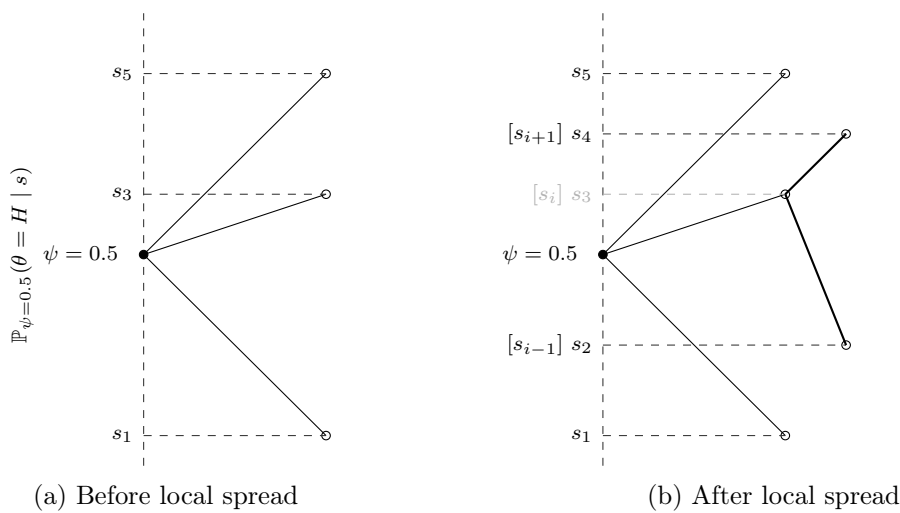
\begin{figure}[h]
                \centering
                \hspace{0pt}
                \begin{subfigure}[b]{0.45\textwidth}
                    \centering
                    \begin{tikzpicture}[scale=0.8]
                        \begin{scope}[every node/.style={font=\scriptsize}]
                            % Draw the dashed vertical line from (0,0) to (0,8)
                        \node[rotate=90] at (-2,5) {$\Prob_{\interim = 0.5} ( \quality = H \mid \typesig )$};
                        \draw[dashed, thin,black!80] (0,1.5) -- (0,9);
                        
                        % Draw the black filled circle at (0,5)
                        \filldraw[black] (0,5) circle (2pt);
                        
                        % Draw the lines emanating from (0,5) to (4,2.5) and (4,7.5)
                        \draw (0,5) -- (3,2);
                        \draw (0,5) -- (3,8);
                        \draw (0,5) -- (3,6);
                        
                        % Draw the empty circles at (4,2.5) and (4,7.5)
                        \draw (3,2) circle (2pt);
                        \draw (3,8) circle (2pt);
                        \draw (3,6) circle (2pt);
                        
                        % Draw the dashed horizontal line from (0,2.5) to (8,2.5)
                        \draw[dashed,thin,black!80] (0,2) -- (3,2);
                        \draw[dashed,thin,black!80] (0,8) -- (3,8);
                        \draw[dashed,thin,black!80] (0,6) -- (3,6);
                        
                        % Label the left of node (0,2.5) as $s_L$
                        \node[left] at (0,2) {$s_1$};
                        \node[left] at (0,8) {$s_5$};
                        \node[left] at (0,6) {$s_3$};
                        
                        \node[left] at (-0.25,5) {$\interim = 0.5$};    
                        \end{scope}
                    \end{tikzpicture}
                    \caption{Before local spread}
                \end{subfigure}
                %\hfill
                \hspace{10pt}
                \begin{subfigure}[b]{0.45\textwidth}
                    %\centering
                    \begin{tikzpicture}[scale=0.8]
                        \begin{scope}[every node/.style={font=\scriptsize}]
                            % Draw the dashed vertical line from (0,0) to (0,8)
                        \draw[dashed, thin,black!80] (0,1.5) -- (0,9);
                        
                        % Draw the black filled circle at (0,5)
                        \filldraw[black] (0,5) circle (2pt);
                        
                        % Draw the lines emanating from (0,5) to (4,2.5) and (4,7.5)
                        \draw (0,5) -- (3,2);
                        \draw (0,5) -- (3,8);
                        \draw[black] (0,5) -- (3,6);
                        
                        % Draw the empty circles at (4,2.5) and (4,7.5)
                        \draw (3,2) circle (2pt);
                        \draw (3,8) circle (2pt);
                        \draw (3,6) circle (2pt);
                        \draw[color=black] (4,3.5) circle (2pt);
                        \draw[color=black] (4,7) circle (2pt);
                        
                        % Draw the dashed horizontal line from (0,2.5) to (8,2.5)
                        \draw[dashed,thin,black!80] (0,2) -- (3,2);
                        \draw[dashed,thin,black!80] (0,8) -- (3,8);
                        \draw[dashed,thin,black!30] (0,6) -- (3,6);
                        \draw[dashed,thin,black!80] (0,7) -- (4,7);
                        \draw[dashed,thin,black!80] (0,3.5) -- (4,3.5);
                        
                        % Label the left of node (0,2.5) as $s_L$
                        \node[left] at (0,2) {$s_1$};
                        \node[left] at (0,8) {$s_5$};
                        \node[left,black!30] at (0,6) {$[s_i]\textrm{ } s_3$};
                        \node[left,black] at (0,7) {$[s_{i+1}] \textrm{ } s_4$};
                        \node[left,black] at (0,3.5) {$[s_{i-1}] \textrm{ } s_2$};
            
                        \draw[thick,black] (3,6) -- (4,7);
                        \draw[thick,black] (3,6) -- (4,3.5);
                                    
                        \node[left] at (-0.25,5) {$\interim = 0.5$};      
            
                        \end{scope}
                    \end{tikzpicture}    
                    \caption{After local spread}
                \end{subfigure}
                \caption{A Local Mean Preserving Spread 
                }
                \label{fig:localspread}
            \end{figure}
\end{center}

\begin{defn}
    \label{defn:overrides}
    Let experiment $\experiment'$ differ from $\experiment$ by a local spread at $s_j$ and $\strat_{\experiment}^{\star}$ denote the equilibrium$^{\star}$ strategy under experiment $\experiment$. This local spread is a \textit{negative override} under equilibrium$^{\star}$ if $\strat_{\experiment}^{\star} (s_j) = 1$, and a \textit{positive override} if $\strat_{\experiment}^{\star} (s_j) = 0$.
\end{defn}

Both positive and negative overrides are local spreads of a buyer's experiment, but they differ in \textit{which} seller they inform the buyer about. A negative override is a local spread of a signal upon which a seller would have traded with the buyer.
A positive override is a local spread of a signal upon which the seller would have been rejected by the buyer. 

\begin{defn}
    \label{defn:irrelevance}
    Let $\strat$ be a monotone strategy for a fixed experiment $\experiment$. \textit{Adverse selection is} $\strat$\textit{-irrelevant for signal} $s \in \signalset$ if:
    \begin{equation*}
        \frac{\prior}{1 - \prior} \times \left[ \frac{r_H \left( \strat; \experiment \right)}{r_L \left( \strat; \experiment \right)} \right]^{n-1} \times \frac{\typesig}{1 - \typesig} \geq \frac{c}{1-c}
    \end{equation*}
\end{defn}

When adverse selection is $\strat$-irrelevant for a signal $\typesig \in \signalset$, a buyer finds it optimal to trade upon the signal $\typesig \in \signalset$ even if all other buyers rejected the seller---provided those buyers used the strategies $\strat$. 

\begin{restatable}{thm}{thmintensive}
    \label{thm:intensive}
    Let the experiment $\experiment'$ differ from $\experiment$ by a local spread at $s_j \in \signalset$. Equilibrium$^{\star}$ total surplus is:
    \begin{enumerate}
        \item weakly greater under $\experiment'$ if the local spread is a negative override under equilibrium$^{\star}$.
        \item weakly less under $\experiment'$ if the local spread is a positive override under equilibrium$^{\star}$, unless adverse selection is $\strat_{\experiment}^{\star}$-irrelevant for signal $s_{j+1}$.
    \end{enumerate}
\end{restatable}

Theorem \ref{thm:intensive} shows that the effect of a local spread on total surplus depends on the \textit{kind} of the spread. Negative overrides always increase total surplus. Positive overrides decrease it---unless adverse selection is irrelevant for a buyer who receives the override. 

To prove Theorem \ref{thm:intensive}, I show that a negative override indeed pushes buyers to reject the seller more often in the new equilibrium$^{\star}$. This always increases total surplus. A positive override pushes buyers to trade more often with him. This decreases total surplus unless adverse selection is $\strat^{\star}_{\experiment}$-irrelevant for signal $s_{j+1}$---the new signal upon which the buyer may approve the seller.

This exercise is complicated by the fact that identifying how buyers' interim beliefs change when their experiment changes is infeasible beyond the simplest cases. 
Studying ``local'', not ordinary, mean preserving spreads is crucial for tractability; this allows me to identify how equilibria respond to improvements in information without needing to pinpoint changes in interim beliefs.

Unlike Theorem \ref{thm:intensive_binary}, an analyst who wishes to use Theorem \ref{thm:intensive} to identify the effect of an improvement on total surplus needs knowledge of buyers' equilibrium$^{\star}$ strategies. In practice, the analyst might want to remain agnostic about equilibrium$^{\star}$ strategies. To alleviate this concern, Proposition \ref{prop:sufficient_harm} offers a sufficient condition for a positive override to decrease total surplus in the most selective equilibrium.

\begin{restatable}{prop}{propsufficientharm}
    \label{prop:sufficient_harm}
    Let the experiment $\signalstr'$ differ from $\signalstr$ by a local spread at $s_j$. Total surplus in the most selective equilibrium is lower under $\experiment'$ if the following conditions hold:
\begin{align*}
    \frac{\prior}{1 - \prior} \times \left( \frac{s_j}{1 - s_j} \right) &\leq \frac{c}{1-c}
    &
    &\textrm{and}
    &
    \frac{\prior}{1 - \prior}
            \times
            \left(
                \frac{s_j}{1 - s_j}
            \right)^{n-1}
            \times
            \frac{s_{j+1}}{1 - s_{j+1}}
            &\leq
            \frac{c}{1-c}
\end{align*}

\end{restatable}

The sufficient conditions in Proposition \ref{prop:sufficient_harm} strengthen the necessary and sufficient conditions supplied by Theorem \ref{thm:intensive}. The condition on the left ensures that the local spread is a negative override: a buyer rejects the seller upon $s_j$ in any equilibrium since interim beliefs always lie below the prior belief, $\prior$. The condition on the right strengthens the irrelevance condition for adverse selection: it requires that a rejection be optimal even if the $n-1$ rejections the seller received were due to the best signals below $s_{j+1}$, $s_j$.

\section{Maximising Surplus by Coarsening Information}
\label{section:regulator}

Section \ref{section:betterinfo} showed that efficiency might be lower in a market where buyers are better informed: more information in the market accentuates the adverse selection problem each buyer faces, counteracting each buyer's improved ability to screen the asset's quality. This invites a question: if a regulator could coarsen buyers' information---perhaps through banning the use and dissemination of certain data---could this raise efficiency? How should a regulator who can use this tool to maximise total surplus in the market go about this exercise?

In this section, I consider the problem of a regulator who wishes to garble buyers' experiment $\experiment$ in order to maximise total surplus in (the most selective) equilibrium. As in the previous section, I take the number of buyers $n$ to be a primitive. The regulator can choose any finite garbling $\experiment^G = \left( \signalset^G, p_L^G, p_H^G \right)$ of buyers' experiment $\experiment$; i.e. any finite set of outcomes $\signalset^G = \left\{ \typesig_1^G, \typesig_2^G, ..., \typesig_R^G \right\}$ and conditional distributions $p_{\quality}^G(.)$ over it such that for some Markov matrix $\mathbf{T}_{m \times R}$:
\begin{equation*}
   \underbrace{
        \begin{bmatrix}
            p_L(\typesig_1) & \cdots & p_L(\typesig_m) \\
            p_H(\typesig_1) & \cdots & p_H(\typesig_m)
        \end{bmatrix}
   }_{= \mathbf{P}}
    \times
    \mathbf{T}
    =
    \underbrace{
        \begin{bmatrix}
            p_L^G(\typesig_1^G) & \cdots & p_L^G(\typesig_R^G) \\
            p_H^G(\typesig_1^G) & \cdots & p_H^G(\typesig_R^G)
        \end{bmatrix}
    }_{= \mathbf{P^G}}
\end{equation*}
We can interpret this as a coarsening of the original data available to each buyer. This original data is generated with the process that matrix $\mathbf{P}$ describes. The regulator garbles this data into a set of summary statistics through the process described by the Markov matrix $\mathbf{T}$. The buyer only observes this set of summary statistics, whose data generating process conditional on the asset's quality can now described by the matrix $\mathbf{P^G}$.

Once the regulator chooses the garbled experiment $\experiment^G$, the game proceeds as before; only, buyers' experiment $\experiment$ is replaced by $\experiment^G$. An equilibrium, as before, is a pair $\left( \strat^G, \interim^G \right)$ such that the strategy $\strat^G: \signalset^G \to [0,1]$ is optimal given the interim belief $\interim^G$, and the interim belief $\interim^G$ is consistent with the strategy $\strat^G$. I call this the game induced by the garbling $\experiment^G$. 
I call a garbling $\experiment^G$ regulator-optimal if total surplus in the most selective equilibrium of the game it induces weakly exceeds equilibrium total surplus in the game induced by any other garbling of $\experiment$.

A particularly simple class of garblings are the \textit{monotone binary garblings}, which (i) have two possible outcomes, $\lvert \signalset^G \rvert = 2$, and (ii) for a cutoff signal $\typesig^* \in \signalset$, the entries $\{ t_{ij} \}$ of matrix $\mathbf{T}_{m \times 2}$ for which $\mathbf{P} \times \mathbf{T} = \mathbf{P^G}$ are given by:
\begin{align*}
    t_{i1} &= 
    \begin{cases}
        1 & i < i^* \\
        \vspace{-25pt} \\
        \in [0,1] & i = i^*  \\
        \vspace{-25pt} \\
        0  & i > i^* 
    \end{cases}
    &
    t_{i2} &= 1 - t_{i1}
\end{align*}
I refer to the signal $s_{i^*} \in \signalset$ as the \textit{threshold signal} of the garbling $\experiment^G$. 

A monotone binary garbling gives the buyer an ``acceptance recommendation'', $s_H^G$, when her original signal realises above a threshold signal $\typesig_{i^*} \in \signalset$, and a ``rejection recommendation'', $s_L^G$, whenever it lies below it. Following these recommendations---accepting trade upon the signal $s_H^G$ and rejecting when it upon the signal $s_L^G$---need not be an equilibrium strategy in the game induced by the coarsened experiment $\experiment^G$. When it does, I say that the garbling $\experiment^G$ \textit{produces incentive compatible recommendations}.

\begin{defn}
    A monotone binary garbling $\experiment^G$ \textit{produces incentive compatible (IC) recommendations} if the strategy $\strat^G$, defined below, is an equilibrium strategy in the game induced by $\experiment^G$:
    \begin{equation*}
        \strat^G (s^G) := 
        \begin{cases}
            0 & s^G = s_L^G \\
            1 & s^G = s_H^G
        \end{cases}
    \end{equation*}
\end{defn}

\begin{restatable}{lem}{lemdesignmonotonebinary}
        \label{lem:design_monotone_binary_garblings}
        Where it exists, the regulator-optimal garbling is {monotone binary} and {produces IC recommendations}.
    \end{restatable}

The reason that the regulator can restrict herself to binary garblings is closely connected to a fundamental principle in information design. Ultimately, a buyer distils the information relayed by the garbled experiment into which action it recommends: an acceptance, or a rejection. The regulator can distil information herself---revealing only a recommended action to a buyer.\footnote{Note that coarsening buyers' experiment $\experiment$ may also create new equilibrium outcomes, some of which yield lower payoff than the previous least selective equilibrium. Our focus on the most selective equilibrium---besides being the appropriate focus for this exercise---frees us from the need to worry about this complication and utilise this fundamental principle.} The regulator wishes trade to be as likely as possible with a High quality seller but not with a Low quality one, so she prefers monotone recommendations---those which recommend an acceptance above a threshold signal. Lemma \ref{lem:design_monotone_binary_garblings} establishes that monotone recommendations are not at the expense of incentive compatibility: every garbling that supplies IC recommendations is outperformed by a monotone binary garbling that supplies IC recommendations.

We can adopt the selectivity order for strategies to a selectivity order for monotone binary garblings, too.

\begin{defn}
    A monotone garbling $\experiment^G$ of $\experiment$ is \textit{more selective} than another, $\experiment^{G'}$, if $p_{\quality}^{G} \left( s_H^{G} \right) \leq p_{\quality}^{G'} \left( s_H^{G'} \right)$ for all $\quality \in \{L, H\}$.
\end{defn}
Like monotone strategies, and for the same reason, selectivity is a complete order over the set of monotone binary garblings. Similarly, we can adapt the ``irrelevance of adverse selection'' notion to monotone binary garblings as well.

\begin{defn}
    \label{defn:irrelevant_garbling}
    Let $\experiment^G$ be a monotone binary garbling of $\experiment$, with the threshold signal $\typesig^* \in \signalset$. 
    \textit{Adverse selection is irrelevant} under $\experiment^G$ either if:
    \begin{equation*}
        \frac{\prior}{1 - \prior} \times \frac{p_H (s^*)}{p_L(s^*)} \times \left( 
            \frac{p_{H}^G (s_L^G)}{p_L^G (s_L^G)}
        \right)^{n-1} \geq \frac{c}{1-c}
    \end{equation*}
    or either of the following two conditions hold:
    \begin{enumerate}
        \item $\experiment^G$ recommends no acceptances; i.e. $p_{\quality} (s_H^G) = 0$.
        \item $\experiment^G$ recommends no rejections; i.e. $p_{\quality} (s_L^G) = 0$ and $\frac{\prior}{1 - \prior} \times \left( \frac{s_1}{1 - s_1} \right)^n \geq \frac{c}{1-c}$.
    \end{enumerate}
\end{defn}

Adverse selection is irrelevant under a garbling if a buyer who receives an ``acceptance'' recommendation need not be concerned about the number of buyers who received ``rejection'' recommendations---even if she believes that her acceptance recommendation is based on the lowest signal that might have triggered it.  This condition is also satisfied if the garbling never recommends an ``acceptance'', or if, even if all $n$ buyers were to observe the lowest signal under experiment $\experiment$, the expected gains from trade would be positive. When this condition is violated, I say that adverse selection is \textit{not} irrelevant under $\experiment^G$.

\begin{restatable}{prop}
    {propgarbling}
    \label{prop:optimalgarbling}
    If the least selective monotone binary garbling under which adverse selection is irrelevant produces IC recommendations, it is the regulator-optimal garbling. Otherwise, the regulator-optimal garbling is either: 
    \vspace{-2pt}
    \begin{itemize}[itemsep=-2pt]
        \item the least selective garbling under which adverse selection is irrelevant, or
        \item the most selective garbling under which adverse selection is not irrelevant
    \end{itemize}
    among monotone binary garblings which produce IC recommendations.
\end{restatable}

\begin{cor}
    \label{cor:optimalgarbling}
    When the seller's reservation value $c$ weakly exceeds the prior belief $\prior$, the regulator-optimal garbling is the least selective monotone binary garbling under which adverse selection is irrelevant.
\end{cor}

Proposition \ref{prop:optimalgarbling} shows that the regulator wishes to recommend a rejection following every signal a buyer could observe, unless adverse selection is irrelevant at that signal. Although the regulator wishes to maximise a buyer's \textit{expected} contribution to trade surplus, she focuses on the ``worst case'' where a buyer is the last to receive the seller.

The intuition is intimately connected to the insight the previous section delivers: unless ``adverse selection is irrelevant'', information which pushes buyers to accept trade more often can harm efficiency. The regulator wishes to censor such information by coarsening buyers' experiment: if she were not bound by buyers' incentives to follow her recommendations, her optimal garbling would bundle every outcome of the original experiment $\experiment$ into a ``rejection recommendation'' unless adverse selection is irrelevant at that outcome. 
Corollary \ref{cor:optimalgarbling} establishes that when the seller's reservation value $c$ weakly exceeds buyers' prior belief $\prior$, such recommendations are IC, and hence are adopted by the regulator.

\section{Ultimatum Price Offers by Buyers}
\label{section:ultimatumprices}

In this section, I relax the assumption that the seller trades with the first buyer willing to pay his reservation value, $c$. Instead, each buyer he visits offers the seller a take-it-or-leave-it (ultimatum) price. The seller then decides whether to trade at this price or move on to the next buyer in search of a better offer. I show that in this extended game, there exists an equilibrium where a buyer never offers a price above the seller's value $c$. Furthermore and crucially, any level of total surplus that can be achieved in some equilibrium can be achieved in an equilibrium where a buyer never offers a price above the seller's value $c$. Thus, restricting to such equilibria is without loss when studying the possible levels of equilibrium total surplus.

As before, the seller visits the $n$ prospective buyers sequentially and in a uniformly random order, and each buyer holds a private signal about the asset's quality. The seller holds a private signal about the asset's quality---this signal is the outcome of a Blackwell experiment $\dot{\experiment}$. 
Conditional on the asset's quality, the seller's and buyers' signals are mutually independent. The seller enjoys a surplus of $o - c$ if he trades with some buyer at price $o$, and a surplus of $0$ if he does not trade. A buyer enjoys a surplus of $\indic \left\{ \quality = H \right\} - o$ if she trades with the seller at price $o$, and a surplus of $0$ if she does not trade. So, trade generates a total surplus of $\indic \left\{ \quality = H \right\} - c$.

Each buyer he visits offers the seller a take-it-or-leave-it price for the asset. A buyer's strategy $\omega : \signalset \to \Delta \left( \{ 0 \} \cup [c,1] \right)$ maps every signal she might observe to a distribution from which her offer is drawn.\footnote{I exclude offers in the interval $(0,c)$ without loss of generality: such offers will always be rejected by the seller since they are below his reservation value; the buyer might as well offer $0$ instead.}

\looseness=-1
Once the seller receives an offer of $o_k \in [0,1]$ from the $k$\textsuperscript{th} buyer he visits, he updates the probability $q_{k-1}$ he places on the asset having High quality to a probability $q_k$. He does so with a belief updating rule described by $q_k = \zeta_{\omega} \left( q_{k-1}, o_k \right)$, where $\zeta_{\omega} \left( q_{k-1}, o_k \right): [0,1]^2 \to [0,1]$ is a measurable function that is continuous in its first argument and can depend on the strategy $\omega$ the seller believes buyers to be using.\footnote{For instance, this updating rule might restrict the seller to update using Bayes' Rule wherever possible.} This belief updating procedure is commonly known to all players. 

The seller then decides whether to take or refuse the buyer's offer; his strategy $\chi_{k} (q_k, o_k): [0,1]^2 \to [0,1]$ is a measurable mapping from his belief that the asset has High quality, $q_k$, and the $k$\textsuperscript{th} buyer's offer, $o_k$, to a probability that he takes this $k$\textsuperscript{th} offer. 

Given a strategy $\omega^*$ for buyers and the seller's belief $q_k$ that the asset has High quality, denote the maximum expected surplus the seller can achieve after he rejects his $k$\textsuperscript{th} offer as $V_k^* (q_k)$. 
The seller secures zero surplus when he does not trade with any buyer; so, $V_{n}^*(q_n)= 0$ for all $q_n \in [0,1]$. We can then calculate each $V_k^*(q_k)$ recursively:
\vspace{-0.5cm}
\begin{equation*}
    V_k^* (q_k) = \sum\limits_{\typesig \in \signalset} \left( q_k \times p_H(s) + (1 - q_k) \times p_L(s) \right) \times \int\limits_{0}^{1} \max \left\{ V_{k+1}^* \left( \zeta_{\omega} (q_k, o_{k+1}) \right), o_{k+1} - c \right\} d \omega^*(s)\left( o_{k+1} \right)
\end{equation*}

Given a strategy $\omega^*$ for buyers and $\left\{ \chi_k^* \right\}_{k=1}^{n}$ for the seller, each buyer uses Bayes' Rule to calculate the probability that (i) she is the $k$\textsuperscript{th} buyer to be visited given the seller has not traded yet, (ii) given the seller is in his $k$\textsuperscript{th} visit, he will take an offer of $o_k$, and (iii) given the seller takes that offer in his $k$\textsuperscript{th} visit and the private signal $\typesig \in \signalset$ the buyer observed, the quality is High. Denote these probabilities as $\kappa^*(k)$, $\tau^*_{k,o_k}$, and $\small{\mathcal{H}}^*_{\normalsize{k,o_k,s}}$.  

The strategies $\omega^*$ and $\left\{ \chi_k^* \right\}_{k=1}^{n}$ form an equilibrium of this extended model if and only if:
\begin{enumerate}
    \item A buyer's offer maximises her expected surplus given her beliefs about the quality and behaviour of the seller:
    \begin{equation*}
        \textrm{supp } \omega^*(s) \subseteq \underset{o \in [0,1]}{\argmax} \sum\limits_{k=1}^{n} \kappa^*(k) \times \tau^*_{k,o} \times \left[ \small{\mathcal{H}}^*_{\normalsize{k,o,s}} - c \right]
    \end{equation*}
    \item The seller takes an offer if and only if his surplus from doing so weakly exceeds the expected surplus from refusing it and continuing his visits:
    \begin{equation*}
        \chi_k^* (q_k, o_k) = \indic \left\{ o_k - c \geq V_{k}^*(q_k) \right\}
    \end{equation*}
\end{enumerate}

\begin{restatable}{prop}{propultimatum}
    \label{prop:ultimatum}
    There exists an equilibrium where the seller trades whenever he is offered any price above $c$, and a buyer offers the seller either a price of $c$, or $0$. Furthermore, any level of total surplus that can be achieved in equilibrium can be achieved by one such equilibrium. 
\end{restatable}

\begin{proof}
The first part of the Proposition is seen easily: unless no buyer ever offers a price above $c$, no buyer has an incentive to offer such a price---the seller will accept the first offer weakly above $c$. 
By Proposition \ref{prop:eqmexist}, then, the existence of such an equilibrium is guaranteed. 

I prove the second part in two steps. I first show that buyers can never be offering a price above $c$ in equilibrium, unless a seller of some quality receives the same price offer from every buyer. Then, I show that if a seller of some quality does receive the same price offer from every buyer, the total surplus in this equilibrium equals total surplus in another equilibrium where the seller is never offered a price above $c$. 

\noindent \textit{Step 1:} Observe that the seller's value function $V_k^*(.)$ is continuous: $V_k^* (.)$ is continuous provided $V_{k+1}^*(.)$ is, and $V_n^*(.)$ is a constant function. So, choose $k^{\textrm{max}} \in \{1,2,...,n\}$ and $q^{\textrm{max}} \in [0,1]$ such that $V_{k^{\textrm{max}}}^* (q^{\textrm{max}}) = V^{\textrm{max}} := \max\limits_{k \in \{1,2,...,n\}, q \in [0,1]} V_k^*(q)$. If $V^{\textrm{max}} = 0$, a seller at most expects an offer of $c$; we are done. So, let $V^{\textrm{max}} > 0$. Then, no buyer can be offering a price above $V^{\textrm{max}} + c$ in equilibrium: any seller accepts such a price, so the buyer would have a profitable deviation to $V^{\textrm{max}} + c$. But then, a seller with a belief $q^{\textrm{max}}$ must be expecting an offer of $V^{\textrm{max}} + c$ with certainty in his $k+1$\textsuperscript{th} visit. This implies that for some quality $\quality \in \{L,H\}$, buyers offer the price $V^{\textrm{max}} + c$ after any signal $\typesig \in \textrm{supp } p_{\quality}$.

\noindent \textit{Step 2:} Consider an equilibrium where for some quality $\quality \in \{L,H\}$, buyers offer the price $V^{\textrm{max}} + c > c$ after any signal $\typesig \in \signalset$ such that $\typesig \in \textrm{supp } p_{\quality}$. 

\begin{enumerate}
    \item If buyers offer this price for any $\typesig \in \signalset$, there is another equilibrium where they instead offer the seller $c$ after every signal. In both equilibria, the seller trades with the first buyer he visits, regardless of his quality.
    \item If buyers offer this price for any $\typesig \in \textrm{supp } p_H$, the seller trades with the first buyer who offers this price whenever he has High quality. Buyers must be offering $0$ after any $\typesig \in \textrm{supp } p_L \setminus \textrm{supp } p_H$. There is then another equilibrium achieving the same total surplus where buyers offer $c$ after any $\typesig \in \textrm{supp } p_H$ but $0$ after any any $\typesig \in \textrm{supp } p_L \setminus \textrm{supp } p_H$. The seller trades with the first buyer who offers $c$.
    \item If buyers offer this price for any $\typesig \in \textrm{supp } p_L$, the seller trades with the first buyer whenever he has Low quality. Buyers cannot be ever offering the seller a price of $0$ following $s \in \textrm{supp } p_H \setminus \textrm{supp } p_L$; otherwise, a buyer would have a strictly positive deviation to an offer of $c$, which guarantees a profitable trade when she is the last buyer to be visited by the seller. So, the seller trades with probability one regardless of his quality. There is then another equilibrium where buyers offer the seller a price of $c$ regardless of their signal, and the seller trades with probability one at this price.  This equilibrium achieves the same total surplus. 
\end{enumerate}
\end{proof}

\newpage

\section{Supplemental Appendix}
\label{section:supplements}

Section \ref{section:morebuyers} claimed that the outcome in a large market might be partially informative about the common value of the asset. This section presents a numerical example to demonstrate this claim.

Let buyers' experiment $\experiment$ be binary; $\signalset = \{0.2, 0.8\}$ and: 
\begin{align*}
    p_L(\typesig) &=
    \begin{cases}
       0.8 & \typesig = 0.2 \\
       0.2 & \typesig = 0.8
    \end{cases}
    &
    p_H(\typesig) &=
    \begin{cases}
        0.2 & \typesig = 0.2 \\
        0.8 & \typesig = 0.8
    \end{cases}
\end{align*}
Furthermore, let buyers' common prior be $\prior = 0.5$ and the seller's reservation value be $c = 0.6$. For any number of buyers, the equilibrium of this game is unique:\footnote{This follows a simple argument. No buyer can accept trade upon the low signal, since $\interim^* \leq 0.5$, so $\Prob_{\interim^*} \left( \quality = H \mid 0.2 \right) \leq 0.2 < c$. As Lemma \ref{lem:poolbelief_approval} shows, the interim belief is strictly decreasing in the probability that buyers approve upon the high signal. Thus, buyers either always accept trade at this signal, or there is a unique interior probability of acceptance that---given the interim belief consistent with those strategies---buyers are indifferent to accept.} buyers always reject upon the low signal, $\strat_n^{*} (0.2) = 0$, but accept with some probability upon the high signal, $\strat_n^{*} (0.8) \in [0,1]$.

Figures \ref{fig:sigma} through \ref{fig:notradebelief} plot (a) total surplus in the market, and the probabilities that (b) a buyer trades upon receiving the high signal, (c and d) some buyer trades with the seller given his quality, (e and f) the probability that a seller has High quality given he trades with some buyer or no buyer, as the number of buyers rises from 1 to 50. 

\begin{center}
    % Begin group to apply local figure numbering
    \begingroup
    \setcounter{figure}{0} % Reset figure counter
    \renewcommand{\thefigure}{8.\alph{figure}} % Set numbering to 8.a, 8.b, etc.

    % First row with two figures

\begin{minipage}[t]{0.48\textwidth}
  \centering
  \includegraphics[width=\textwidth]{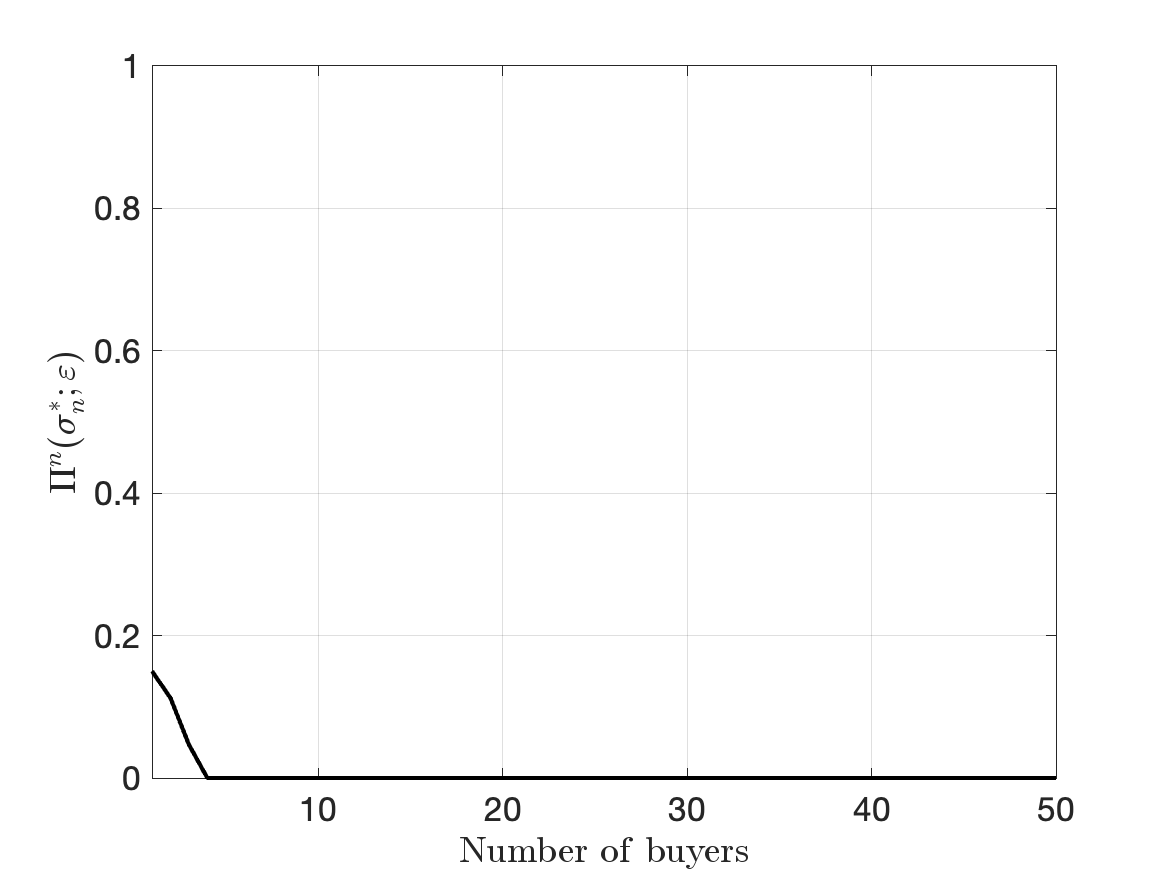}
    \captionof{figure}{Total surplus}
    \label{fig:surplus}
\end{minipage}
\hfill
\begin{minipage}[t]{0.48\textwidth}
  \centering
  \includegraphics[width=\textwidth]{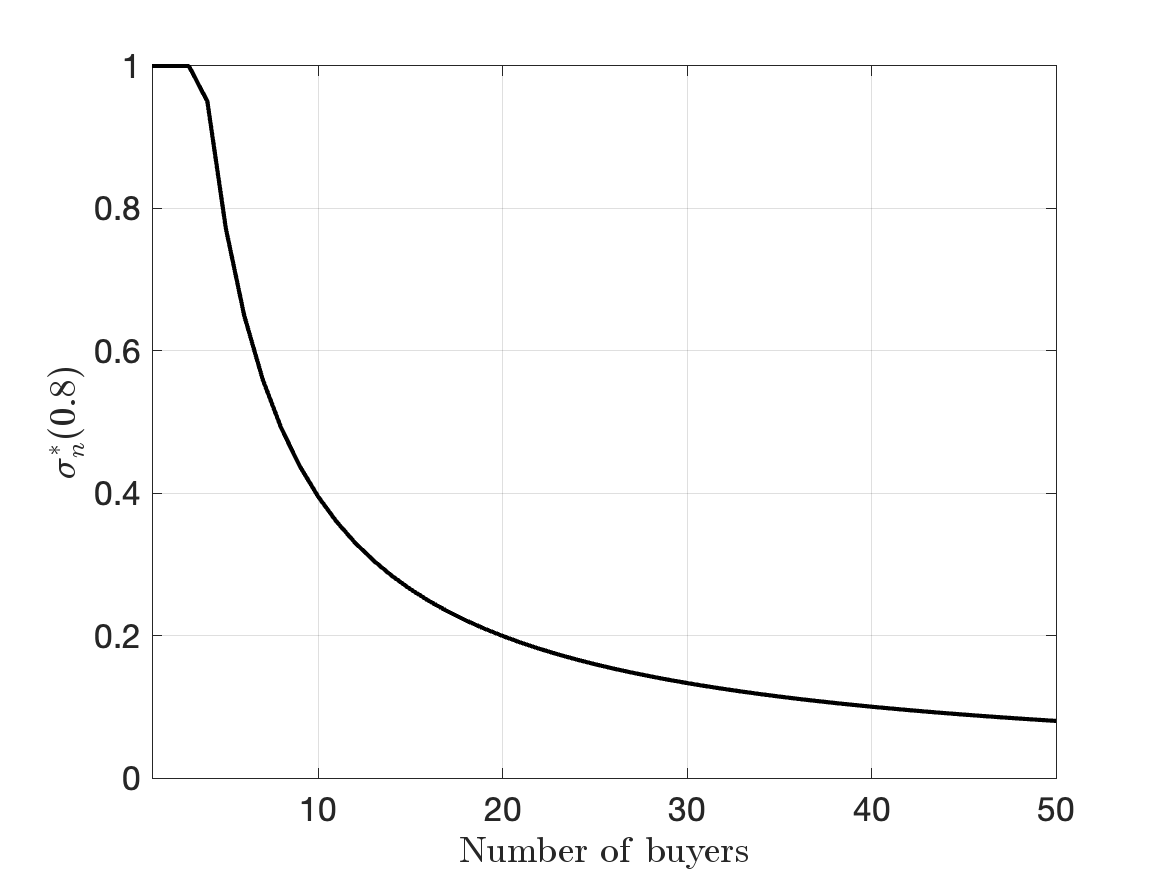}
    \captionof{figure}{Probability that \\ buyer accepts upon $\typesig = 0.8$}
    \label{fig:sigma}
\end{minipage}

\vspace{0.5cm}

% Second row with one figure spanning both columns
\begin{minipage}[t]{0.48\textwidth}
  \centering
  \includegraphics[width=\textwidth]{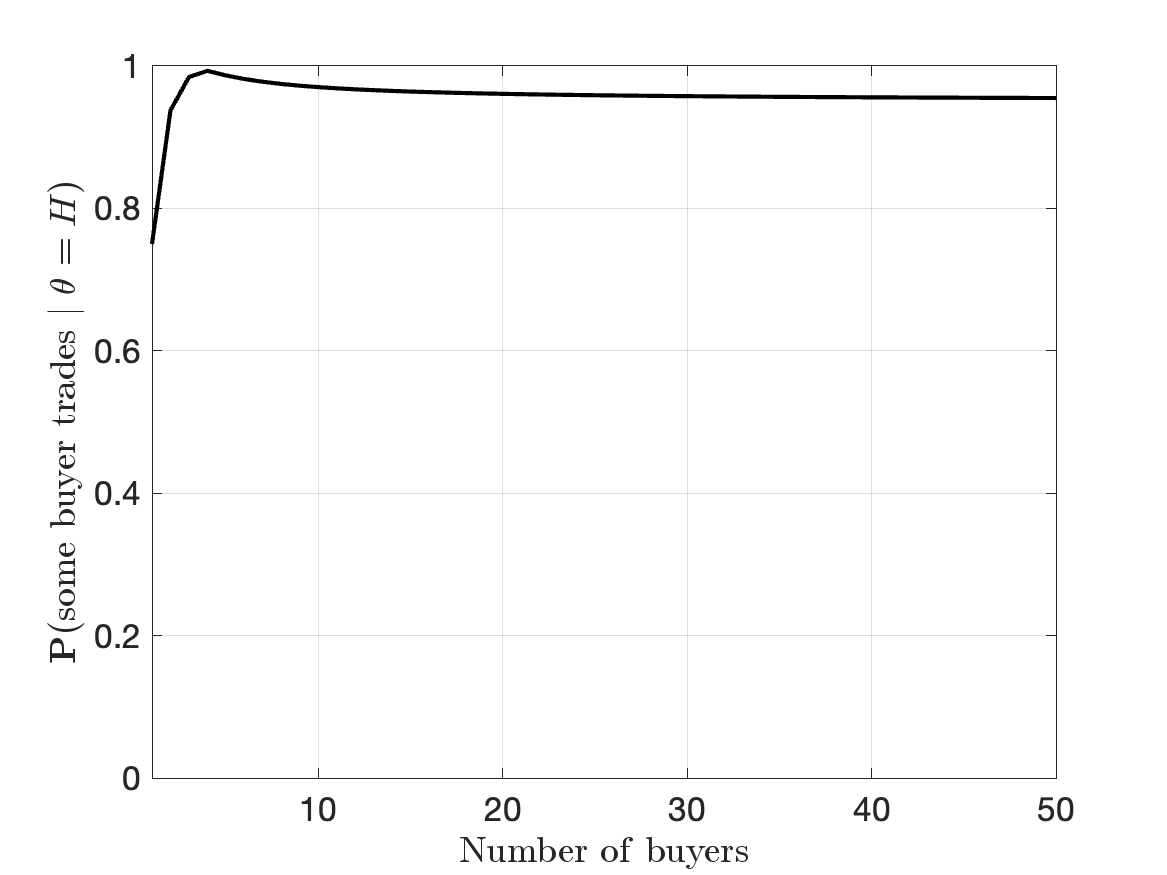}
    \captionof{figure}{Probability that \\ some buyer trades when $\quality = H$}
    \label{fig:ph}
\end{minipage}
\hfill
\begin{minipage}[t]{0.48\textwidth}
  \centering
  \includegraphics[width=\textwidth]{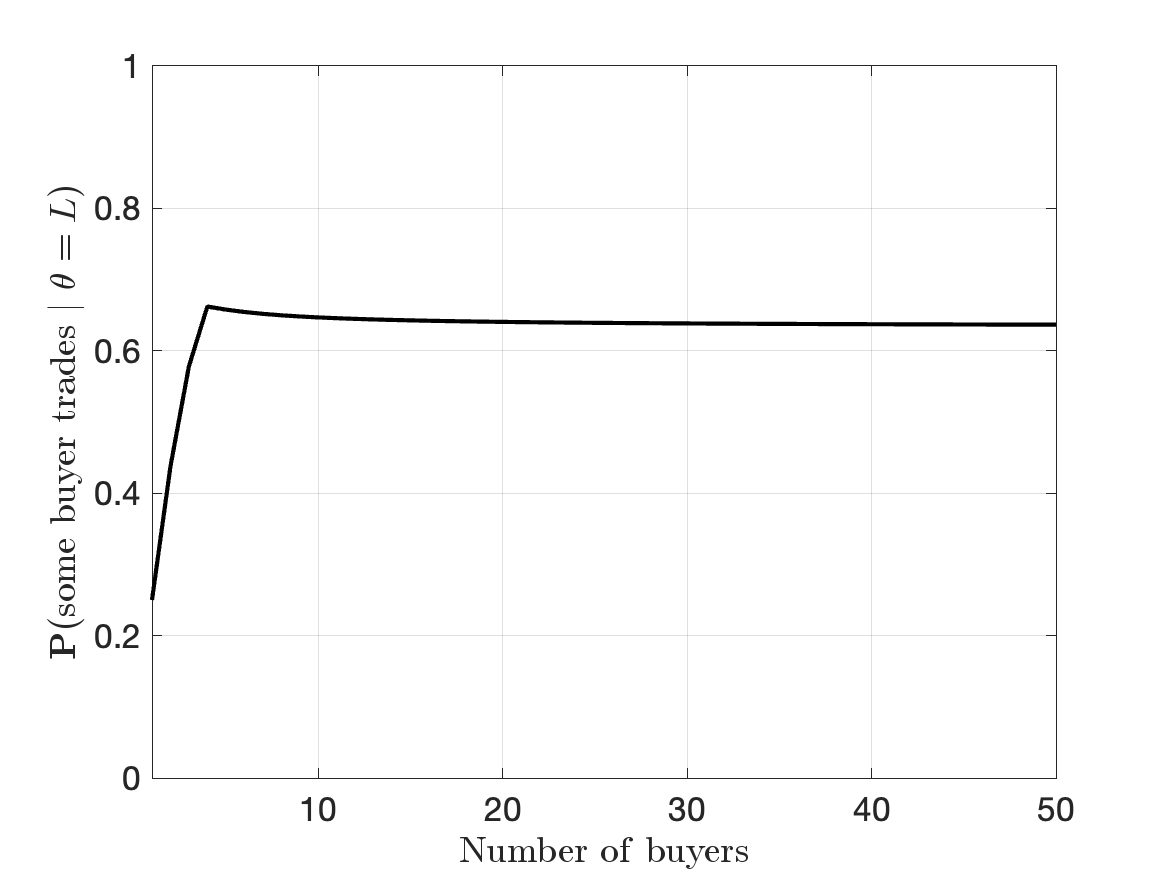}
    \captionof{figure}{Probability that \\ some buyer trades when $\quality = L$}
    \label{fig:pl}
\end{minipage}

\vspace{0.5cm}

% Third row with two figures
\begin{minipage}[t]{0.48\textwidth}
  \centering
  \includegraphics[width=\textwidth]{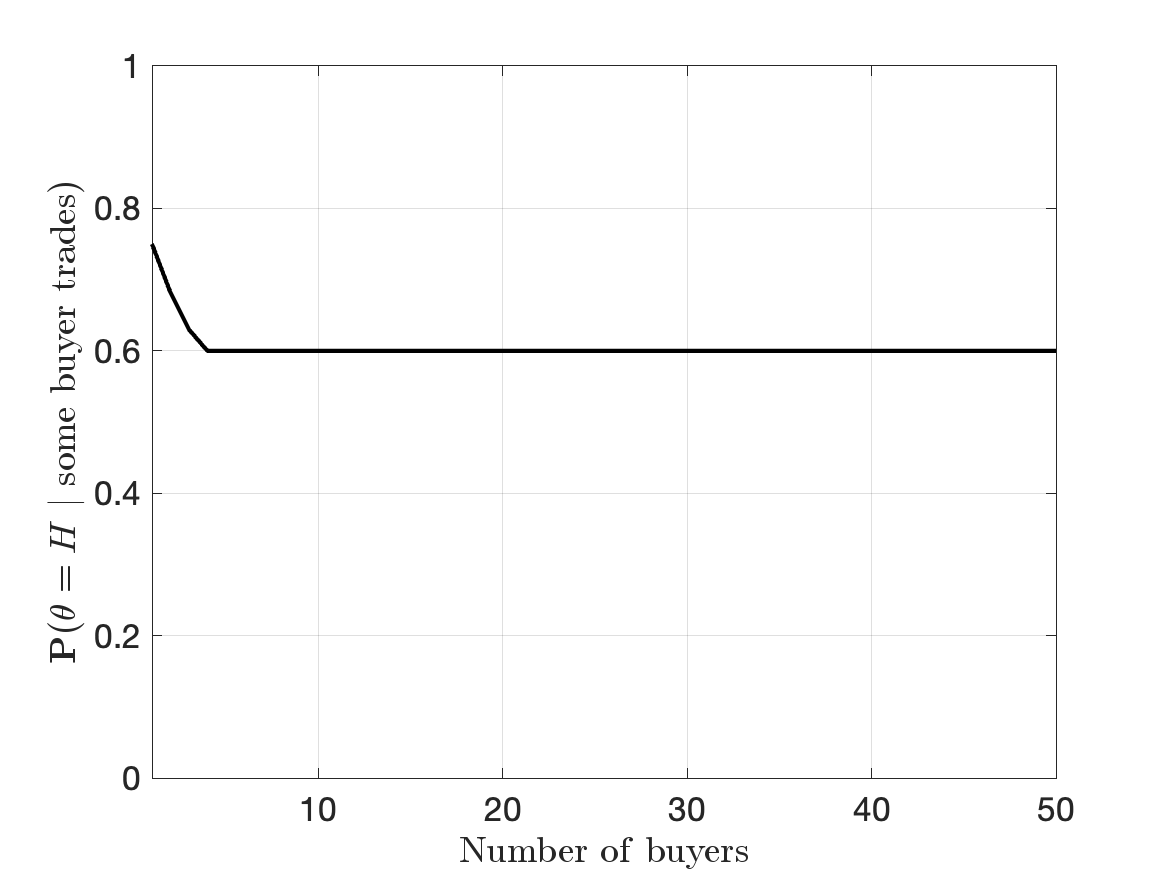}
    \captionof{figure}{Probability that \\ $\quality = H$ given some buyer trades}
    \label{fig:tradebelief}
\end{minipage}
\hfill
\begin{minipage}[t]{0.48\textwidth}
  \centering
  \includegraphics[width=\textwidth]{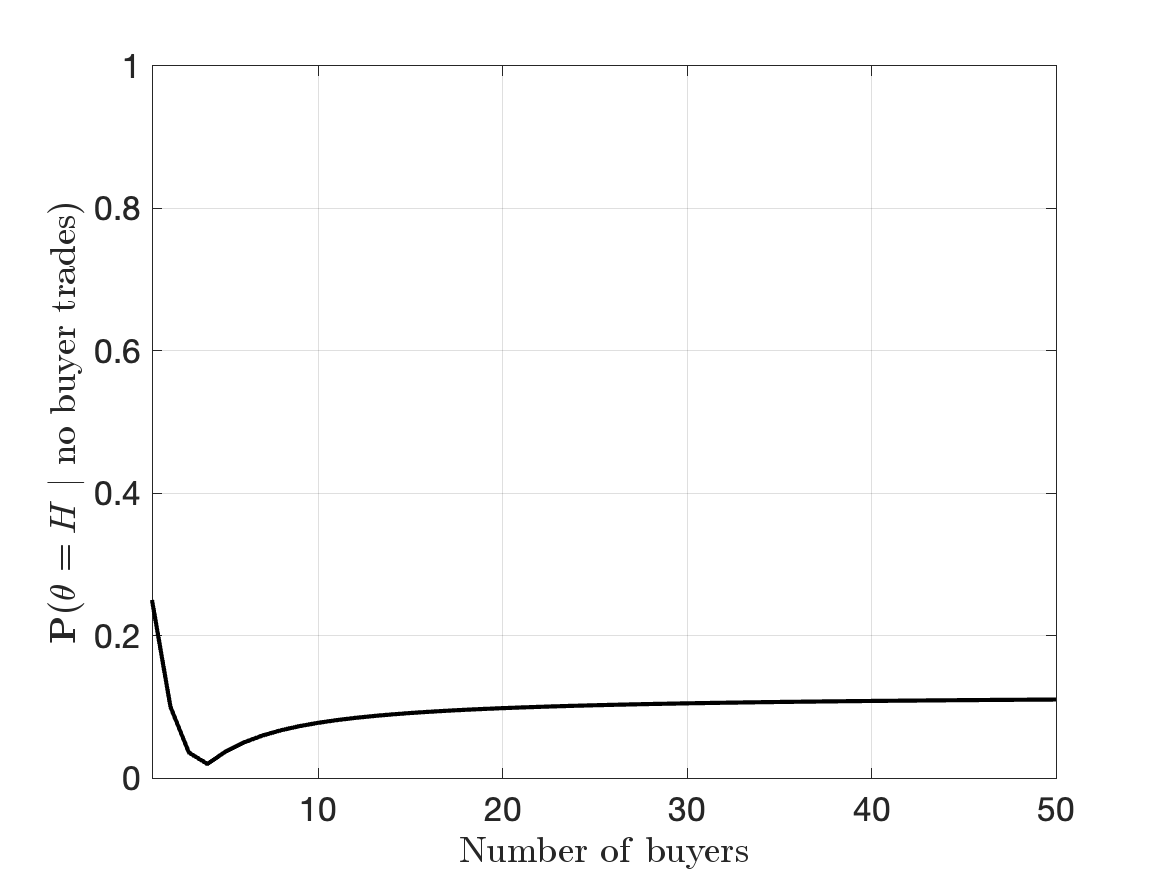}
        \captionof{figure}{{Probability that \\ $\quality = H$ given no buyer trades}}
        \label{fig:notradebelief}
\end{minipage}

    \endgroup
\end{center}

\vspace{0.5cm}

Once the number of buyers in the market exceeds 3, total surplus in the market equals $0$---surplus in the no-information benchmark. A buyer only trades if she receives the high signal. Even then, she rejects the seller with a strictly positive probability. This probability increases with the number of buyers in the market. However, the probability that some buyer trades with the seller reaches a constant level: $0.95$ if he has High quality, and $0.63$ if he has Low quality. Despite having no bearing on market participants' surplus, trade is informative about the seller's quality: the probability that the seller has High quality conditional on trading is $0.6$ not $\prior = 0.5$; and the probability she has High quality conditional on not trading reaches $0.11$.\footnote{For clarity, the figures illustrate results as $n \to 50$; however, these asymptotic values remain valid as $n \to 1000$. As $n \to 1000$, the probability that a buyer accepts trade upon a high signal converges to 0.}

\newpage

\section{Proof Appendix}
\label{section:proofs}

\subsection{Useful Definitions and Notation}
\label{section:useful_definitions}

In what follows, I occasionally operate with the likelihood ratios of beliefs for convenience. The reader can easily verify the identities:
\begin{align*}
    \frac{\interim}{1 - \interim} &=
    \frac{\prior}{1 - \prior} 
    \times
    \frac{ 
        \nu_{H}\left( \strat; \experiment \right) 
    }{
        \nu_{L}\left( \strat; \experiment \right) 
    }
    &
    \frac{
        \Prob_{\interim} \left( \quality = H \mid  s_i \right)
    }{
        1 - \Prob_{\interim} \left( \quality = H \mid  s_i \right)
    }
    &=
    \frac{\interim}{1 - \interim} \times
    \frac{s_i}{1 - s_i}
\end{align*}
Through similar reasoning, the reader can verify that it is optimal for a buyer to accept trade when:
\begin{equation*}
    \frac{
        \Prob_{\interim} \left( \quality = H \mid s_i \right)
    }{
        1 - \Prob_{\interim} \left( \quality = H \mid s_i \right)
    }
    >
    \frac{c}{1-c}
\end{equation*}

Some strategies require buyers to randomise upon observing a particular outcome. To facilitate technical discussion, where it is warranted I assume that each buyer observes the realisation of a \textit{tie-breaking signal} $u \sim U[0,1]$ alongside the outcome of her experiment. This signal is not informative about the asset's quality: it is distributed independently from it conditional on the experiment's outcome. 
I denote the outcome of buyer $i$'s experiment as $s^i$ and her tie-breaking signal as $u^i$. Without loss, buyer $i$ accepts trade if and only if $\strat (s^i) \leq u^i$; where $\strat$ is her strategy. I call the pair  $(s^i, u^i)$ the \textit{score} buyer $i$ observes for the seller. 

\begin{defn}
    \label{defn:score}
    The tuple $Z^i = (\typesig^i, u^i)$, where $u^i \overset{IID}{\sim} U[0,1]$
    is the \textit{score} buyer $i$ observes for the seller. The seller's \textit{score profile} $\mathbf{z}$ is the set of scores each buyer observes; $\mathbf{z} = \{ (\typesig^i, u^i) \}_{i=1}^{n}$. Analogously, the seller's \textit{signal profile} $\mathbf{s} = \{ s^i \}_{i=1}^{n}$ is the set of outcomes of each buyer's experiment.
\end{defn}

Some proofs in Section \ref{section:omitted_proofs} require comparing interim beliefs across pairs of strategies and experiments; $(\strat, \experiment)$. For convenience, I define the mapping from such a pair to the interim belief consistent with them as 
$\Psi (.; \experiment): [0,1]^n \to [0,1]$:
    \begin{equation*}
        \Psi \left( \strat; \experiment \right) 
        :=
        \frac{
            \prior \times \nu_H \left( \strat; \experiment \right)
        }{
            \prior \times \nu_H \left( \strat; \experiment \right)
            +
            ( 1 - \prior ) \times \nu_L \left( \strat; \experiment \right)
        }
    \end{equation*}

    Wherever necessary, I treat each strategy $\strat: \signalset \to [0,1]$ for an experiment $\experiment$ as a vector in the compact set $[0,1]^m \subset \mathbb{R}^n$. This is a finite dimensional vector space, so I endow it with the metric induced by the taxicab norm without loss of generality (see \textcite{kreyszig_functionalanalysis} Theorem 2.4-5):
    \begin{align*}
        \lvert \lvert \strat' - \strat \rvert \rvert &= \sum\limits_{j=1}^{m} \lvert \strat' (s_j) - \strat (s_j) \rvert
        &
        \textrm{for any two strategies } \strat' \textrm{ and } \strat
    \end{align*}
    Note that the interim belief function $\Psi(.;\experiment)$ is thus a continuous function of buyers' strategies.\footnote{$r_{\quality}$ is continuous in $\strat$; so, both the nominator and denominator are strictly positive continuous in $\strat$.}

\begin{defn}
    \label{defn:mute}
    Where the experiment $\experiment$ is binary, $s_L^{\textrm{mute}}$ is the \textit{strongest} level of bad news for which there is an equilibrium where a buyer trades regardless of her signal:
	    \begin{equation*}
    	    \frac{\prior}{1 - \prior} \times 
        		\frac{s_L^{\textrm{mute}}}{1 - s_L^{\textrm{mute}}}
	        =
    	    \frac{c}{1-c}
	    		\end{equation*}
\end{defn}

\subsection{Omitted Results}
\label{section:omitted_results}

\begin{restatable}{lem}{lemselectivebetter}
        \label{lem:selectivebetter}
        Let $\strat^*$ and $\strat$ be two monotone strategies such that (i) $\strat^*$ is more selective than $\strat$, and (ii) $\strat^*$ is an equilibrium strategy. Then, total surplus is higher under $\strat^*$: $\sumpayoff (\strat^*; \signalstr) \geq \sumpayoff (\strat; \signalstr)$.
\end{restatable}

\begin{proof}
    \label{proof:prop_selectivebetter}
    Let $\mathbf{z}$ be the seller's \textit{score profile}. Take an equilibrium strategy $\strat^*$ and a less selective strategy $\strat$ such that:
    \begin{equation*}
        \strat(s) - \strat^*(s) =
        \begin{cases}
            \varepsilon & s =  \underaccent{\bar}{s} \\
            0 & \textrm{otherwise}
        \end{cases}
    \end{equation*}
    for some $\varepsilon > 0$, where $\underaccent{\bar}{s} := \min \{ s \in S: \strat^*(s) < 1 \}$. I show that:
    \begin{equation*}
        \lim\limits_{\varepsilon \to 0} 
        \sumpayoff \left( \strat; \varepsilon \right) - \sumpayoff \left( \strat^*; \varepsilon \right) \leq 0
    \end{equation*}
    By Lemma \ref{lem:overapprove}, this establishes the result. 
    
    Now, let $Z \subset (S \times [0,1])^n$ be the set of score profiles under which some buyer trades under $\strat$, but all buyers reject the seller under $\strat^*$:
    \begin{equation*}
        \mathbf{z} \in Z \iff \hspace{0.3cm}
        \begin{aligned}
            \strat^*(s^i) > u^i \quad &\text{for all } i \in \{1, 2, \ldots, n\}, \\
            &\text{and} \\
            \strat(s^i) \leq u^i \quad &\text{for some } i \in \{1, 2, \ldots, n\}.
        \end{aligned}
    \end{equation*}
    Furthermore, for a given score profile $\mathbf{z}$, 
    let $\#$ be the number of buyers whose observed scores are such that $\strat (s^i) \geq u^i > \strat^* (s^i)$. These buyers would accept trade under the strategy $\strat$, but not under $\strat^*$. 

    The seller's eventual outcome differs between the strategy profiles $\strat$ and $\strat^*$ if and only if his score profile $\mathbf{z}$ lies in $Z$. Furthermore, his eventual outcome can only change from a rejection by all buyers in $\strat^*$ to an approval by some buyer in $\strat$. Thus:
    \begin{align*}
        \sumpayoff (\strat; \signalstr) - \sumpayoff (\strat^*; \signalstr) 
        &= \left[ \Prob ( \quality = H \mid \mathbf{z} \in Z ) - c \right] \times \Prob ( \mathbf{z} \in Z )
        \\
        &\propto \Prob ( \quality = H \mid \mathbf{z} \in Z ) - c
    \end{align*}
    Focus therefore, on the probability that $\quality = H$ given the seller's signal profile lies in $Z$:
    \begin{equation*}
        \Prob \left( \quality = H \mid \mathbf{z} \in Z \right) = 
        \sum\limits_{i=1}^{n}
        \Prob \left( \quality = H \mid \# = i \right) \times 
        \frac{
            \Prob \left( \# = i \right)
        }{
            \Prob (\mathbf{z} \in Z)
        }
    \end{equation*}
    Now note:
    \begin{equation*}
        \Prob \left( \# = i \mid \quality \right) = \left( p_{\quality}(\underaccent{\bar}{s}) \right)^i \times \left( 1 - p_{\quality}(\underaccent{\bar}{s}) \right)^{n-i} \times \varepsilon^i
    \end{equation*}
    and thus $\Prob \left( \# = i \right) \propto \varepsilon^i$. Since $\Prob (\mathbf{z} \in A) = \sum\limits_{i=1}^{n} \Prob \left( \# = i \right)$, we have $\lim\limits_{\varepsilon \to 0} \frac{
        \Prob \left( \# = i \right)
    }{
        \Prob (\mathbf{z} \in A)
    } = 0$ for any $i > 1$. Thus:
    \begin{equation*}
        \lim\limits_{\varepsilon \to 0} \Prob \left( \quality = H \mid \mathbf{z} \in A \right) - \Prob \left( \quality = H \mid \# = 1 \right) = 0
    \end{equation*}
    I conclude the proof by showing that $\Prob \left( \quality = H \mid \# = 1 \right) \leq c$ as $\varepsilon \to 0$:
    \begin{align*}
        \lim\limits_{\varepsilon \to 0}
        \frac{
            \Prob \left( \quality = H \mid \# = 1 \right)
        }{
            \Prob \left( \quality = L \mid \# = 1 \right)
        }
        &=
        \lim\limits_{\varepsilon \to 0}
        \frac{ \Prob \left( \quality = H \right) }{ \Prob \left( \quality = L \right) } \times
        \frac{
            \Prob \left( \# = 1 \mid \quality = H \right)
        }{
            \Prob \left( \# = 1 \mid \quality = L \right)
        }
        %\\[0.5pt]
        \\
        &=
        \lim\limits_{\varepsilon \to 0}
        \frac{ \Prob \left( \quality = H \right) }{ \Prob \left( \quality = L \right) } \times
        \left( 
            \frac{
                r_H ( \strat; \signalstr )
            }{
                r_L ( \strat; \signalstr )
            }
        \right)^{n-1}
        \times
        \frac{ p_H (\underaccent{\bar}{s}) }{ p_L (\underaccent{\bar}{s}) }
        %\\[0.5pt]
        \\
        &=
        \frac{ \Prob \left( \quality = H \right) }{ \Prob \left( \quality = L \right) } \times
        \left( 
            \frac{
                r_H ( \strat^*; \signalstr )
            }{
                r_L ( \strat^*; \signalstr )
            }
        \right)^{n-1}
        \times
        \frac{ p_H (\underaccent{\bar}{s}) }{ p_L (\underaccent{\bar}{s}) }
        %\\[0.5pt]
        \\
        &\leq
        \frac{ \interim^* }{ 1 - \interim^* } \times \frac{ p_H (\underaccent{\bar}{s}) }{ p_L (\underaccent{\bar}{s}) }
        \leq
        \frac{c}{1-c}
    \end{align*}
    where $\interim^* = \interimfcn ( \strat^*; \signalstr )$ is the interim belief consistent with $\strat^*$. The penultimate inequality holds due to the straightforward fact that:
    \begin{equation*}
        \frac{
            \interim^*
        }{
            1 - \interim^*
        }
        =
        \frac{\prior}{1 - \prior}
        \times
        \frac{
            1 + r_H^* + ... + (r_H^*)^{n-1}
        }{
            1 + r_L^* + ... + (r_L^*)^{n-1}
        }
        \leq
        \frac{\prior}{1 - \prior} \times
        \left( 
            \frac{r_H^*}{r_L^*}
        \right)^{n-1}
    \end{equation*}
    where $r_{\quality}^* := r_{\quality} (\strat^*; \signalstr)$. The last inequality is due to the fact that $\underaccent{\bar}{s} \in S$ is optimally rejected under $\strat^*$.
    
\end{proof}

\begin{lem}
    \label{lemma:blackwell_necessity}
    Suppose there is a single buyer, $n = 1$. Equilibrium total surplus under experiment $\experiment'$ exceeds that under $\experiment$ \textit{regardless of the seller's reservation value} $c \in [0,1]$ \textit{and buyer's prior belief} $\prior \in [0,1]$ if and only if $\experiment'$ is (Blackwell) more informative than $\experiment$.
\end{lem}

\begin{proof}
    The sufficiency part of this Lemma follows from Blackwell's Theorem (\textcite{blackwell_girshick}, Theorem 12.2.2). To show necessity, I fix an arbitrary prior belief $\prior$ for the evaluator. 

    Let $q_j$ be the posterior belief the buyer forms about the asset's quality upon observing the outcome $s_j \in \signalset$:
    \begin{equation*}
        q_j = 
        \frac{
            \prior \times s_j
        }{
            \prior \times s_j + (1 - \prior) \times (1 - s_j)
        }
    \end{equation*}
    Furthermore, let $F(.)$ and $F'(.)$ be the CDFs of the posterior beliefs $\experiment$ and $\experiment'$ induce, respectively, for this prior belief $\prior$:
    \begin{align*}
        F (q) &= (1 - \prior) \times \sum\limits_{s \in \signalset: s \leq q} p_L(s)
        +
        \prior \times \sum\limits_{s \in \signalset: s \leq x} p_H(s)
        \\
        F'( q ) &= (1 - \prior) \times \sum\limits_{s \in \signalset: s \leq q} p_L'(s)
        +
        \prior \times \sum\limits_{s \in \signalset: s \leq x} p_H'(s)
    \end{align*}
    Equilibrium total surplus (and the buyer's expected payoff) under $\experiment$ is given by:
    \begin{align*}
        \int\limits_{c}^{1} ( q - c ) dF( q ) = \int\limits_{c}^{1} q dF( q ) - c \times \left( 1 - F(c) \right) = ( 1 - c ) - \int\limits_{c}^{1} F( q ) d q
    \end{align*}
    An analogous expression gives equilibrium total surplus under $\experiment'$. For the former to exceed the latter for any $c \in [0,1]$, we must have:
    \begin{equation*}
        \int\limits_{c}^{1} \left( F(q) - F'(q) \right) d q 
        \geq 0
    \end{equation*}
    which is equivalent to $\experiment'$ being Blackwell more informative than $\experiment$.\footnote{See \textcite{muller_stoyan_comparison}, Theorem 1.5.7. The Blackwell order between signal structures is equivalent to the convex order between the posterior belief distributions they induce; see \textcite{gentzkow_kamenica_2016}.}
    
\end{proof}

Lemma \ref{lem:eqm_algorithm} proves useful when proving Theorem \ref{thm:intensive}, the main result of Section \ref{section:finite}. This Lemma can also be used towards an alternative and direct proof for the equilibrium existence claim of Proposition \ref{prop:eqmexist}.

\begin{lem}
    \label{lem:eqm_algorithm}
    For each $j \in \{ 1,2,...,m \}$, let $\strat_j$ be the strategy defined as:
    \begin{equation*}
        \strat_j (s) = 
        \begin{cases}
            0 & s < s_j \\
            1 & s \geq s_j
        \end{cases}
    \end{equation*}
    and $\interim_j$ be the interim belief consistent with this strategy. Unless $\strat_j$ is itself an equilibrium strategy:
    \begin{enumerate}[label = \roman*]
        \item There is an equilibrium strategy $\strat^*$ that is \textit{more selective than} $\strat_j$ if $\Prob_{\interim_j} \left( \quality = H \mid s_j \right) < c$.
        \item There is an equilibrium strategy $\strat^*$ that is \textit{less selective than} $\strat_j$ if $\Prob_{\interim_j} \left( \quality = H \mid s_{j-1} \right) > c$.
    \end{enumerate}
\end{lem}

\begin{proof}
    Abusing notation slightly, I add two fully revealing outcomes $s_0$ and $s_{m+1}$ to the set $\signalset$ (duplicating $s_1$ and $s_m$ if either of them are already fully revealing), and denote the strategy which \textit{never} accepts trade as  $\strat_{m+1}$:
    \begin{align*}
        \frac{ s_{m+1} }{ 1 - s_{m+1} } &= \infty
        &
        \frac{s_0}{ 1 - s_0 } &= 0
        &
        \frac{ \interim_{m+1} }{ 1 - \interim_{m+1} } = 
        %\frac{ \interim_0 }{ 1 - \interim_0 } = 
        \frac{\prior}{1 - \prior}
    \end{align*}
    
    \par\noindent
    \hspace*{0cm}%
    \textbf{Claim i.}
    \par\vspace{\medskipamount}

    \noindent
    The strategy $\strat_{m+1}$ is the most selective strategy buyers can adopt, and is an equilibrium strategy unless:
    \begin{equation*}
        \frac{
            s_m
        }{
            1 - s_m
        } \times
        \frac{\interim_{m+1}}{ 1 - \interim_{m+1}} > \frac{c}{1-c}
    \end{equation*}
    So, assume this condition is satisfied. Likewise, the strategy $\strat_k$ for $k \geq j$ is an equilibrium if the following inequality is satisfied:
    \begin{equation}
        \label{eqn:algorithm_ineq1}
        \frac{s_k}{1 - s_k} \times \frac{\interim_k}{1 - \interim_k} \geq \frac{c}{1-c} \geq \frac{s_{k-1}}{ 1 - s_{k-1} } \times \frac{\interim_k}{1 - \interim_k}
    \end{equation}
    So, assume inequality \ref{eqn:algorithm_ineq1} is violated for every $k \geq j$. This gives us:
    \begin{equation}
        \label{eqn:algorithm_ineq2}
        \frac{s_{m+1}}{1 - s_{m+1}} \times
        \frac{\interim_{m+1}}{ 1 - \interim_{m+1}} > \frac{c}{1-c} > \frac{s_j}{1 - s_j} \times \frac{\interim_j}{1 - \interim_j}
    \end{equation}
    where the last part of this inequality is by hypothesis. 

    Now, let $k^* \in \{ j, j+1, ..., m \}$ be the first index for which the following inequality is satisfied:
    \begin{equation*}
        \frac{s_{k^*+1}}{1 - s_{k^*+1}} \times
        \frac{\interim_{{k^*}+1}}{ 1 - \interim_{{k^*}+1}} \geq \frac{c}{1-c} \geq \frac{s_{k^*}}{1 - s_{k^*}} \times \frac{\interim_{k^*}}{1 - \interim_{k^*}}
    \end{equation*}
    such a $k^*$ must exist due to inequality \ref{eqn:algorithm_ineq2}. 
    But since inequality \ref{eqn:algorithm_ineq1} is violated, we must have:
    \begin{equation*}
        \frac{s_{k^*}}{1 - s_{k^*}} \times \frac{\interim_{k^* +1}
        }{1 - \interim_{k^* +1}}
        >
        \frac{c}{1-c}
        \geq
        \frac{s_{k^*}}{1 - s_{k^*}} \times
        \frac{\interim_{k^*}}{1 - \interim_{k^*}}
    \end{equation*}
    But since the function $\interimfcn (\strat; \experiment)$ is continuous in buyers' strategy $\strat$,
    we can then find some strategy $\strat^*$:
    \begin{equation*}
        \strat^* (s) =
        \begin{cases}
            1 & s > s_{k^*}  \\
            \in [0,1] & s = s_{k^*} \\
            0 & s < s_{k^*}
        \end{cases}
    \end{equation*}
    such that:
    \begin{equation*}
        \frac{s_{k^*+1}}{1 - s_{k^*+1}} \times \frac{\interimfcn (\strat^*; \experiment)
        }{1 - \interimfcn (\strat^*; \experiment)}
        \geq
        \frac{c}{1-c}
        =
        \frac{s_{k^*}}{1 - s_{k^*}} \times
        \frac{\interimfcn (\strat^*; \experiment)}{1 - \interimfcn (\strat^*; \experiment)}
    \end{equation*}

    The strategy $\strat^*$ is thus an equilibrium strategy. It is clearly more selective than $\strat_j$; since it is more selective than $\strat_{k^*}$, where $k^* \geq j$.

    \par\noindent
    \hspace*{0cm}%
    \textbf{Claim ii.}
    \par\vspace{\medskipamount}

    \noindent
    For any $k \in \{ 1,2,...,j \}$, the strategy $\strat_k$ is an equilibrium if the inequality \ref{eqn:algorithm_ineq1} is satisfied. So, as earlier, assume \ref{eqn:algorithm_ineq1} is violated for every such $k$. Then, we get:
    \begin{equation*}
        \frac{s_j}{1 - s_j}
        \times
        \frac{\interim_j}{1 - \interim_j} 
        \geq
        \frac{s_{j-1}}{1 - s_{j-1}} \times \frac{\interim_j}{1 - \interim_j}
        >
        \frac{c}{1-c}
        >
        \frac{s_1}{1 - s_1}
        \times
        \frac{\interim_1}{1 - \interim_1}
    \end{equation*}
    where the second inequality in the chain follows by hypothesis
    and the last inequality follows from the violation of inequality \ref{eqn:algorithm_ineq1} for $k = 1$. We can now repeat the argument we constructed after inequality \ref{eqn:algorithm_ineq2} to prove Claim i, to prove the existence of an equilibrium strategy $\strat^*$ that is less selective than $\strat_j$.

\end{proof}

\subsection{Omitted Proofs}
\label{section:omitted_proofs}

\propeqmexist*

\begin{proof}
    \label{proof:prop_eqmexist}

    In what follows, I treat each strategy $\strat: \signalset \to [0,1]$ as a vector in the compact set $[0,1]^m \subset \mathbb{R}^n$, endowed with the taxicab metric (see the end of Section \ref{section:useful_definitions}).
    I start by proving that any equilibrium strategy must be monotone and all equilibria exhibit adverse selection. Using these observations, I prove that the set of equilibrium strategies is non-empty and compact. 

    \begin{enumerate}
        \setcounter{enumi}{1}
        \item Any equilibrium strategy is {monotone.}
    \end{enumerate} 

    \noindent
    Any equilibrium strategy $\strat^*$ must be optimal against the interim belief $\interim^*$ consistent with it. Whenever $\prior \in (0,1)$, $\interim^* = \interimfcn \left( \strat^*; \experiment \right) \in (0,1)$, and so 
    $\Prob_{\interim^*} \left( \quality = H \mid s' \right) > \Prob_{\interim^*} \left( \quality = H \mid S = s \right)$ for $s', s \in \signalset$ such that $s' > s$. 

    \begin{enumerate}
        \setcounter{enumi}{2}
        \item {All equilibria exhibit adverse selection.}
    \end{enumerate}

    \noindent
    \textit{A fortiori}, $\Psi (\strat; \experiment) \leq \prior$ for any monotone strategy $\strat$. To see this, note that $p_H(.)$ first order stochastically dominates $p_L(.)$ since it likelihood ratio dominates it.\footnote{Theorem 1.C.1 in \textcite{shaked_shantikumar_2007}.}
    Therefore, $\nu_L (\strat; \experiment) \geq \nu_H (\strat; \experiment)$.
    The result then follows since $\frac{
        \interimfcn (\strat; \experiment)
    }{
        1 - \interimfcn (\strat; \experiment)
    } = \frac{\prior}{1 - \prior} \times \frac{\nu_H (\strat; \experiment)}{\nu_L (\strat; \experiment)}$.

    \begin{enumerate}
        \item The set of equilibrium strategies is non-empty and compact.
    \end{enumerate}

    \begin{enumerate}[label=\roman*]
        \item The set of equilibrium strategies is non-empty.
    \end{enumerate}

    \noindent 
    Define $\Phi(.): [0,1]^m \to 2^{[0,1]^m}$ to be the buyers' \textit{best response correspondence}. $\Phi(.)$ maps any strategy $\strat$ to the set of strategies that are optimal against the interim belief $\Psi(\strat; \experiment)$ it induces:
    \begin{equation*}
        \Phi(\strat) = \left\{ \strat' \in [0,1]^m: \strat' \textrm{ is optimal against } \Psi(\strat; \experiment) \right\}
    \end{equation*}
    A strategy $\strat^*$ is an equilibrium strategy if and only if it is a fixed point of buyers' best response correspondence; $ \strat^* \in \Phi(\strat^*)$. I establish that the correspondence $\Phi$ has at least such fixed point through Kakutani's Fixed Point Theorem.

    $\Phi$ is trivially non-empty; every interim belief has some strategy optimal against it. It is also convex valued; if two distinct approval probabilities are optimal after some outcome $s \in \signalset$, \textit{any} approval probability is optimal upon that outcome.

    The only task that remains is to prove that $\Phi$ is upper-semi continuous. 
    For this, take an arbitrary sequence of strategies $\{ \strat_n \}$ such that $\strat_n \to \strat_{\infty}$. Denote the interim beliefs consistent with these strategies as $\interim_n := \Psi (\strat_n; \experiment)$. Since $\Psi(.;\experiment)$ is continuous in buyers' strategies, we also have $\interim_n \to \interim_{\infty}$ where $\interim_{\infty} = \Psi \left( \strat_{\infty}; \experiment \right)$. Now, take a sequence of strategies $\left\{\strat_n^* \right\}$ where $\strat_n^* \in \Phi (\strat_n)$. 
    Note that every $\strat_n^*$ is monotone since optimality against any interim belief $\interim \in (0,1)$ requires monotonicity. 
    We want to show that $\Phi$ is upper semicontinuous; i.e.:
    \begin{equation*}
        \strat_n^* \to \strat_{\infty}^* \implies \strat_{\infty}^* \in \Phi (\strat_{\infty})
    \end{equation*}

    By the Monotone Subsequence Theorem, the sequence $\left\{\strat_n^* \right\}$ has a subsequence $\strat_{n_k}^* \to \strat_{\infty}^*$ of strategies whose norms $\lvert \lvert \strat_{n_k}^* \rvert \rvert$ are monotone in their indices $n_k$. Here, I take the case where these norms are increasing, the proof is analogous for the opposite case. Since $\strat^*_{\infty}$ is the limit of a subsequence of monotone strategies, it must be a monotone strategy too. Assuming otherwise leads to a contradiction; for any $s, s' \in \signalset$ such that $s' > s$:
    \begin{equation*}
        \strat^*_{\infty}(s) > 0 \textrm{ \& } \strat^*_{\infty} (s') < 1 
        \hspace{0.2cm}
        \implies 
        \hspace{0.2cm}
        \exists N \in \mathbb{N} \textrm{ s.t. } \forall \textrm{ } n_k \geq N 
        \hspace{0.3cm}
        \strat_{n_k}^* (s) > 0 \textrm{ \& } \strat_{n_k}^* (s') < 1
    \end{equation*}

    Now let $\Bar{s}$ be the highest outcome for which $\strat_{\infty}^* (\Bar{s}) > 0$. I show that:
    \begin{itemize}
        \item If $\strat_{\infty}^* (\Bar{s}) \in (0,1)$, then:
        \begin{equation*}
            \frac{
                \interim_{\infty}
            }{
                1 - \interim_{\infty}
            }
            \times
            \frac{\Bar{s}}{1-\Bar{s}} = \frac{c}{1-c}
        \end{equation*}
        \item If $\strat_{\infty}^* (\Bar{s}) = 1$, then:
        \begin{equation*}
            \frac{
                \interim_{\infty}
            }{
                1 - \interim_{\infty}
            }
            \times
            \frac{s}{1-s}
            \hspace{0.2cm}
            \begin{cases}
                \leq \frac{c}{1-c} & s < \Bar{s} \\
                \geq \frac{c}{1-c} & s \geq \Bar{s}
            \end{cases}
        \end{equation*}
    \end{itemize}
    
    The first case easily follows by noting that:
    \begin{equation*}
        \strat_{\infty}^* (\Bar{s}) \in (0,1)
        \hspace{0.1cm}
        \implies
        \hspace{0.1cm}
        \strat_{n_k}^* (\Bar{s}) \in (0,1)
        \hspace{0.1cm}
        \implies
        \hspace{0.1cm}
        \frac{
                \interim_{n_k}
            }{
                1 - \interim_{n_k}
            }
            \times
            \frac{\Bar{s}}{1-\Bar{s}} = \frac{c}{1-c}
            \hspace{0.1cm}
        \implies
        \hspace{0.1cm}
        \frac{
                \interim_{\infty}
            }{
                1 - \interim_{\infty}
            }
            \times
            \frac{\Bar{s}}{1-\Bar{s}} = \frac{c}{1-c}
    \end{equation*}
    for all $n_k \geq N' \in \mathbb{N}$.
    The second case follows similarly, by noting that $\strat_{\infty}^* (\Bar{s}) = 1$ and $\strat_{\infty}^* (s') = 0$ for all $s' < \Bar{s}$ implies $\strat_{n_k}^* (\Bar{s}) > 0$ and $\strat_{n_k}^* (s') = 0$ for all $n_k \geq N'' \in \mathbb{N}$.

    \begin{enumerate}[label=\roman*]
        \setcounter{enumi}{1}
        \item The set of equilibrium strategies is compact. 
    \end{enumerate}

    $\Sigma$ is a subset of $[0,1]^m$ and therefore bounded, hence it suffices to show that is closed. Let $\left\{ \strat_n^{**} \right\}$ be a sequence of equilibrium strategies. Note that this means $\strat_n^{**} \in \Phi (\strat_n^{**})$. Since $\Phi(.)$ is upper semicontinuous, $\strat_n^{**} \to \strat_{\infty}$ implies $\strat_{\infty} \in \Phi (\strat_{\infty})$, and therefore $\strat_{\infty}$ is an equilibrium strategy itself. 
    
\end{proof}

\propselectivebetter*

\begin{proof}
    This is an immediate corollary to Lemmas \ref{lem:selectivebetter} and \ref{lem:overapprove} below; both of independent interest. 
\end{proof}

\begin{restatable}{lem}{lemoverapprove}
    \label{lem:overapprove}
    Take three monotone strategies $\strat'', \strat'$ and, $\strat$, ordered from the least selective to the most. If $\sumpayoff ( \strat'; \experiment ) \leq \sumpayoff ( \strat; \experiment )$, then $\sumpayoff ( \strat''; \experiment ) \leq \sumpayoff ( \strat'; \experiment )$.
\end{restatable}

\begin{proof}
    For the three strategies $\strat'', \strat'$, and $\strat$, consider three sets $Z, Z', Z'' \subset \left( S \times [0,1] \right)^n$ where the seller's score profile $\mathbf{z}$ might lie:
    \begin{equation*}
        \mathbf{z} \in
        \begin{cases}
            Z & \textrm{if } \mathbf{z} \textrm{ trades with some buyer under } \strat'' \textrm{ but not } \strat \\
            Z' & \textrm{if } \mathbf{z} \textrm{ trades with some buyer under } \strat' \textrm{ but not } \strat \\
            Z'' & \textrm{if } \mathbf{z} \textrm{ trades with some buyer under } \strat'' \textrm{ but not } \strat' 
        \end{cases}
    \end{equation*}
    Notice that $Z' \cap Z'' = \emptyset$ and $Z' \cup Z'' = Z$.
    We can write the difference between total surplus under different strategies as:
        \begin{equation*}
            \sumpayoff ( \strat'; \signalstr ) - \sumpayoff ( \strat; \signalstr )
            =
            \Prob \left( 
                \mathbf{z} \in Z'
            \right)
            \times
            \left[ 
                \Prob \left( 
                    \quality = H \mid
                    \mathbf{z} \in Z'
                \right)
                - c
            \right]
        \end{equation*}
        and:
        \begin{equation*}
            \sumpayoff ( \strat''; \signalstr ) - \sumpayoff ( \strat'; \signalstr )
            =
            \Prob \left( 
                \mathbf{z} \in Z''
            \right)
            \times
            \left[ 
                \Prob \left( 
                    \quality = H \mid
                    \mathbf{z} \in Z''
                \right)
                - c
            \right]
        \end{equation*}
        Therefore we want to prove that:
        \begin{equation*}
            \Prob \left( \quality = H \mid \mathbf{z} \in Z' \right) \leq c
            \hspace{0.1cm}
            \implies
            \hspace{0.1cm}
            \Prob \left( \quality = H \mid \mathbf{z} \in Z'' \right) \leq c
        \end{equation*}

        Now, note that $\Prob \left( \quality = H \mid \mathbf{z} \in Z \right)$ is a convex combination of $\Prob \left( \quality = H \mid \mathbf{z} \in Z' \right)$ and $\Prob \left( \quality = H \mid \mathbf{z} \in Z'' \right)$. Furthermore:
        \begin{equation*}
            \Prob \left( \quality = H \mid \mathbf{z} \in Z \right) \geq \Prob \left( \quality = H \mid \mathbf{z} \in Z \cap Z'' \right) = \Prob \left( \quality = H \mid \mathbf{z} \in Z'' \right)
        \end{equation*}
        which then implies:
        \begin{equation*}
            \Prob \left( \quality = H \mid \mathbf{z} \in Z'' \right) \leq \Prob \left( \quality = H \mid \mathbf{z} \in Z \right) \leq \Prob \left( \quality = H \mid \mathbf{z} \in Z' \right) \leq c
        \end{equation*}
        
\end{proof}

\thmextensive*

\begin{proof}
    For each $j \in \left\{ 1,2,...,m \right\}$, define $\strat_j$ to be the strategy:
    \begin{equation*}
        \strat_j (s) := 
        \begin{cases}
            0 & \typesig < \typesig_j \\
            1 & \typesig \geq \typesig_j
        \end{cases}
    \end{equation*}
    Moreover, let $F_{\quality} (x) = \sum\limits_{\typesig \in \signalset: \typesig \leq x } p_{\quality} (x) $. In a market with $n$ buyers, the interim belief $\interim_{j:n}$ consistent with buyers using the strategies $\strat_j$ is then implicitly given by:
    \begin{equation*}
        \frac{\interim_{ j:n }}{1 - \interim_{ j:n }} = \frac{ \sum_{k=0}^{n-1} F_H(s_{j-1})^k }{ \sum_{k=0}^{n-1} F_L(s_{j-1})^k }
    \end{equation*}
    Note that for all $j > 1$, the RHS is bounded and strictly decreasing in $n$, so the sequence 
    $\left\{ \interim_{ j:n } \right\}$ is convergent. 

    \noindent \textbf{Case 1:} $s_m = 1$.

    To prove the Theorem's statement for this case, I first show that $\interim_{m:n} \overset{n}{\to} 0$. 
    Let $X_n$ be the random variable that is uniformly distributed over the set $\left\{ F_H(s_{m-1})^k \right\}_{k=0}^{n-1}$. Then, note that:
    \begin{equation*}
        \frac{\interim_{ m:n }}{1 - \interim_{ m:n }} = \frac{ \sum_{k=0}^{n-1} F_H(s_{m-1})^k }{ \sum_{k=0}^{n-1} F_L(s_{m-1})^k } = \Exp \left[ X_n \right]
    \end{equation*}
    Now, fix any $x > 0$. Since $F_H(s_{m-1})^k$ is strictly decreasing in $k$, for any $\delta < 1$ of our choice, we can find some $N_{x;\delta} \in \mathbb{N}$ such that for all $n \geq N_{x;\delta}$ implies $\Prob \left( X_n \leq x \right) \geq \delta$ and $\Exp \left[ X_n \right] \leq \delta x + (1 - \delta)$. 
    Fixing $x = \frac{\varepsilon}{2 \delta}$ for some $\varepsilon > 0$ small, we have $\Exp \left[ X_n \right] \leq \frac{\varepsilon}{2} + 1 - \delta$. Since we can take $\delta$ arbitrarily close to 1, this shows that $\Exp \left[ X_n \right] \to 0$, proving this first claim. 

    Since buyers must always accept to trade upon observing $s_m = 1$ in equilibrium, this implies that there is some $N \in \mathbb{N}$ for which $\strat_m$ is the most selective equilibrium strategy for all $n \geq N$. So, for $n \geq N$,  a seller with a Low quality asset never trades. Moreover, every additional buyer increases the probability that a seller with a High quality asset trades. As $n \to \infty$, such a seller trades almost surely. We thus prove that $\sumpayoff^n \left( \hat{\strat}; \experiment \right) \to \sumpayoff^f$.

    \noindent \textbf{Case 2:} $s_m < 1$

    The case where buyers' experiment $\experiment$ is uninformative is trivial; it always yields the no-information benchmark. So, I assume that $s_m \neq s_1$. 
    
For any $j \in \left\{1,2,...,m\right\}$, the sequence $\left\{ \frac{ \interim_{j;n} }{ 1 - \interim_{j;n} } \times \frac{s_j}{1 - s_j} \right\}_{n=1}^{\infty}$ is bounded and monotone decreasing, thus convergent. Let $\mathcal{L}_j$ be the limit of this sequence. If $\mathcal{L}_m < \frac{c}{1-c}$, by Lemma \ref{lem:eqm_algorithm}, there is some $N' \in \mathbb{N}$ such that for all $n \geq N'$, the most selective equilibrium is more selective than $\strat_m$. So, buyers must be indifferent when they trade---expected trade surplus must be 0. Since total surplus is bounded below by $\sumpayoff^{\emptyset}$, we conclude that $\sumpayoff^n \left( \hat{\strat}; \experiment \right) = \sumpayoff^{\emptyset} = 0$ for all $n \geq N'$. 

    Now consider the case $\mathcal{L}_m \geq \frac{c}{1-c}$. Since $\interim_{j:n}$ is decreasing in $n$, the most selective equilibrium with $n$ buyers, $\hat{\strat}_n$, must get weakly more selective with $n$. So, the sequence $\left\{ r_{\quality} (\hat{\strat}_n; \experiment) \right\}_{n=1}^{\infty}$ is weakly increasing in $n$, convergent, and Cauchy. If for any $N \in \mathbb{N}$, $\hat{\strat}_N$ is more selective than $\strat_m$, we are done by the argument in the preceding paragraph. Otherwise, the sequence $\left\{ r_{\quality} (\hat{\strat}_n; \experiment) \right\}_{n=1}^{\infty}$ converges to a number below $1$.
    If are both constant at 0, we are done---total surplus is equal to that under the no-information benchmark along the sequence. Otherwise, $r_{L} (\hat{\strat}_n; \experiment) > r_{H} (\hat{\strat}_n; \experiment)$ along the sequence. Since both sequences are Cauchy, there exists some $N \in \mathbb{N}$ and $M \geq N$ such that for all $m \geq M$, we have:
    \begin{align*}
        &\prior \times (1 - c) \times \left[ 1 -  r_H (\hat{\strat}_m; \experiment) \right]^{m}  - ( 1 - \prior ) \times c \times \left[ 1 - r_L (\hat{\strat}_m; \experiment)^{m} \right] 
        \\
        \approx
        &\prior \times (1 - c) \times \left[ 1 -  r_H (\hat{\strat}_N; \experiment) \right]^{m}  - ( 1 - \prior ) \times c \times \left[ 1 - r_L (\hat{\strat}_N; \experiment)^{m} \right]
        \\
        >
        &\prior \times (1 - c) \times \left[ 1 -  r_H (\hat{\strat}_N; \experiment) \right]^{m+1}  - ( 1 - \prior ) \times c \times \left[ 1 - r_L (\hat{\strat}_N; \experiment)^{m+1} \right]
        \\
        \approx
        &\prior \times (1 - c) \times \left[ 1 -  r_H (\hat{\strat}_{m+1}; \experiment) \right]^{m+1}  - ( 1 - \prior ) \times c \times \left[ 1 - r_L (\hat{\strat}_{m+1}; \experiment)^{m+1} \right]
    \end{align*}
    This proves that total surplus is eventually decreasing. 

    Furthermore, since there is at least one outcome of the experiment $\experiment$ where a buyer surely trades with the seller, as $n \to \infty$, the seller trades almost surely regardless of quality. Hence, total surplus converges to $\prior - c$. Since total surplus can never be negative, it must be that $\sumpayoff^{\emptyset} = \prior - c$ in this case. 

\end{proof}

\thmintensivebinary*

I will use 
Lemmas \ref{lem:poolbelief_approval}, \ref{lem:binary_nolowmixing}, and \ref{lem:binary_eqmregions} below, of independent interest, to prove Theorem \ref{thm:intensive_binary}. Throughout, I denote the most and least selective equilibrium strategies under the experiment $\signalstr$ as $\hat{\strat_{\experiment}^*}$ and $\check{\strat_{\experiment}^*}$, respectively. I drop the subscript whenever the experiment in question is obvious. 

\begin{lem}
    \label{lem:poolbelief_approval}
    Let $\signalstr$ be a binary experiment, with outcomes in $\signalset = \{s_L, s_H\}$; $s_L \leq s_H$. 
    $\interimfcn(\strat; \signalstr)$ is:
    \begin{enumerate}[label=\roman*]
        \item strictly increasing in $\strat(s_L)$, whenever $\strat(s_H) = 1$,
        \item strictly decreasing in $\strat(s_H)$ whenever $\strat(s_L) = 0$.
    \end{enumerate}
\end{lem}

\begin{proof}
    \textit{Part i:}

    \noindent
    Let $\strat (s_L) \in (0,1)$ and $\strat (s_H) = 1$. 
    The interim belief $\interimfcn (\strat; \experiment)$ is then given by:
    \begin{align*}
        \interimfcn ( \strat; \signalstr ) &= 
        \Prob \left( \quality = H \mid
        \textrm{visit received} \right)
        \\
        &= \sum\limits_{i=0}^{n-1} \Prob ( \textrm{visited after i\textsuperscript{th} rejection} \mid \textrm{visit received} ) \times \Exp \left[ \quality = H \mid \textrm{i } s_L \textrm{ signals} \right]
        \\
        &= \sum\limits_{i=0}^{n-1} 
        \frac{
            \Prob ( \textrm{visited after i\textsuperscript{th} rejection})
        }{
            \Prob ( \textrm{visit received} )
        }
        \times \Exp \left[ \quality = H \mid \textrm{i } s_L \textrm{ signals}
        \right]
    \end{align*}

    \noindent
    Note that $\Exp \left[ \quality = H \mid \textrm{i } s_L \textrm{ signals} \right] < \Exp \left[ \quality = H \mid \textrm{i+1 } s_L \textrm{ signals} \right]$; since every $s_L$ signal is further evidence for $\quality = L$. We have:
    \begin{align*}
        \Prob ( \textrm{ visited after i\textsuperscript{th} rejection} ) = 
        &\Prob \left( 
            \textrm{buyer was (i+1)\textsuperscript{st} in order } \mid \textrm{seller got i rejections}
        \right)
        \\
        &\times
        \Prob \left( \textrm{seller got i rejections} \right)
        \\
        = &\frac{1}{n} \times \Prob ( \textrm{i } s_L \textrm{ signals} ) \times \left[ 1 - \strat(s_L) \right]^{i}
    \end{align*}
    The proof is completed by noting that:
    \begin{equation*}
        \frac{
            \Prob ( \textrm{ visited after (i+1)\textsuperscript{st} rejection} ) 
        }{
            \Prob ( \textrm{ visited after i\textsuperscript{th} rejection} ) 
        } =
        \frac{
            \Prob ( \textrm{i+1 } s_L \textrm{ signals} )
        }{
            \Prob ( \textrm{i } s_L \textrm{ signals} ) 
        }
        \times [1 - \strat(s_L) ]
    \end{equation*}
    decreases, and thus $\interimfcn(\strat; \experiment)$ increases, in $\strat (s_L)$.

    \noindent
    \textit{Part ii:}

    \noindent
    Now take $\strat(s_L) = 0$. We then have:
    \begin{align*}
        r_H (\strat; \signalstr) &= 1 - p_H(s_H) \strat(s_H)
        &
        r_L (\strat; \signalstr) &= 1 - p_L(s_H) \strat(s_H)
    \end{align*}
    and:
    \begin{align*}
        \interimfcn(\strat; \signalstr) &\propto \frac{1 + r_H + ... + r_H^{n-1}}{1 + r_L + ... + r_L^{n-1}}
        \\
        &=
        \frac{
            1 - r_H^{n}
        }{
            1 - r_L^{n}
        }
        \times
        \frac{
            1 - r_L
        }{
            1 - r_H
        } = 
        \frac{
            1 - r_H^{n}
        }{
            1 - r_L^{n}
        }
        \times
        \frac{
            p_L(s_H)
        }{
            p_H(s_H)
        }
        \\
        &\propto \frac{
            1 - r_H^{n}
        }{
            1 - r_L^{n}
        } = 
        \frac{
            1 - \left( 1 - p_H(s_H) \strat(s_H) \right)^n
        }{
            1 - \left( 1 - p_L(s_H) \strat(s_H) \right)^n
        }
    \end{align*}
    Differentiating the last expression with respect to $\strat(s_H)$ and rearranging its terms reveals that this derivative is proportional to:
    \begin{equation*}
        \frac{s_H}{1-s_H}
        \times
        \left( 
            \frac{
                 r_H 
            }{
                 r_L 
            }
        \right)^{n-1}
        -
        \frac{
            1 - \left( r_H \right)^n
        }{
            1 - \left( r_L \right)^n
        }
    \end{equation*}
    The positive term is
    the likelihood ratio of one $s_H$ signal and $n-1$ rejections, and the negative term is the likelihood ratio from \textit{at most} $n-1$ rejections. Since acceptances only happen with $s_H$ signals, the negative term strictly exceeds the positive term. This can be verified directly, too:
    \begin{align*}
        \frac{
            1 - (r_H)^n
        }{
            1 - (r_L)^n
        }
        >
        \frac{s_H}{1-s_H}
        \times
        \left(
            \frac{r_H}{r_L}
        \right)^{n-1}
        &\iff
        %\\
        \frac{
            1 - (r_H)^n
        }{
            1 - (r_L)^n
        }
        \times
        \frac{1 - r_L}{1 - r_H}
        >
        \left( \frac{r_H}{r_L} \right)^{n-1}
        \\
        &\iff
        \frac{
            1 + ... + (r_H)^{n-1}
        }{
            1 + ... + (r_L)^{n-1}
        }
        >
        \left( \frac{r_H}{r_L} \right)^{n-1}
    \end{align*}
    The last inequality can be verified easily. Thus, $\interimfcn(\strat; \signalstr)$ decreases in $\strat(s_H)$. 
    
\end{proof}

The Corollary below follows from Lemma \ref{lem:poolbelief_approval}. Let both $\experiment'$ and $\experiment$ are binary experiments, where the former is Blackwell more informative than the latter. If, under both experiments, every buyer accepts upon ``good news'' and rejects upon ``bad news'', the interim belief under $\experiment'$ is lower. 

\begin{cor}
    \label{cor:poolbelief_information}
    Let $\experiment'$ and $\experiment$ be two binary experiments, where the former is Blackwell more informative than the latter. Let the strategies $\strat'$ and $\strat$ for these respective experiments be defined as:
    \begin{align*}
        \strat' (s') &:= 
        \begin{cases}
            0 & s' = s_L' \\
            1 & s' = s_H'
        \end{cases}
        &
        \strat (s) &:= 
        \begin{cases}
            0 & s = s_L \\
            1 & s = s_H
        \end{cases}
    \end{align*}
    Then, $\interimfcn \left( \strat'; \signalstr' \right) \leq \interimfcn \left( \strat; \signalstr \right)$.
    
\end{cor}

\begin{proof}
    Establishing that this holds for a pair $(\experiment', \experiment)$ for which either (i) $s_H' > s_H$ and $s_L = s_L'$, or (ii) $s_H' = s_H$ and $s_L > s_L'$ suffices. I will only prove the first case, the second is analogous. 
    Below I show that the outcome induced by $\strat$ under experiment $\signalstr$ can be replicated by some strategy $\Tilde{\strat}$ under experiment $\signalstr'$, where $\Tilde{\strat} (s_L) > 0$ and $\Tilde{\strat}(s_H) = 1$. Then, the desired conclusion follows from Lemma \ref{lem:poolbelief_approval}. 

    Take the pair $(\strat, \signalstr)$. The probabilities that a buyer accepts or rejects trade, conditional on $\quality$, is given by:
    \begin{align*}
        \frac{
            \Prob_{\strat} \left( \textrm{rejected}  \mid \quality = H \right)
        }{
            \Prob_{\strat} \left( \textrm{rejected}  \mid \quality = L \right)
        } &= \frac{s_L}{1-s_L}
        &
        \frac{
            \Prob_{\strat} \left( \textrm{accepted}  \mid \quality = H \right)
        }{
            \Prob_{\strat} \left( \textrm{accepted}  \mid \quality = L \right)
        }
        &=
        \frac{s_H}{1-s_H}
    \end{align*}
    For the pair $( \Tilde{\strat}, \signalstr' )$ where $\Tilde{\strat}(s_H') = 1$, we have:
    \begin{align*}
        \frac{
            \Prob_{\Tilde{\strat}} \left( \textrm{rejected}  \mid \quality = H \right)
        }{
            \Prob_{\Tilde{\strat}} \left( \textrm{accepted} \mid \quality = L \right)
        } &= \frac{s_L}{1-s_L}
        &
        \frac{
            \Prob_{\Tilde{\strat}} \left( \textrm{accepted} \mid \quality = H \right)
        }{
            \Prob_{\Tilde{\strat}} \left( \textrm{accepted}  \mid \quality = L \right)
        }
        &=
        \frac{p_H'(s_H) + \Tilde{\strat} (s_L) p_H'(s_L)
        }{p_L'(s_H) + \Tilde{\strat} (s_L) p_L'(s_L)}
    \end{align*}
    where $\{ p_L', p_H' \}$ are the distributions for the experiment $\signalstr'$. It is easy to verify that the expression on the right falls from $\frac{s_H'}{1-s_H'}$ to 1 monotonically and continuously as $\Tilde{\strat} (s_L)$ rises from 0 to 1. Thus, there is a unique interior value of $\Tilde{\strat} (s_L)$ that replicates the outcome of $(\strat; \signalstr)$. 
    
\end{proof}

\begin{lem}
    \label{lem:binary_nolowmixing}
    Let $\signalstr$ be a binary experiment, with outcomes in $\signalset = \{s_L, s_H\}$; $s_L \leq s_H$. Let $\strat^{\star}$ be buyers' equilibrium strategy, either in the most or least selective equilibrium. Then, $\strat^{\star} (s_L) \in \{0,1\}$.
\end{lem}

\begin{proof}
    
    I start by proving this for the least selective equilibrium; i.e. $\check{\strat^*}(s_L) \in \{ 0,1 \}$.  For $s_L^{\textrm{mute}}$ defined in Definition \ref{defn:mute}, observe that when $s_L \geq s_L^{\textrm{mute}}$, $\strat(s_L) = \strat(s_H) = 1$ is an equilibrium; so we must have $\check{\strat^*}(s_L) = 1$. The strategy $\strat$ defined by $\strat(s_H) = \strat(s_L) = 1$ gives rise to the interim belief $\interimfcn (\strat; \experiment) = \prior$, which in turn renders approving upon the outcome $s_L$ optimal. In turn, if $s_L < s_L^{\textrm{mute}}$, we must have $\strat^* (s_L) = 1$ for any equilibrium strategy; since the equilibrium interim belief always lies below the prior belief (Proposition \ref{prop:eqmexist}).
    
    Now consider the most selective equilibrium strategy; $\hat{\strat^*}$.
    For contradiction, let $1 > \hat{\strat^*}(s_L) > 0$ and $\hat{\strat^*}(s_H) = 1$. Lemma \ref{lem:poolbelief_approval} establishes that the interim belief falls as $\strat(s_L)$ falls; which implies there must be another, more selective equilibrium strategy $\strat^*$ such that $\strat^* (s_L) = 0$ and $\strat^* (s_H) = 1$.
    
\end{proof}

\noindent
Lemma \ref{lem:binary_nolowmixing} establishes that
when their experiment $\experiment$ is binary,
buyers \textit{never} mix upon seeing ``bad news'', $s = s_L$, neither in the most nor the least selective equilibrium. Following up, Lemma \ref{lem:binary_eqmregions} establishes that a more informative binary experiment pushes buyers to reject upon bad news in both equilibria. 

\begin{lem}
    \label{lem:binary_eqmregions}
    Let $\signalstr$ be a binary experiment, with outcomes in $\signalset = \{ s_L, s_H \}$; $s_L \leq s_H$. Buyers' acceptance probabilities upon ``bad news'', $s = s_L$, in the least and most selective equilibrium strategies are given by:
    \begin{align*}
        \check{\strat^*}(s_L) &=
            \begin{cases}
                1 & s_L \geq s_L^{\textrm{mute}} \\
                0 & s_L < s_L^{\textrm{mute}}
            \end{cases}
        &
        \hat{\strat^{*}} (s_L) &=
            \begin{cases}
                1 & s_L < s_L^{\dag} (s_H) \\
                0 & s_L \geq s_L^{\dag} (s_H)
            \end{cases}
    \end{align*}
    where $s_L^{\dag} (.)$ is an increasing function of $s_H$, and $s_L^{\dag} (s_H) \geq s_L^{\textrm{safe}}$.
\end{lem}

\begin{proof}
    Note that there exists an equilibrium where $\strat(s_L) = 1$ if and only if:
    \begin{equation*}
        \frac{\prior}{1-\prior} \times \frac{s_L}{1-s_L} \geq \frac{c}{1-c}
    \end{equation*}
    which, combined with Lemma \ref{lem:binary_nolowmixing}, proves the part of the Lemma for the selective equilibrium. 

    Now, define the strategies $\strat_0$ as $\strat_1$ as:
    \begin{align*}
        \strat_0(s)
        &=
        \begin{cases}
            0 & s = s_L \\
            0 & s = s_H
        \end{cases}
        &
        \strat_1(s)
        &=
        \begin{cases}
            0 & s = s_L \\
            1 & s = s_H
        \end{cases}
    \end{align*}
    A necessary and sufficient condition for an equilibrium $\strat^*$ where $\strat^*(s_L) = 0$ to exist is:
    \begin{equation*}
        \frac{
            \interimfcn (\strat_{1} ; \signalstr)
        }{
            1 - \interimfcn (\strat_{1} ; \signalstr)
        } \times
        \frac{s_L}{1 - s_L} \leq \frac{c}{1-c}
    \end{equation*}
    Sufficiency follows since either:
    \begin{equation*}
        \frac{
            \interimfcn (\strat_{0} ; \signalstr)
        }{
            1 - \interimfcn (\strat_{0} ; \signalstr)
        } \times
        \frac{s_H}{1 - s_H} \leq \frac{c}{1-c}
    \end{equation*}
    which implies $\strat_0$ is an equilibrium, or there is an equilibrium strategy $\strat^*$ such that $\strat^* (s_L) = 0$ and $\strat^* (s_H) > 0$ by Lemma \ref{lem:poolbelief_approval}. The condition is necessary, since any strategy that is less selective than $\strat_1$ induces a higher interim belief, by Lemma \ref{lem:poolbelief_approval}. 

    By Corollary \ref{cor:poolbelief_information}, whenever this necessary and sufficient condition holds for an experiment $\experiment$, it also holds for a (Blackwell) more informative experiment $\experiment'$. Moreover, whenever the low signals are rejected in the least selective equilibrium, they must be in the most selective equilibrium.
    This concludes the proof. 
    
\end{proof}

\begin{proof}[Proof, Theorem \ref{thm:intensive_binary}:]
    By Lemma \ref{lem:binary_eqmregions}, Blackwell improving buyers' experiment shifts
    both their least selective and most selective equilibrium strategies once from \textit{always} accepting trade to rejecting upon the low signal. By Lemma \ref{lem:overapprove}, this shift in buyers' strategy increases efficiency---and therefore each buyer's expected surplus. 

    Let $\left\{ \strat_{\alpha} \right\}_{\alpha \in [0,1]}$ be the family of strategies where buyers reject upon the low signal:
    \begin{equation*}
        \strat_{\alpha} (s) = 
        \begin{cases}
            0 & s = s_L \\
            1 & s = s_H
        \end{cases}
    \end{equation*}

    By Lemma \ref{lem:poolbelief_approval}, the interim belief $\interim_{\alpha}$ that the strategy $\strat_{\alpha}$ induces is strictly decreasing in $\alpha$. Thus, at most one of these can be an equilibrium strategy for a given experiment. Furthermore, whenever buyers' expected surplus  from $\strat_1$ is weakly positive, this must be the equilibrium strategy; decreasing $\alpha$ can only make approving upon the high signal \textit{more} profitable. Hence, whenever buyers reject upon the low signal in equilibrium, efficiency is given by: $\sumpayoff ( \strat^*; \experiment ) = \max \left\{ 0 , \sumpayoff (\strat_1; \experiment ) \right\}$. 
    The Theorem then follows from the Claim below:

    \begin{claim*}
        $\max \left\{ 0, \sumpayoff ( \strat_{1}; \signalstr ) \right\}$ is:
        \begin{enumerate}[label=\roman*]
            \item weakly increasing in $s_H$ whenever there is some equilibrium strategy $\strat^*$ s.t. $\strat^* (s_L) = 0$.
            \item hump-shaped in $s_L$. As $s_L$ falls, it: 
                \begin{itemize}
                    \item weakly increases when $s_L \geq s_L^{as}$,
                    \item weakly decreases when $s_L \leq s_L^{as}$
                \end{itemize}
        \end{enumerate}
        where $s_L^{as}$ is defined implicitly as:
        \begin{equation*}
        \frac{
            \prior
        }{
            1 - \prior
        }
        \times
        \left( 
            \frac{
                s_L^{\textrm{as}}
            }{
                1 - s_L^{\textrm{as}}
            }
        \right)^{n-1}
        \times
        \frac{s_H}{1 - s_H} = \frac{c}{1-c}
    \end{equation*}
    \end{claim*}

    \begin{proof}[Proof of the Claim.]\leavevmode\\

    \vspace{-20pt}

    \par\noindent
    \hspace*{0cm}%
    \textbf{Part i.} Increasing the strength of good news; i.e. $s_H$.
    \par\vspace{\medskipamount}

    \noindent
    Let $\signalstr$ and $\signalstr'$ be two binary experiments with outcome sets $\signalset = \{ s_L, s_H \}$ and $\signalset' = \{ s_L', s_H' \}$.
    The experiment $\experiment'$ carries \textit{marginally stronger good news} than experiment $\experiment$:
    \begin{align*}
        s_L' &= s_L & s_H' &= s_H + \delta
    \end{align*}
    for some small $\delta$ such that $1 - s_H \geq \delta > 0$. I show that $\sumpayoff ( \strat_{1}' ; \experiment') > \sumpayoff ( \strat_{1}; \experiment )$; where $\strat_1'$ is defined analogously to $\strat_1$ for experiment $\experiment'$.

    \par\noindent
    \hspace*{0cm}%
    \textbf{Step 1.} Replicating $\signalstr'$ with a signal pair $(\typesig, \hat{\typesig})$.
    \par\vspace{\medskipamount}

    Rather than observing the outcome of experiment $\signalstr'$,
    say a buyer initially observes her original signal $\typesig$, and then potentially an additional auxilliary signal $\hat{\typesig}$. The first signal she receives, $s$, records the outcome of $\experiment$. If the low outcome $s_L$ materialises, the buyer observes no more information. 
    If, however, the high outcome $s_H$ materialises, she then observes the additional auxiliary signal $\hat{s}$. This auxiliary signal records the outcome of \textit{another} binary experiment, $\hat{\experiment}$. The outcome of $\hat{\experiment}$ is independent both from $\typesig$ and anything else any other buyer observes. Conditional on the asset's quality $\quality$, the distribution over its outcomes is given by the pmf $p_{\quality}(.)$:
    \begin{align*}
        \hat{p}_{H} (\hat{s_H})
        &= 1 - \varepsilon \times
        \frac{s_L}{1-s_L}
        &
        \hat{p}_{L} (\hat{s_H})
        &= 1 - \varepsilon \times \frac{s_H}{1-s_H}
    \end{align*}
    The evolution of the buyer's beliefs when she observes this signal pair is determined by the two likelihood ratios:
    \begin{align}
        \dfrac{
            \Prob ( 
            (\typesig, \hat{\typesig}) = (s_H, \hat{s_H})
            %\typesig = s_H, \hat{\typesig} = \hat{s_H} 
            \mid \quality = H )
        }{
            \Prob ( 
            (\typesig, \hat{\typesig}) = (s_H, \hat{s_H})
            %\typesig = S, \hat{\typesig} = \hat{s_H} 
            \mid \quality = L )
        }
        &=
        \dfrac{s_H}{1-s_H} \times 
        \dfrac{
            1 - \varepsilon \times \dfrac{s_L}{1-s_L}
        }{
            1 - \varepsilon \times \dfrac{s_H}{1-s_H}
        }
        \label{eqn:likelihood_high}
        \\
         \dfrac{
            \Prob ( (\typesig, \hat{\typesig}) = (s_H, \hat{s_L})
            \mid \quality = H )
        }{
            \Prob ( 
            (\typesig, \hat{\typesig}) = (s_H, \hat{s_L})
            \mid \quality = L )
        }
        &=
        \dfrac{s_L}{1-s_L}
        \label{eqn:likelihood_2_high}
    \end{align}
    Note that the likelihood ratio \ref{eqn:likelihood_high} increases continuously with $\varepsilon$. 
    
    The information from observing the pair $(s, \hat{s})$ as such is equivalent to observing the outcome of experiment $\experiment'$, when:
    \begin{equation}
    \label{eqn:epsiloneqvhigh}
        \dfrac{s_H}{1-s_H} \times 
        \dfrac{
            1 - \varepsilon \times \dfrac{s_L}{1-s_L}
        }{
            1 - \varepsilon \times \dfrac{s_H}{1-s_H}
        }
        =
        \frac{s_H + \delta}{1 - (s_H + \delta)}
    \end{equation}
    for our chosen $(\delta, \varepsilon)$. I choose $\varepsilon$ to satisfy this equality for our $\delta$. As such, $\varepsilon$ becomes a continuously increasing function of $\delta$.
    Furthermore, note that by varying $\varepsilon$ between 0 and $\frac{1 - s_H}{s_H}$, we can replicate \textit{any} experiment $\signalstr'$ with $s_L' = s_L$ and $1 \geq s_H' \geq s_H$.

    \par\vspace{\medskipamount}
    \par\noindent
    \hspace*{0cm}%
    \textbf{Step 2.} $\payoff (\strat_{1}'; \signalstr' ) \geq \payoff (\strat_{1}; \signalstr ) $.
    \par\vspace{\medskipamount}

    The buyer who observes the signal pair $(s, \hat{s})$ obtains equivalent information to that from $\experiment'$. We now must identify the strategy $\Tilde{\strat}: \{ s_L, (s_H, \hat{s_H}), (s_H, \hat{s_L}) \} \to [0,1]$ for this signal pair that replicates the outcome of the strategy $\strat_1'$ for experiment $\experiment'$. This strategy is defined as:
    \begin{align*}
        \Tilde{\strat} ( s_H, \hat{s_H} ) &= 1
        &
        \Tilde{\strat} (s_L) = \Tilde{\strat} ( s_H, \hat{s_L} ) &= 0
    \end{align*}
    and replicates the likelihood ratios of an acceptance and rejection signal under $\experiment'$.

    Now, fix the seller's \textit{signal profile} $\mathbf{s} = \{ (\typesig^i, \hat{\typesig}^i) \}_{i=1}^{n}$ (defined in Section \ref{section:useful_definitions}). I call a seller a \textit{marginal admit} if his score profile is such that:
    \begin{enumerate}[label=\roman*]
        \item for at least one $i \in \{ 1,2,...,n \}$, $\typesig^i = \typesig_H$, and
        \item for \textit{every} $i \in \{ 1,2,...,n \}$, either $\typesig^i = \typesig_L$, or $\hat{\typesig}^i = \hat{\typesig}_L$.
    \end{enumerate}
    These marginal admits drive the wedge between efficiency under $\experiment'$ and $\signalstr$: while some buyer trades under $\experiment$, they \textit{all} reject him under $\hat{\experiment}$. So:
    \begin{equation*}
        \sumpayoff (\strat_{1}'; \signalstr') - \sumpayoff (\strat_{1}; \signalstr)
        =
        \Prob \left( \textrm{marginal admit} \right)
        \times
        \underbrace{
            \left[ 
                c -
            \Prob \left( \quality = H \mid \textrm{marginal admit} \right)
        \right]
        }_{(1)}
    \end{equation*}
    A marginal admit only has signal realisations $(\typesig, \hat{\typesig}) = (s_H, \hat{s_L})$ or $\typesig = s_L$. These carry equivalent information about $\quality$. Thus, the expression (1) above equals:
    \begin{equation*}
        c - \Prob \left[ \quality = H \mid \typesig^1 = ... = \typesig^n = s_L \right]
    \end{equation*}
    In the relevant region where there is an equilibrium strategy that leads to rejections after the low outcome $s_L$, the expression above must be weakly positive:
    \begin{align*}
        c - \Prob \left[ \quality = H \mid \typesig^1 = ... = \typesig^n = s_L \right]
        &\propto
        \frac{c}{1-c} \times \frac{\prior}{1 - \prior} \times \left( 
            \frac{s_L}{1 - s_L}
        \right)^n
        \\
        &\leq 
        \frac{c}{1-c} - \frac{\prior}{1 - \prior} \times 
        \frac{\sum\limits_{k=0}^{n-1} p_H(s_L)^k}{\sum\limits_{k=0}^{n-1} p_L(s_L)^k} \times \frac{s_L}{1 - s_L}
        \\
        &= \frac{c}{1-c} - \frac{\interimfcn (\strat_1; \experiment)}{1 - \interimfcn (\strat_1; \experiment)} \times \frac{s_L}{1 - s_L} \leq 0
    \end{align*}
    where the last inequality follows from the necessary and sufficient condition the proof of Lemma \ref{lem:binary_eqmregions} introduced for such an equilibrium to exist. 

    \par\vspace{\medskipamount}
    \par\noindent
    \hspace*{0cm}%
    \textbf{Part ii.} Increasing the strength of bad news; i.e. decreasing $s_L$.
    \par\vspace{\medskipamount}

    Now, let the experiment $\experiment'$ carry \textit{marginally stronger bad news} than experiment $\experiment$ instead; for some arbitrarily small $\delta \in [0, s_L]$:
    \begin{align*}
        s_L' &= s_L - \delta & s_H' &= s_H
    \end{align*}
    Where $\strat_1'$ and $\strat_1$ are defined as before, I show that:
    \begin{enumerate}[label=\roman*]
        \item $\sumpayoff ( \strat'_{1}; \experiment) - \sumpayoff (\strat_1; \experiment) \geq 0$ when $s_L \geq s_L^{as}$, and
        \item $\sumpayoff ( \strat'_{1}; \experiment) - \sumpayoff (\strat_1; \experiment) \leq 0$ when $s_L \leq s_L^{as}$
    \end{enumerate}

    \noindent

    \par\noindent
    \hspace*{0cm}%
    \textbf{Step 1.} Replicating $\experiment'$ with a signal pair $(s, \hat{s})$.
    \par\vspace{\medskipamount}
    
    As before, let each buyer observe \textit{two} signals, potentially: $\typesig$ and $\hat{\typesig}$. She first observes $\typesig$, which records the outcome of $\experiment$. If the high outcome $\typesig_H$ materialises, she receives no further information. If, however, the low outcome $\typesig_L$ materialises, she then observes the additional auxiliary signal $\hat{\typesig}$, which records the outcome of \textit{another} binary experiment, $\hat{\experiment}$. As before, the outcome of this experiment is independent both from $\typesig$ and anything observed by any other buyer. Its distribution conditional on the asset's quality $\quality$ is given by the pmf $p_{\quality} (.)$:
    \begin{align*}
        \hat{p}_H \left( \hat{\typesig}_H \right)
        &=
        \varepsilon \times \frac{s_H}{1 - s_H}
        &
        \hat{p}_L \left( \hat{\typesig}_H \right) &= \varepsilon \times \frac{s_L}{1 - s_L}
    \end{align*}
    The evolution of the buyer's beliefs upon seeing the signal pair $(\typesig, \hat{\typesig})$ is then determined by the two likelihood ratios:
    \begin{align}
        \dfrac{
            \Prob \left( (\typesig, \hat{\typesig}) = ( s_L, \hat{s_H}) \mid \quality = H \right)
        }{
            \Prob \left( (\typesig, \hat{\typesig}) = ( s_L, \hat{s_H}) \mid \quality = H \right)
        }
        &=
        \dfrac{s_H}{1 - s_H}
        \\
        \dfrac{
            \Prob \left( (\typesig, \hat{\typesig}) = ( s_L, \hat{s_L}) \mid \quality = H \right)
        }{
            \Prob \left( (\typesig, \hat{\typesig}) = ( s_L, \hat{s_L}) \mid \quality = H \right)
        }
        &=
        \dfrac{
            s_L
        }{
            1 - s_L
        }
        \times
        \dfrac{
            1 - \varepsilon \times \dfrac{s_H}{1 - s_H}
        }{
            1 - \varepsilon \times \dfrac{s_L}{1 - s_L}
        }
        \label{eqn:likelihoodlow}
    \end{align}
    Note that \ref{eqn:likelihoodlow} is continuously and strictly decreasing with $\varepsilon$, taking values between $\frac{s_L}{1 - s_L}$ and 0 as $\varepsilon$ varies between 0 and $\frac{s_H}{1 - s_H}$. The signal pair $(\typesig, \hat{\typesig})$ is informationally equivalent to $\signalstr'$ when:
    \begin{equation*}
        \dfrac{
            s_L
        }{
            1 - s_L
        }
        \times
        \dfrac{
            1 - \varepsilon \times \dfrac{s_H}{1 - s_H}
        }{
            1 - \varepsilon \times \dfrac{s_L}{1 - s_L}
        }
        =
        \frac{s_L - \delta}{1 - \left( s_L - \delta \right)}
    \end{equation*}
    I choose $\varepsilon$ to satisfy this equality. As before, $\varepsilon$ then becomes a continuously increasing function of $\delta$. 

    \par\noindent
    \hspace*{0cm}%
    \textbf{Step 2.} $\sumpayoff ( \strat_{1}'; \experiment' ) - \sumpayoff ( \strat_{1}; \experiment )
    \hspace{0.2cm}
    \begin{cases}
        \geq 0 & s_L \geq s_L^{\textrm{as}} \\
        \leq 0 & s_L \leq s_L^{\textrm{as}}
    \end{cases}
    $
    \par\vspace{\medskipamount}

    The buyer who observes the signal pair $(s, \hat{s})$ obtains equivalent information to that from $\experiment'$. We now must identify the strategy $\Tilde{\strat}: \{ (s_L, \hat{s_H}), (s_L, \hat{s_L}) , s_H \} \to [0,1]$ for this signal pair that replicates the outcome of the strategy $\strat_1'$ for experiment $\experiment'$. This strategy is defined as:
    \begin{align*}
        \Tilde{\strat} ( s_L, \hat{s_H} ) = \Tilde{\strat} ( s_H ) &= 1
        &
        \Tilde{\strat} (s_L, \hat{s_L}) &= 0
    \end{align*}
    and replicates the likelihood ratios of an approval and rejection signal under $\experiment'$.

    Now, fix the seller's \textit{score profile}: $\mathbf{s} = \{ (\typesig^i, \hat{\typesig}^i) \}_{i=1}^{n}$. I call a seller a \textit{marginal reject} if:
    \begin{enumerate}[label=\roman*]
        \item for every $i \in \{ 1,2,...,n \}$, $\typesig^i = \typesig_L$, and
        \item for at least one $i \in \{ 1,2,...,n \}$, $\hat{\typesig}^i = \hat{\typesig}_H$.
    \end{enumerate}
    Marginal rejects drive the wedge between efficiency under $\experiment'$ and $\experiment$: while \textit{no} buyer trades under $\experiment$, \textit{at least one} buyer does under $\experiment'$. So:
    \begin{equation*}
        \sumpayoff (\strat_{1}'; \signalstr') - \sumpayoff (\strat_{1}; \signalstr)
        =
        \Prob \left( \textrm{marginal reject} \right)
        \times
        \underbrace{
            \left[ 
            \Prob \left( \quality = H \mid \textrm{marginal reject} 
            \right)
            -
            c
        \right]
        }_{(2)}
    \end{equation*}
    For a marginal reject, buyers observe either $(\typesig^i, \hat{\typesig}^i) = (\typesig_L, \hat{\typesig}_L)$, or $(\typesig^i, \hat{\typesig}^i) = (\typesig_L, \hat{\typesig}_H)$. Denote the number of buyers who observed the latter as $\#$. Since the seller is a marginal reject, $\# \geq 1$. Then, (2) equals:
    \begin{equation*}
        \sum\limits_{i=1}^{n}
        \underbrace{
            \frac{
                \Prob \left( i \textrm{ } \hat{s_H} \textrm{ signals } \mid \typesig^1 = ... = \typesig^n = s_L \right)  
            }{
                \sum\limits_{j=1}^{n} 
                \Prob \left( j \textrm{ } \hat{s_H} \textrm{ signals } \mid \typesig^1 = ... = \typesig^n = s_L \right)
            }
        }_{(3)}
        \times 
        \Prob \left( \quality = H \mid \# = i \right) - c
    \end{equation*}
    where:
    \begin{align*}
        \Prob \left( i \textrm{ } \hat{s_H} \textrm{ signals } \mid \typesig^1 = ... = \typesig^n = s_L \right)
        =
        k \times &\binom{n}{i} \times
        \left( \frac{s_H}{1 - s_H} \times \varepsilon \right)^i \times 
        \left(
            1 - \frac{s_H}{1 - s_H} \times \varepsilon
        \right)^{n-i}
        \\
        +
        (1 - k) \times &\binom{n}{i} \times
        \left( \frac{s_L}{1 - s_L} \times \varepsilon \right)^i \times 
        \left(
            1 - \frac{s_L}{1 - s_L} \times \varepsilon
        \right)^{n-i}
    \end{align*}
    and $k = \Prob \left( \quality = H \mid \typesig^1 = ... = \typesig^n = s_L \right)$. Thus, the limit of expression (3) as $\varepsilon \to 0$ (and therefore, $\delta \to 0$) for any $i > 1$ is:
    \begin{equation}
        \label{eqn:thmbinary_limit}
        \lim\limits_{\varepsilon \to 0}
            \frac{
                \frac{1}{\varepsilon} \times
                \Prob \left( i \textrm{ } \hat{\typesig} = \hat{s_H} \textrm{ signals } \mid \typesig^1 = ... = \typesig^n = s_L \right)  
            }{
                \frac{1}{\varepsilon} \times
                \sum\limits_{j=1}^{n} 
                \Prob \left( j \textrm{ } \hat{\typesig} = \hat{s_H} \textrm{ signals } \mid \typesig^1 = ... = \typesig^n = s_L \right)
            }
            = 0
    \end{equation}
    Therefore, we get:
    \begin{align*}
        &\lim\limits_{\varepsilon \to 0}
        \sum\limits_{i=1}^{n}
            \frac{
                \Prob \left( i \textrm{ } \hat{s_H} \textrm{ signals } \mid \typesig^1 = ... = \typesig^n = s_L \right)  
            }{
                \sum\limits_{j=1}^{n} 
                \Prob \left( j \textrm{ } \hat{s_H} \textrm{ signals } \mid \typesig^1 = ... = \typesig^n = s_L \right)
            }
        \times 
        \Prob \left( \quality = H \mid \# = i \right) - c
        \\
        &= 
        \Prob \left( \quality = H \mid \# = 1 \right)
        - c
        \\
        &\propto
        \frac{\prior}{1 - \prior}
        \times 
        \left( 
            \frac{
                s_L
            }{
                1 - s_L
            }
        \right)^{n-1}
        \times
        \frac{s_H}{1 - s_H} - \frac{c}{1-c}
    \end{align*}
    proving the claim.
    \end{proof}

\end{proof}

\propthresholdbinary*

\begin{proof}\leavevmode\\

    \vspace{-25pt}

    \begin{enumerate}[label=\roman*]
        \item The least selective equilibrium:
    \end{enumerate}

    By Lemma \ref{lem:binary_nolowmixing}, the probability that the seller trades upon the low outcome in
    the least selective equilibrium is:
    \begin{equation*}
        \check{\strat}^* (s_L)
        =
        \begin{cases}
            1 & s_L \geq s_L^{\textrm{mute}} \\
            0 & s_L < s_L^{\textrm{mute}}
        \end{cases}
    \end{equation*}
    Thus, efficiency equals (i) the expected surplus from always approving the applicant when $s_L \geq s_L^{\textrm{mute}}$, and (ii) $\max \left\{ 0, \sumpayoff \left( \strat_1; \experiment \right) \right\}$ when $s_L < s_L^{\textrm{mute}}$ (established in the proof of Theorem \ref{thm:intensive_binary}):
    \begin{equation*}
        \sumpayoff (\check{\strat}^*; \experiment) = 
        \begin{cases}
            \prior - c & s_L \geq s_L^{\textrm{mute}} \\ 
            \max \left\{ 0, \sumpayoff \left( \strat_1; \experiment \right) \right\}
            & s_L < s_L^{\textrm{mute}}
        \end{cases}
    \end{equation*}

    Since always trading is always feasible, we have $\max \left\{ 0, \sumpayoff \left( \strat_1; \experiment \right) \right\} \geq \prior - c$ when $s_L < s_L^{\textrm{mute}}$ by Lemma \ref{lem:overapprove}. Furthermore, the final Claim in Theorem \ref{thm:intensive_binary}'s proof establishes that as $s_L$ falls, the expression $\max \left\{ 0, \sumpayoff \left( \strat_1; \experiment \right) \right\}$ weakly increases (decreases) when $s_L \geq s_L^{\textrm{as}}$ ($s_L \leq s_L^{\textrm{as}}$).
    Thus the desired conclusion is established. 

    \begin{enumerate}[label=\roman*]
        \setcounter{enumi}{1}
        \item The most selective equilibrium:
    \end{enumerate}

    \noindent
    By Lemma \ref{lem:binary_eqmregions}, the most selective equilibrium shifts from one where a buyer always trades to one where she rejects upon the low signal when $s_H \geq s_H^{\dag} (s_L)$, where $s_H^{\dag} (.)$ is an increasing function of $s_L$. Following the arguments made for the least selective equilibrium then, efficiency:
    \begin{itemize}
        \item weakly increases as $s_L$ decreases, when $s_L \geq \min \left\{ s_L^{\textrm{as}}, s_L^{\dag} (s_H) \right\}$
        \item weakly decreases as $s_L$ decreases, when $s_L \leq \min \left\{ s_L^{\textrm{as}}, s_L^{\dag} (s_H) \right\}$.
    \end{itemize}
    The desired result follows by noting that $s_L^{\dag}(s_H) \geq s_L^{\textrm{safe}}$, and therefore 
    $\min \left\{ s_L^{\dag}, s_L^{\textrm{as}} (s_H) \right\}
    \geq
    \min \left\{ s_L^{\textrm{safe}}, s_L^{\textrm{as}} (s_H) \right\}$.
    
\end{proof}

\thmintensive*

\begin{proof}
    The Theorem focuses either on the least, or the most selective equilibrium strategies under both experiments. In the discussion below, I let $\strat^*$ and $\strat^{*'}$ denote whichever equilibria we are focusing on under the respective experiments $\experiment$ and $\experiment'$. When I need to distinguish between the least and most selective equilibria, I denote them as $ (\check{\strat}, \check{\strat}')$ and $ (\hat{\strat}, \hat{\strat}')$, respectively.
    Following the notation introduced in Definition \ref{defn:local_mps}, let $\signalset \cup \signalset' = \left\{ s_1, s_2, ..., s_M \right\}$ be the joint support of the experiments $\signalstr$ and $\signalstr'$, with elements increasing in their indices as usual. 
    Since $\experiment'$ is obtained by a \textit{local} mean preserving spread of $\experiment$, there is a monotone strategy $\strat': \signalset' \to [0,1]$ whose outcome under $\experiment'$ replicates the outcome of $\strat^*$ under $\experiment$:
    \begin{equation*}
        {\strat}'(s) = 
            \begin{cases}
                \strat^*(s_j) & s \in \{ s_{j-1}, s_{j+1} \} \\
                \strat^*(s) & s \notin \{ s_{j-1}, s_{j+1} \}
            \end{cases}
    \end{equation*}

    \noindent
    \textbf{Claim 1.} Efficiency under the most (least) selective equilibrium of $\experiment'$ weakly exceeds that under $\experiment$ when $\hat{\strat} (s_j) = 1$ ($\check{\strat} (s_j) = 1$).

    Now suppose $s_j$ leads to trade under $\strat^*$; $\strat^*(s_j) = 1$. Therefore, $\strat'(s_{j-1}) = {\strat}' (s_{j+1}) = 1$. Below, I show that $\strat^{*'}$ is \textit{more selective than} $\strat'$. By Lemma \ref{lem:selectivebetter}, it follows that $\sumpayoff \left( \strat^{*'}; \experiment' \right) \geq \sumpayoff \left( \strat'; \experiment' \right) = \sumpayoff \left( \strat; \experiment \right)$. 
    
    If $s_{j-1} = \min \signalset \cup \signalset'$ or $\strat^{*'} (s_{j-2}) = 0$, $\strat^{*'}$ must necessarily be more selective than $\strat'$; and we are done. So, for contradiction, I assume the following:
    \begin{itemize}
        \item $s_{j-1} > \min \signalset \cup \signalset'$
        \item $\strat^{*'} (s_{j-2}) > 0$
        \item $\strat^{*'}$ is \textit{less} selective than $\strat'$, where the two strategies differ.
    \end{itemize}

    \par\noindent
    \hspace*{0cm}%
    \textbf{Case i.} $\strat^*$ and $\strat^{*'}$ are the least selective equilibrium strategies; i.e. $\strat^* = \check{\strat}$ and $\strat^{*'} = \check{\strat}'$.
    \par\vspace{\medskipamount}

    I will prove the contradiction by constructing a strategy $\Tilde{\strat}: \signalset \to [0,1]$ for experiment $\experiment$, such that:
    \begin{enumerate}[label=\roman*]
        \item $\Tilde{\strat}$ replicates the outcome $\check{\strat}'$ induces in $\experiment'$,
        \item That $\check{\strat}'$ is an eqm. strategy under $\experiment'$ implies that $\Tilde{\strat}$ is an eqm. strategy under $\experiment$,
        \item But $\Tilde{\strat}$ is less selective than $\check{\strat}$, contradicting that $\check{\strat}$ is the least selective equilibrium strategy under $\experiment$.
    \end{enumerate}

    I define the strategy $\Tilde{\strat}: \signalset \to [0,1]$ for $\signalstr$ as:
    \begin{equation*}
        { \Tilde{ \strat } } ( s ) := 
            \begin{cases}
                1 & s = s_i \\
                \strat' (s) & s \neq s_i
            \end{cases}
    \end{equation*}
    it is seen easily that $\Tilde{\strat}$ replicates the outcome of $\check{\strat}'$. Furthermore, $\check{\strat}'$ is an equilibrium under $\experiment'$ if and only if $\Tilde{\strat}$ is an equilibrium under $\experiment$: they induce the same interim belief $\interim$, and share the following necessary and sufficient condition for optimality:
    \begin{equation*}
        \Prob_{\interim} \left( 
            \quality = H \mid s_{j-2}
        \right)
        \begin{cases}
             = c & \strat'(s_{j-2}) < 1
             \\
             \geq c & \strat'(s_{j-2}) = 1
        \end{cases}
    \end{equation*}
    The strategy
    $\Tilde{\strat}$ under experiment $\experiment$ replicates the outcome of $\check{\strat}'$ under experiment $\experiment'$, and $\strat'$ under $\experiment'$ replicates the outcome of 
    $\check{\strat}$ under experiment $\experiment$. Since we assumed that $\check{\strat}'$ is less selective than $\strat'$, it must be that $\Tilde{\strat}$ is less selective than $\check{\strat}$.

    \par\vspace{\medskipamount}
    \par\noindent
    \hspace*{0cm}%
    \textbf{Case ii.} $\strat^*$ and $\strat^{*'}$ are the most selective equilibrium strategies; i.e. $\strat^* = \hat{\strat}$ and $\strat^{*'} = \hat{\strat}'$.
    \par\vspace{\medskipamount}

    Since strategy $\strat'$ for experiment $\experiment'$ replicates the outcome of $\hat{\strat}$ for experiment $\experiment$, the two strategies induce the same interim belief $\interim$. Therefore, if $\Prob_{\interim} \left( \quality = H \mid s_{j-1} \right) \geq c$, $\strat'$ is an equilibrium under $\experiment'$; meaning $\hat{\strat}'$ must be more selective than $\strat'$. 

    Otherwise, say $\Prob_{\interim} \left( \quality = H \mid s_{j-1} \right) < c$. Then, by Lemma \ref{lem:eqm_algorithm}, there must be an equilibrium strategy that is more selective than $\strat'$ under $\experiment'$.

    \noindent
    \textbf{Claim 2.} Efficiency under the most (least) selective equilibrium of $\experiment'$ falls weakly below that under $\experiment$ if:
    \begin{enumerate}[label= \roman*.]
        \item $s_j$ leads to rejections under $\experiment$; i.e. $\hat{\strat} (s_j) = 0$ ($\check{\strat} (s_j) = 0$), and
        \item the following condition holds:
        \begin{equation*}
            \frac{\prior}{1 - \prior} \times \left( \frac{ r_H ( \strat; \experiment ) }{  r_L ( \strat; \experiment )  } \right)^{n-1} \times \frac{s_{j+1}}{1 - s_{j+1}} \leq \frac{c}{1-c}
        \end{equation*}
    \end{enumerate}

    Now, suppose $s_j$ leads to rejections under $\strat^*$; $\strat^* (s_j) = 0$. Consequently, 
    we have $\strat' ( s_{j-1} ) = \strat' ( s_{j+1} ) = 0$. 
    I establish Claim 2 in two steps:
    \par\vspace{\medskipamount}
    \noindent Step 1: $\strat^{*'}$ is less selective than $\strat'$; trade is likelier when when $s_j$ is locally spread.
    \par\vspace{\medskipamount}
    \noindent Step 2: This efficiency when the condition in Claim 2 is met; $ \sumpayoff ( \strat^{*'}; \experiment' ) \leq \sumpayoff ( \strat' ; \experiment' ) = \sumpayoff (\strat^*; \experiment)$.

    \par\vspace{\medskipamount}
    \noindent
    \textbf{Step 1.}
    \par\vspace{\medskipamount}

    \noindent
    If $s_{j+1} = \max \signalset \cup \signalset'$ or $\strat^{*'}(s_{j+1}) > 0$, it must be the case that $\strat^{*'}$ is less selective than $\strat'$, and we are done. So instead, I assume that $s_{j+1} < \max \signalset \cup \signalset'$ and $\strat^{*'}(s_{j+1}) = 0$.

    \par\vspace{\medskipamount}
    \noindent
    Case i. $\strat^*$ and $\strat^{*'}$ are the least selective equilibrium strategies; i.e. $\strat^* = \check{\strat}$ and $\strat^{*'} = \check{\strat}'$
    \par\vspace{\medskipamount}

    Since $\strat'$ replicates the outcome of $\check{\strat}$, we have $\interimfcn ( \check{\strat}; \experiment ) = \interimfcn ( \strat'; \experiment' ) = \interim$. Thus, $\strat'$ must be an equilibrium strategy under $\experiment'$ if $\Prob_{\interim} \left( \quality = H \mid s_{j+1} \right) \leq c$: the optimality conditions for all signals below $s_{j+1}$ are satisfied \textit{a fortiori}, and those for the signals above $s_{j+1}$ are satisfied 
    since $\check{\strat}$ has the same optimality conditions under $\experiment$. So, $\check{\strat}'$ must be less selective than $\strat'$, since the former is the least selective equilibrium.
    If on the other hand, $\Prob_{\interim} \left( \quality = H \mid 
    s_{j+1} \right) > c$, there must be an equilibrium strategy under experiment $\experiment'$ that is \textit{less} selective than  $\strat'$, by Lemma \ref{lem:eqm_algorithm}.
    
    \par\vspace{\medskipamount}
    \noindent
    Case ii. $\strat^*$ and $\strat^{*'}$ are the most selective equilibrium strategies; i.e. $\strat^* = \hat{\strat}$ and $\strat^{*'} = \hat{\strat}'$.
    \par\vspace{\medskipamount}

    $\hat{\strat}'$ is the most selective equilibrium strategy under experiment $\experiment'$, and we assumed that $\hat{\strat}' (s_{j+1}) = 0$. The strategy $\Tilde{\strat}$ defined below for experiment $\experiment$ replicates the outcome $\hat{\strat}'$ generates under experiment $\experiment'$:
    \begin{equation*}
        \Tilde{\strat} (s) =
        \begin{cases}
            0 & s \leq s_j \\
            \hat{\strat}' (s) & s > s_j
        \end{cases}
    \end{equation*}
    Note that $\Tilde{\strat}$ must be an equilibrium under experiment $\experiment$, since the interim belief it induces is the same as the one $\hat{\strat}'$ does, and its optimality constraints are a subset of the latter's. But since $\hat{\strat}$ is the \textit{most} selective equilibrium strategy under $\experiment$, $\Tilde{\strat}$ must be less selective than it. 

    \par\vspace{\medskipamount}
    \noindent
    \textbf{Step 2.}
    \par\vspace{\medskipamount}

    The statement is trivially true when $\strat' = \strat^{*'}$, so I focus on the case where these two strategies differ. As Step 1 established, $\strat^{*'}$ must be less selective than $\strat'$. This implies that $\strat^{*'} ( s_{j+1} ) > 0$. To see why, say we had $\strat^{*'} (s_{j+1}) = 0$ instead. We can then construct a strategy $\Tilde{\strat}$ for experiment $\experiment$, which replicates the outcome $\strat^{*'}$ generates under experiment $\experiment'$:
    \begin{equation*}
        \Tilde{\strat} (s) =
        \begin{cases}
            0 & s \leq s_j \\
            \strat^{*'} (s) & s > s_j
        \end{cases}
    \end{equation*}
    As they induce the same interim belief and the optimality constraints of the latter are a subset of the former's, $\Tilde{\strat}$ must be an equilibrium under $\experiment$. This contradicts with $\strat^{*}$ and $\strat^{*'}$ being the least selective strategies; since $\strat^{*'}$ being less selective than $\strat'$ implies that $\Tilde{\strat}$ must be less selective than $\strat^*$. It also contradicts with $\strat^*$ and $\strat^{*'}$ being the most selective strategies; since it would imply that $\strat'$, more selective than $\strat^{*'}$, should be an equilibrium under $\experiment'$.

    Given that $\strat^{*'} (s_{j+1}) > 0$, I now take another strategy $\strat^{\delta}_{\experiment'}: \signalset' \to [0,1]$ for experiment $\experiment'$: 
    \begin{align*}
        \strat^{\delta}_{\experiment'} (s) =
        \begin{cases}
            1 & s > s_{j+1} \\
            \delta & s = s_{j+1} \\
            0 & s < s_{j+1}
        \end{cases}
    \end{align*}
    where $\delta > 0$ is small enough so that $\strat^{\delta}_{\experiment'}$ is more selective than $\strat^{*'}$, but less selective than $\strat'$. I will show that, when the condition stated in Claim 2 holds, we have $\sumpayoff \left( \strat^{\delta}_{\experiment'}; \experiment' \right) \leq 
    \sumpayoff \left( \strat'; \experiment' \right)
    $ for $\delta \to 0$. Lemma \ref{lem:overapprove} then implies that $ \sumpayoff \left( \strat^{*'}; \experiment' \right) \leq \sumpayoff \left( \strat^{\delta}_{\experiment'}; \experiment' \right)$, which coins the result. 

    To show this, I construct another experiment $\experiment^{\textrm{re}}$ under which I will use compare two strategies, $\strat_{\textrm{re}}$ and $\strat_{\textrm{re}}^{\delta}$, that replicate the outcomes of the strategies $\strat'$ and $\strat^{\delta}_{\experiment'}$, respectively. The experiment $\experiment^{\textrm{re}}$ has three possible outcomes, $\left\{ s_L^{\textrm{re}}, s_{\delta}^{\textrm{re}}, s_H^{\textrm{re}} \right\}$. Conditional on the applicant's quality $\quality$, its outcome distribution is independent from any other information any evaluator sees, and is given by the following pmf $p_{\quality}^{\textrm{re}}$:
    \begin{align*}
        p_{\quality} (s^{\textrm{re}}) = 
        \begin{cases}
            1 - r_{\quality} ( \strat^*; \experiment ) & s = s^{\textrm{re}}_H \\
            \delta \times p_{\quality}' ( s_{j+1} ) & s = s^{\textrm{re}}_{\delta} \\
            r_{\quality} \left( \strat^*; \experiment \right) - \delta \times p_{\quality}' ( s_{j+1} )
            & s = s^{\textrm{re}}_L
        \end{cases}
    \end{align*}
    Define the strategies $\strat_{\textrm{re}}$ and $\strat_{\textrm{re}}^{\delta}$ for this experiment as follows:
    \begin{align*}
        \strat_{\textrm{re}}(s) &= 
        \begin{cases}
            1 & s = s^{\textrm{re}}_H \\
            0 & s = s^{\textrm{re}}_{\delta} \\
            0 & s = s^{\textrm{re}}_L
        \end{cases}
        &
        \strat_{\textrm{re}}^{\delta} &= 
        \begin{cases}
            1 & s = s^{\textrm{re}}_H \\
            1 & s = s^{\textrm{re}}_{\delta} \\
            0 & s = s^{\textrm{re}}_L
        \end{cases}
    \end{align*}

    Now note that these two strategies replicate the outcomes of the strategies $\strat'$ and $\strat^{\delta}_{\experiment'}$, respectively. Under $\strat_{\textrm{re}}(s)$, the probability that a buyer trades, conditional on the seller's quality, is the same as it is under strategy $\strat'$ (or $\strat^{*}$, which it replicates), and under $\strat_{\textrm{re}}^{\delta}$, it is the same as it is under $\strat^{\delta}_{\experiment'}$. 

    So, the difference between efficiency under these two strategies is determined by the \textit{marginal reject} who:
    \begin{itemize}
        \item is rejected by \textit{every} buyer under the strategy $\strat_{\textrm{re}}$.
        \item is accepted by \textit{at least one} buyer under the strategy $\strat_{\textrm{re}}^{\delta}$.
    \end{itemize}
    Where $\mathbf{s^{\textrm{\textbf{re}}}} = \left\{ s^1, ..., s^n \right\}$ is the seller's signal profile under the experiment $\experiment^{\textrm{re}}$, he has:
    \begin{itemize}
        \item \textit{no} $s_H^{\textrm{re}}$ signals; $s^i \neq s_H^{\textrm{re}}$ for all $i \in \{1,2,...,n\}$ and
        \item \textit{at least one} $s_{\delta}^{\textrm{re}}$ signal; there exists some $i \in \{1,2,...,n\}$ such that $s^i = s_H^{\textrm{re}}$.
    \end{itemize}
    Thus we have:
    \begin{align*}
        \sumpayoff ( \strat^{\delta}_{\experiment'} ; \experiment') - 
        \sumpayoff ( \strat'; \experiment' )
        &=
        \sumpayoff ( \strat^{\delta}_{\textrm{re}} ; \experiment^{\textrm{re}}) - 
        \sumpayoff ( \strat_{\textrm{re}}; \experiment^{\textrm{re}} ) 
        \\
        &=
        \Prob \left( \textrm{marginal reject} \right) \times
        \underbrace{
            \left[ 
            \Prob \left( \quality = H \mid \textrm{marginal reject} 
            \right)
            -
            c
        \right]
        }_{(2)}
    \end{align*}

    The expression labelled (2) above equals:
    \begin{align*}
        \sum\limits_{i=1}^{n}
        \frac{
            \Prob \left( i \textrm{ } s_{\delta}^{\textrm{re}} \textrm{ and } n - i \textrm{ } s_{L}^{\textrm{re}} \textrm{ signals} \right)
        }{
            \sum\limits_{k = 1}^{n} \Prob \left( k \textrm{ } s_{\delta}^{\textrm{re}} \textrm{ and } n - k \textrm{ } s_{L}^{\textrm{re}} \textrm{ signals} \right)
        }   
        \times
        %\\
        \Prob \left( \quality = H \mid i \textrm{ } s_{\delta}^{\textrm{re}} \textrm{ and } n-i \textrm{ }s_{L}^{\textrm{re}} \textrm{ signals} \right) - c
    \end{align*}

    Since the probability that a buyer observes the $s_{\delta}^{\textrm{re}}$ signal is proportional to $\delta$, we have\footnote{See expression \ref{eqn:thmbinary_limit} and the surrounding discussion in the proof of Theorem \ref{thm:intensive_binary} for a more detailed explanation of this.}:
    \begin{equation*}
        \lim\limits_{\delta \to 0}
        \frac{
            \Prob \left( i \textrm{ } s_{\delta}^{\textrm{re}} \textrm{ and } n - i \textrm{ } s_{L}^{\textrm{re}} \textrm{ signals} \right)
        }{
            \sum\limits_{k = 1}^{n} \Prob \left( k \textrm{ } s_{\delta}^{\textrm{re}} \textrm{ and } n - k \textrm{ } s_{L}^{\textrm{re}} \textrm{ signals} \right)
        }
        = 0 
    \end{equation*}
    Therefore, we get:
    \begin{align*}
        &\lim\limits_{\delta \to 0}
        \sum\limits_{i=1}^{n}
        \frac{
            \Prob \left( i \textrm{ } s_{\delta}^{\textrm{re}} \textrm{ and } n - i \textrm{ } s_{L}^{\textrm{re}} \textrm{ signals} \right)
        }{
            \sum\limits_{k = 1}^{n} \Prob \left( k \textrm{ } s_{\delta}^{\textrm{re}} \textrm{ and } n - k \textrm{ } s_{L}^{\textrm{re}} \textrm{ signals} \right)
        }   
        \times
        %\\
        \Prob \left( \quality = H \mid i \textrm{ } s_{\delta}^{\textrm{re}} \textrm{ and } n-i \textrm{ }s_{L}^{\textrm{re}} \textrm{ signals} \right) - c
        \\[0.2cm]
        &\lim\limits_{\delta \to 0} \Prob \left( \quality = H \mid \textrm{one } s_{\delta}^{\textrm{re}} \textrm{ signal and } n - 1 \textrm{  } s_L^{\textrm{re}} \textrm{ signals} \right) - c
        \\[0.2cm]
        \propto 
        &\lim\limits_{\delta \to 0} 
        \frac{\prior}{1 - \prior} \times 
        \frac{ p_H' (s_{j+1}) }{ p_L' (s_{j+1}) } \times 
        \left( 
            \frac{
                r_H (\strat^*; \experiment) - \delta \times p_H' (s_{j+1})
            }{
                r_L (\strat^*; \experiment) - \delta \times p_L' (s_{j+1})
            }
        \right)^{n-1} - \frac{c}{1-c}
        \\[0.2cm]
        =
        &\frac{\prior}{1 - \prior} \times \frac{s_{j+1}}{1 - s_{j+1}} \left( \frac{r_H (\strat^*; \experiment)}{r_L (\strat^*; \experiment)} \right)^{n-1} - \frac{c}{1-c}
    \end{align*}
    
\end{proof}

\propsufficientharm*

\begin{proof}
    First, I let $\hat{\strat} (s_{j+2}) < 1$. I show that this implies $\hat{\strat}$ and $\hat{\strat}'$ induce equivalent outcomes under their respective experiments. The strategy $\strat': \signalset' \to [0,1]$ which replicates the outcome of $\hat{\strat}'$ under experiment $\experiment'$:
    \begin{align*}
        \strat'(s) = 
        \begin{cases}
            \hat{\strat} (s) & s \geq s_{j+2} \\
            0 & s < s_{j+2}
        \end{cases}
    \end{align*}
    must then be an equilibrium strategy under experiment $\experiment'$. This is because these strategies induce the same interim belief, that $\hat{\strat}$ is an equilibrium strategy under $\experiment$ ensures that the optimality conditions of $\strat'$ for signals below $s_{j+2}$ are satisfied, and for signals above $s_{j+2}$, the optimality conditions are the same as those for $\strat'$. This means that $\hat{\strat}'$ must be more selective than $\strat'$. However, when proving Theorem \ref{thm:intensive}, we established that $\strat'$ must be more selective than $\hat{\strat}'$. So it must be that $\strat' = \hat{\strat}'$, and we are done.

    So instead, let $\hat{\strat} (s_{j+2}) = 1$. But then, it is easily established that:
    \begin{equation*}
        \frac{
            r_H ( \hat{\strat}; \experiment )
        }{
            r_L ( \hat{\strat}; \experiment )
        }
        \leq
        \frac{s_{j}}{1 - s_{j}}
    \end{equation*}
    since $r_{\quality} ( \hat{\strat}; \experiment ) = \sum\limits_{k=1}^{j} p_{\quality} (s_k)$. So, the condition Proposition \ref{prop:sufficient_harm} supplies is sufficient for the one Theorem \ref{thm:intensive} does. 
    
\end{proof}

\lemdesignmonotonebinary*

\begin{proof}
    To prove this statement, I 
    take some garbling $\experiment^G$ and an equilibrium $\strat^G: \signalset^G \to [0,1]$ it supports. I then construct a monotone binary garbling $\experiment^{G*}$ which produces IC recommendations, and show that efficiency under $\experiment^{G*}$ and the strategy which obeys its recommendations, $\strat^{G*}$, are higher than those under $\experiment^G$ and $\strat^G$.

    For the monotone binary garbling $\experiment^{G*} = (\signalset^{G}, \mathbf{P}^{G*})$ and the garbling $\experiment^G = (\signalset^G, \mathbf{P}^G)$  in question:
    \begin{align*}
        \mathbf{P} \times \mathbf{T} &= \mathbf{P}^{G}
        &
        \mathbf{P} \times \mathbf{T}^* &= \mathbf{P}^{G*}
    \end{align*}
    define the expressions:
    \begin{align*}
        f^*(s) &:= p_L (s) \times t_{i1}^*
        &
        f (s) &:= p_L (s) \times
        \sum\limits_{s_j^G \in \signalset^G} t_{ij} \times \left( 1 - \strat^G (s_j^G) \right)
    \end{align*}
    for each $s \in \signalset$. Given the asset has Low quality,
    $f^*(s)$ is the probability that
    (i) a buyer would have observed the signal $s \in \signalset$ in her original experiment, \textit{and} (ii) the garbling $\experiment^{G*}$ issues her a ``rejection recommendation''. Similarly, $f(s)$ is the probability that (i) a buyer would observe the signal $\typesig \in \signalset$ in her original experiment, \textit{and} (ii) he would be rejected under the equilibrium strategies $\strat^G$. For this Low quality seller,  $r_L^{G*}$ below is the probability that the buyer receives a rejection recommendation under $\experiment^{G*}$; and $r_L^G$ is the probability that the buyer rejects him under $(\experiment^G, \strat^{G})$:
    \begin{align*}
        r_L^{G*} &:= \sum\limits_{s \in \signalset} f^*(s)
        &
        r_L^{G} &:= 
        \sum\limits_{s \in \signalset} 
        f (s)
    \end{align*}

    Now, take the least selective monotone binary garbling $\experiment^{G*}$ such that $r_L^{G*} = r_L^G$. Evidently, this garbling exists. 
    
    Clearly, one can treat $f^*$ and $f$ as probability density functions over $\signalset$ when normalised. Furthermore, the distribution the former describes is first order stochastically dominated by the one the latter does; $\frac{f^*(s_j)}{ \sum\limits_{s \in \signalset} f^*(s) }$ crosses $\frac{ f(s) }{ \sum\limits_{s \in \signalset} f(s) }$ once from below. Therefore we get:
        \begin{align*}
            r_H^*
            :=
            \sum\limits_{s \in \signalset} \frac{p_H (s)}{p_L (s)} \times \frac{f^*(s)}{\sum\limits_{s \in \signalset} f^*(s)}&
            \\
            &\leq 
            \sum\limits_{s \in \signalset} \frac{p_H (s)}{p_L (s)} \times \frac{f(s)}{\sum\limits_{s \in \signalset} f(s)} =: r_H
        \end{align*}
        where $r_H^*$ and $r_H$ are the probabilities that a High quality seller is rejected from a visit under the strategies $\strat^{G*}$ and $\strat^G$, respectively. 

        Since $r_H^* \geq r_H$ and $r_L^* = r_L$, efficiency is higher under $\strat^*$ than it is under $\strat$. It only remains to show that the strategy $\strat^*$ is optimal against the interim belief $\interim^*$ consistent with it. 

        The interim belief $\interim^*$ consistent with $\experiment^{G*}$ and $\strat^{G*}$ lies below $\interim$---the interim belief consistent with $\experiment^{G}$ and $\strat^{G}$:
        \begin{align*}
            \frac{
                \interim^*
            }{
                1 - \interim^*
            }
            \hspace{0.2cm}
            =
            \hspace{0.2cm}
            \frac{
                \sum\limits_{k=0}^{n-1} \left( r_H^{*} \right)^k
            }{
                \sum\limits_{k=0}^{n-1} \left( r_L^{*} \right)^k
            }
            \hspace{0.2cm}
            =
            \hspace{0.2cm}
            \frac{
                \sum\limits_{k=0}^{n-1} (r_H)^k
            }{
                \sum\limits_{k=0}^{n-1} (r_L)^k
            }
            \hspace{0.2cm}
            \leq
            \hspace{0.2cm}
            \frac{
                \interim
            }{
                1 - \interim
            }
        \end{align*}
        Under the interim belief $\interim^*$,
        it is optimal for a buyer upon the signal $s_L^{G*}$ if and only if:
        \begin{equation*}
            \frac{\interim^*}{1 - \interim^*} \times \frac{r_H^*}{r_L^*} \leq \frac{c}{1-c}
        \end{equation*}
        But this inequality must hold; since $\frac{r_H^*}{r_L^*} \leq \frac{r_H}{r_L}$, $\interim^* \leq \interim$, and $\strat^G$ is optimal against $\interim$:
        \begin{equation*}
            \frac{\interim^*}{1 - \interim^*} \times \frac{r_H^*}{r_L^*} \leq
            \frac{\interim}{1 - \interim} \times \frac{r_H}{r_L} \leq \frac{c}{1-c}
        \end{equation*}

        Furthermore, that $\sumpayoff ( \strat^{G*}; S^{G*} ) \geq \sumpayoff \left( \strat; S^G \right) \geq 0$ suggests that the expected suprlus from accepting a seller upon the ``approve'' recommendation must be weakly positive; hence optimal. Thus, the strategy $\strat^{G*}$ is optimal against $\interim^*$.

\end{proof}

\propgarbling*

\begin{proof}
      \noindent
      \textbf{Step 1:} The following are well-defined:
      \begin{itemize}
          \item the least selective monotone binary garbling under which adverse selection is irrelevant,
          \item the least (most) selective monotone binary garbling under which adverse selection is (not) irrelevant among those which produce IC recommendations. 
      \end{itemize}

      I first show the least selective monotone binary garbling under which adverse selection is irrelevant is well defined. For any monotone binary garbling $\experiment^G = (\signalset^G, \mathbf{P}^G)$, let $\mathbf{P} \times \mathbf{T} = \mathbf{P}^G$ and define $d (\experiment^G) := \sum\limits_{i=1}^{m} t_{i2}$.
      Evidently, $d(.)$ is a bijection between the space of monotone binary garblings of $\experiment$ and $[0,m]$. Also, where both are monotone binary garblings of $\experiment$, $\experiment^G$ is more selective than $\experiment^{G'}$ if and only if $d (\experiment^G) \leq d (\experiment^{G'})$. Thus, we seek the monotone binary garbling $d^{-1} ( D^* )$ where $D^* := \max \left\{ D \in [0,m]: \textrm{ a.s. is irrelevant under }d^{-1} (D)  \right\}$.
      We must only show that $D^*$ is well defined. To that end, define the Real valued function $F$ over the space of monotone binary garblings, where:
      \begin{align*}
          F(S^G) = 
          \begin{cases}
            \frac{\prior}{1 - \prior}
              \frac{p_H(s^*)}{p_L(s^*)} \times \left( 
                \frac{r_H^G}{r_L^G}
              \right)^{n-1} -
              \frac{c}{1-c} & d(\experiment^G) \in (0,m) \\
              \lim\limits_{D \downarrow 0} F \circ d^{-1} ( D ) & d(\experiment^G) = 0 \\
              + \infty & d (\experiment^G) = m
          \end{cases}
          \hspace{0.4cm}
          &
          \hspace{0.4cm}
          r_{\quality}^G := \sum\limits_{s \in \signalset} p_s^G(s_L^G) \times p_{\quality} (s)
      \end{align*}
    where $s^*$ is the threshold signal of this garbling. 
    
    So, we seek $D^* := \max \left\{ D \in [0,m]: F \circ d^{-1} (D) \geq 0 \right\}$.
    But this maximiser exists because the function $F \circ d^{-1}$ is upper semicontinuous: $F \circ d^{-1}$ is a decreasing function, and for any $\Bar{D} \in [0,m]$ and $\varepsilon > 0$, we can find some $\delta_{\varepsilon}$ such that $D \in ( \Bar{D} - \delta_{\varepsilon} , \Bar{D}) \cap [0,m]$ implies $d^{-1} (\Bar{D})$ and $d^{-1} (D)$ have the same threshold signal $s^*$ and thus $F \circ d^{-1} (D) < F \circ d^{-1} ( \Bar{D} ) + \varepsilon$ since $\sfrac{r_H^G}{r_L^G}$ is continuous in $d(S^G)$.

    Now say this garbling does not provide IC recommendations.
    Denote the interim belief that is consistent with evaluators following $\experiment^G$'s recommendations as $\interim^G$. The garbling $\experiment^G$ provides IC recommendations if:
    \begin{align*}
        \underbrace{\frac{r_H^G}{r_L^G} \times \frac{\interim^G}{1 - \interim^G} }_{:= f_1 (d (\experiment^G)) }
        &\leq \frac{c}{1-c}
        &
        \underbrace{\frac{1 - r_H^G}{1 - r_L^G} \times \frac{\interim^G}{1 - \interim^G}}_{:= f_2 ( d (\experiment^G) )}
        &\geq \frac{c}{1-c}
    \end{align*}
    As defined above, both $f_1(.)$ and $f_2(.)$ are continuous. Therefore, the set of monotone binary garblings with optimal recommendations---$\left\{ D \in [0,m] : f_1(D) \leq \frac{c}{1-c} \textrm{ and } f_2(D) \geq \frac{c}{1-c} \right\}$---is compact. Thus, both objects below are well-defined:
    \begin{align*}
        &\max \left\{ D \in [0,m] \textrm{ and } d^{-1}(D) \textrm{ has IC rec.s } : F \circ d^{-1} (D) \geq 0 \right\}
        \\
        &\min \left\{ D \in [0,m] \textrm{ and } d^{-1}(D) \textrm{ has IC rec.s } : F \circ d^{-1} (D) \leq 0 \right\}
    \end{align*}
    Among those with IC recommendations, 
    the former gives us the least selective garbling under which adverse selection is irrelevant. The latter gives us the most selective garbling under which adverse selection is not irrelevant among such garblings, since the least-selective garbling under which adverse selection is irrelevant does \textit{not} have IC recommendations (the minimiser of this set \textit{must} have $F \circ d^{-1} (D) < 0$).

    \par\vspace{\medskipamount}
    \noindent
    \textbf{Step 2:} Proving Proposition \ref{prop:optimalgarbling} as stated. 
    \par\vspace{\medskipamount}

    Efficiency under a monotone binary garbling $S^G$ and strategies $\strat^G$ that obey its recommendations is given by:
    \begin{equation*}
        \sumpayoff \left( \strat^G; S^G \right)
        =
        \prior - c - \prior \times \left( r_H^G \right)^{n} \times ( 1-c ) + ( 1 - \prior ) \times \left( r_L^G \right)^{n} \times c
    \end{equation*}
    As a function of $d^{-1}(.)$, efficiency is continuous and therefore attains its maximum over the set $[0,m]$. I show that this maximum is attained with the least selective garbling under which adverse selection is irrelevant.

    For the garbling $\experiment^G$, define $\experiment^{G}_{+\varepsilon} := d^{-1} ( d (\experiment^G) + \varepsilon )$ and $S^{G}_{-\varepsilon} := d^{-1} ( d (S^G) - \varepsilon )$. 
    Likewise, let $s^*_{+\delta}$ and $s^*_{-\delta}$ be the threshold signals of these experiments, and $r_{\quality;+\delta}^*$, $r_{\quality;-\delta}^*$ be the probability that a seller of quality $\quality$ is rejected in a visit, under each garbling. From our earlier reasoning about the impact of making evaluators strategies marginally more (less) selective, we observe that:
     \begin{equation*}
        \lim\limits_{\delta \to 0} \sumpayoff \left( \strat^G_{+\delta} ; \experiment^{G}_{+\delta} \right) - \sumpayoff \left( \strat^G ; \experiment^{G} \right)
        \propto
        \lim\limits_{\delta \to 0}
        \frac{\prior}{1 - \prior} \times \frac{p_H (\typesig^*_{+\delta})}{p_L (\typesig^*_{+\delta})} \times \left( \frac{r_{H; +\delta}^G}{r_{L; +\delta}^G} \right)^{n-1} - \frac{c}{1-c}
        \leq 0
    \end{equation*}
    where the last inequality follows since $\experiment^G$ is the least selective garbling under which adverse selection is irrelevant. We conclude that giving evaluators a marginally less selective garbling, and therefore (Lemma \ref{lem:overapprove}) any garbling that is less selective than $S^G$, cannot improve their payoffs. Likewise, for a marginally more selective garbling we have:
    \begin{equation*}
        \lim\limits_{\delta \to 0} \sumpayoff \left( \strat^G_{+\delta} ; \experiment^{G}_{-\delta} \right) - \sumpayoff \left( \strat^G ; \experiment^{G} \right)
        \propto
        -
        \lim\limits_{\delta \to 0}
        \frac{\prior}{1 - \prior} \times \frac{p_H (\typesig^*)}{p_L (\typesig^*)} \times \left( \frac{r_{H; +\delta}^G}{r_{L; +\delta}^G} \right)^{n-1} - \frac{c}{1-c}
        \geq 0
    \end{equation*}
    where the term on the RHS is now negative because the probability of trade \textit{decreases} when strategies become more selective. By a reasoning similar to that behind Lemma \ref{lem:overapprove}, this reveals that \textit{no} garbling that is more selective can improve efficiency either. 

    This also proves that among those with optimal recommendations, the least selective garbling under which adverse selection is irrelevant cannot be improved with a more selective garbling and the most selective garbling under which adverse selection is not irrelevant cannot be improved with a less selective garbling. 
    
\end{proof}

\newpage

\printbibliography

\end{document}